\documentclass[11pt]{article} 
\usepackage{amsfonts,amssymb,amsmath,amsthm,dsfont,graphicx,hyperref,mathrsfs,euscript}
\usepackage[top=1in,bottom=1in,left=1in,right=1in]{geometry} 
\usepackage{caption}
\usepackage{xcolor}
\usepackage{verbatim}
\usepackage{booktabs}
\usepackage{multirow}
\usepackage{framed}
\usepackage{forloop}
\usepackage{lscape}
\usepackage[onehalfspacing]{setspace}
\usepackage{tocloft} 
\usepackage{natbib}
\usepackage{ifthen}
\usepackage{subfig}
\usepackage{tikz}
\usepackage{slantsc}
\hypersetup{pdfborder = {0 0 0},colorlinks=true,linkcolor=blue,urlcolor=blue,citecolor=blue}

\usepackage{algorithm,algorithmic}

\newtheoremstyle{exampstyle}
  {} 
  {} 
  {} 
  {} 
  {\bfseries\color{blue}} 
  {.} 
  {.5em} 
  {} 
\theoremstyle{exampstyle}\newtheorem{thm}{Theorem}
\theoremstyle{exampstyle}\newtheorem*{thm*}{Theorem}
\theoremstyle{exampstyle}\newtheorem{defn}{Definition}     
\theoremstyle{exampstyle}     
\theoremstyle{exampstyle}\newtheorem{lem}{Lemma}  
\theoremstyle{exampstyle}\newtheorem{coro}{Corollary}        
\theoremstyle{exampstyle}

\theoremstyle{exampstyle}  
\theoremstyle{exampstyle}\newtheorem{Assumption}{Assumption}
\theoremstyle{exampstyle}

\theoremstyle{definition}

\theoremstyle{exampstyle}
\newtheorem{remarktmp}{Remark}
\newenvironment{remark}
	{ \begin{remarktmp} 	}
	{ 
		\medskip\hfill{\LARGE$\lrcorner$}
		\end{remarktmp} 
	}

\allowdisplaybreaks[1]

\newtheorem{exampletmp}{Example}
\newenvironment{example}
	{ \begin{exampletmp} 	}
	{ 
		\medskip\hfill{\LARGE$\lrcorner$}
		\end{exampletmp} 
	}

\allowdisplaybreaks[1]

\renewcommand{\thethm}{\arabic{thm}}
\newcommand{\Prob}{\mathbb{P}}
\newcommand{\Indicator}{\mathds{1}}

\newcommand{\Trans}{\prime}

\renewcommand{\epsilon}{\varepsilon}

\newcommand{\bC}{\mathbf{C}}

\newcommand{\bR}{\mathbf{R}}

\newcommand{\bz}{\mathbf{z}}

\newcommand{\bmu}{\boldsymbol{\mu}}

\newcommand{\bsigma}{\boldsymbol{\sigma}}

\newcommand{\bOmega}{\boldsymbol{\Omega}}

\newcommand{\Universe}{X}

\newcommand{\dimUniverse}{K}

\newcommand{\Menus}{\mathcal{S}}

\newcommand{\ParSet}{\Theta}

\newcommand{\AttFilter}{\mu}
\newcommand{\bAttFilter}{\boldsymbol{\mu}}

\newcommand{\ChoiProb}{\pi}
\newcommand{\bChoiProb}{\boldsymbol{\pi}}

\newcommand{\Tstat}{\mathscr{T}}

\DeclareSymbolFontAlphabet{\amsmathbb}{AMSb}

\begin{document}

\title{A Random Attention Model\thanks{We thank Ivan Canay, Ignacio Esponda, Rosa Matzkin, Francesca Molinari, Jose Luis Montiel-Olea, Pietro Ortoleva, Kota Saito, Joerg Stoye, and Rocio Titiunik for very helpful comments and suggestions that improved this paper. We also thank the Editor, Emir Kamenica, and four reviewers for their constructive criticism of our paper, which led to substantial improvements. Financial support from the National Science Foundation through grant SES-1628883 is gratefully acknowledged.}\bigskip}
\author{
Matias D. Cattaneo\thanks{Department of Operations Research and Financial Engineering, Princeton University.} 
\and 
Xinwei Ma\thanks{Department of Economics, University of California at San Diego.}
\and
Yusufcan Masatlioglu\thanks{Department of Economics, University of Maryland.}
\and
Elchin Suleymanov\thanks{Department of Economics, Krannert School of Management, Purdue University.}
}\maketitle

\begin{abstract}

This paper illustrates how one can deduce preference from observed choices when attention is not only limited but also random. In contrast to earlier approaches, we introduce a Random Attention Model (RAM) where we abstain from any particular attention formation, and instead consider a large class of nonparametric random attention rules. Our model imposes one intuitive condition, termed \textit{Monotonic Attention}, which captures the idea that each consideration set competes for the decision-maker's attention. We then develop revealed preference theory within RAM and obtain precise testable implications for observable choice probabilities. Based on these theoretical findings, we propose econometric methods for identification, estimation, and inference of the decision maker's preferences. To illustrate the applicability of our results and their concrete empirical content in specific settings, we also develop revealed preference theory and accompanying econometric methods under additional nonparametric assumptions on the consideration set for binary choice problems. Finally, we provide general purpose software implementation of our estimation and inference results, and showcase their performance using simulations.

\end{abstract}


\thispagestyle{empty}
\clearpage

\doublespacing
\setcounter{page}{1}
\pagestyle{plain}

\section{Introduction}

Revealed preference theory is not only a cornerstone of modern economics, but is also the source of important theoretical, methodological and policy implications for many social and behavioral sciences. This theory aims to identify the preferences of a decision maker (e.g., an individual or a firm) from her observed choices (e.g., buying a house or hiring a worker). In its classical formulation, revealed preference theory assumes that the decision maker selects the best available option after full consideration of all possible alternatives presented to her. This assumption leads to specific testable implications based on observed choice patterns but, unfortunately, empirical testing of classical revealed preference theory shows that it is not always compatible with observed choice behavior \citep{Hauser_Wernerfelt_1990_JCR,Goeree_2008_ECMA,vanNierop_et_al_2010_JMR,Honka_Hortacsu_Vitorino_2017_RAND}. For example, \citet*{Reutskaja_et_al_2011_AER} provides interesting experimental evidence against the full attention assumption using eye tracking and choice data.

Motivated by these findings, and the fact that certain theoretically important and empirically relevant choice patterns can not be explained using classical revealed preference theory based on full attention, scholars have proposed other economic models of choice behavior. An alternative is the limited attention model \citep*{Masatlioglu-Nakajima-Ozbay_2012_AER, Lleras_et_al_2017_JET, Dean_Kibris_Masatlioglu_2017_JET}, where decision makers are assumed to select the best available option from a subset of all possible alternatives, known as the consideration set. This framework takes the formation of the consideration set, also known as attention rule or consideration map, as unobservable and hence as an intrinsic feature of the decision maker. Nonetheless, it is possible to develop a fruitful theory of revealed preference within this framework, employing only mild and intuitive nonparametric restrictions on how the decision maker decides to focus attention on specific subsets of all possible alternatives presented to her.

Until very recently, limited attention models have been deterministic, a feature that diminished their empirical applicability: testable implications via revealed preference have relied on the assumption that the decision maker pays attention to the same subset of options every time she is confronted with the same set of available alternatives. This requires that, for example, an online shopper uses always the same keyword and the same search engine (e.g. Google) on the same platform (e.g. tablet) to look for a product. This is obviously restrictive, and can lead to predictions that are inconsistent with observed choice behavior. Aware of this fact, a few scholars have improved deterministic limited attention models by allowing for stochastic attention \citep*{ Manzini-Mariotti_2014_ECMA,Aguiar_2015,Brady-Rehbeck_2016_ECMA,Horan_2018}, which permits the decision maker to pay attention to different subsets with some non-zero probability given the same set of alternatives to choose from. All available results in this literature proceed by first parameterizing the attention rule (i.e., committing to a particular parametric attention rule), and then studying the revealed preference implications of these parametric models. 

In contrast to earlier approaches, we introduce a Random Attention Model (RAM) where we abstain from any specific parametric (stochastic) attention rule, and instead consider a large class of nonparametric random attention rules. Our model imposes one intuitive condition, termed \textit{Monotonic Attention}, which is satisfied by many stochastic attention rules. Given that consideration sets are unobservable, this feature is crucial for applicability of our revealed preference results, as our findings and empirical implications are valid under many different, particular attention rules that could be operating in the background. In other words, our revealed preference results are derived from nonparametric restrictions on the attention rule and hence are more robust to misspecification biases.

RAM is best suited for eliciting information about the preference ordering of a single decision-making unit when her choices are observed repeatedly.\footnote{The finding that individual choices frequently exhibit randomness was first reported in \citet*{Tversky_1969_PR} and has now been illustrated by \citet*{Agranov_Ortoleva_2017} and numerous other studies. Similar to our work, \citet*{Manzini-Mariotti_2014_ECMA}, \citet*{Fudenberg_Iijima_Strzalecki_2015_ECMA}, and \citet*{Brady-Rehbeck_2016_ECMA}, among others, have developed models which allow the analyst to reveal information about the agent's preferences from her observed random choices.} For example, scanner data keeps track of the same single consumer's purchases across repeated visits, where the grocery store adjusts product varieties and arrangements regularly. Another example is web advertising on digital platforms, such as search engines or shopping sites, where not only abundant records from each individual decision maker are available, but also it is common to see manipulations/experiments altering the options offered to them. A third example is given in \citet*{Kawaguchi-Uetake-Watanabe_2016_wp}, where large data on each consumer's choices from vending machines (with varying product availability) is analyzed. In addition, our model can be used empirically with aggregate data on a group of distinct decision makers, provided each of them may differ on what they pay attention to but all share the same preference.

Our key identifying assumption, \textit{Monotonic Attention}, restricts the possibly stochastic attention formation process in a very intuitive way: each consideration set competes for the decision maker's attention, and hence the probability of paying attention to a particular subset is assumed not to decrease when the total number of possible consideration sets decreases. We show that this single nonparametric assumption is general enough to nest most (if not all) previously proposed deterministic and random limited attention models. Furthermore, under our proposed monotonic attention assumption, we are able to develop a theory of revealed preference, obtain specific testable implications, and (partially) identify the underlying preferences of the decision maker by investigating her observed choice probabilities. Our revealed preference results are applicable to a wide range of attention rules, including the parametric ones currently available in the literature which, as we show, satisfy the monotonic attention assumption.

Based on these theoretical findings, we also develop econometric results for identification, estimation, and inference of the decision maker's preferences, as well as specification testing of RAM. We show that RAM implies that the set of partially identified preference orderings containing the decision maker's true preferences is equivalent to a set of inequality restrictions on the choice probabilities (one for each preference ordering in the identified set). This result allows us to employ the identifiable/estimable choice probabilities to (i) develop a model specification test (i.e., test whether there exists a non-empty set of preference orderings compatible with RAM), (ii) conduct hypothesis testing on specific preference orderings (i.e., test whether the inequality constraints on the choice probabilities are satisfied), and (iii) develop confidence sets containing the decision maker's true preferences with pre-specified coverage (i.e., via test inversion). Our econometric methods rely on ideas and results from the literature on partially identified models and moment inequality testing: see \citet*{Canay_Shaikh_2017_ChBook}, \citet*{Ho_Rosen_2017_ChBook} and \citet{Molinari_2019_Handbook} for recent reviews and further references.

RAM is fully nonparametric and agnostic because it relies on the monotonic attention assumption only. As a consequence, it may lead to relatively weak testable implications in some applications, that is, ``little'' revelation or a ``large'' identified set of preferences. However, RAM also provides a basis for incorporating additional (parametric and) nonparametric restrictions that can substantially improve identification power. In this paper, we illustrate how RAM can be combined with additional, mild nonparametric restrictions to tighten identification in non-trivial ways: in Section \ref{section:binary}, we incorporate an additional restriction on attention rule for binary choice problems, and show that this alone leads to important revelation improvements within RAM. We also illustrate this result numerically in our simulation study.

Finally, we implement our estimation and inference methods in the general-purpose software package \texttt{ramchoice} for \texttt{R}---see \url{https://cran.r-project.org/package=ramchoice} for details. Our novel identification results allow us to develop inference methods that avoid optimization over the possibly high-dimensional space of attention rules, leading to methods that are very fast and easy to implement when applied to realistic empirical problems. See the Supplemental Appendix for numerical evidence.

Our work contributes to both economic theory and econometrics. We describe several examples covered by our model in Section \ref{subsection:monotone} after we introduce our proposed RAM. We also discuss in detail the connections and distinctions between this paper and the economic theory literature in Section SA.1 of the Supplemental Appendix. In particular, we show how RAM nests and/or connects to the recent work by \citet*{Manzini-Mariotti_2014_ECMA}, \citet*{Brady-Rehbeck_2016_ECMA}, \citet*{Gul_Natenzon_Pesendorfer_2014_ECMA}, \citet*{Echenique_Saito_Tserenjigmid_2018}, \citet*{Echenique_Saito_2019}, \citet*{Fudenberg_Iijima_Strzalecki_2015_ECMA}, and \citet*{Aguiar_Boccardi_Dean_2016_JET}, among others.

This paper is also related to a rich econometric literature on nonparametric identification, estimation and inference both in the specific context of random utility models, and more generally. See \citet*{Matzkin_2013_ARE} for a review and further references on nonparametric identification, \citet*{Hausman-Newey_2017_ARE} for a recent review and further references on nonparametric welfare analysis, and \citet*{Blundell-Kristensen-Matzkin_2014_JOE}, \citet{Kawaguchi_2017_JoE}, \citet*{Kitamura-Stoye_2018_ECMA}, and \citet*{Deb-Kitamura-Quah-Stoye_2018_wp} for a sample of recent contributions and further references. As mentioned above, a key feature of RAM is that our proposed monotonic attention condition on attention rule nests previous models as special cases, and also covers many new models of choice behavior. In particular, RAM can accommodate more choice behaviors/patterns than what can be rationalized by random utility models. This is important because numerous studies in psychology, finance and marketing have shown that decision makers exhibit limited attention when making choices: they only compare (and choose from) a subset of all available options. Whenever decision makers do not pay full attention to all options, implications from revealed preference theory under random utility models no longer hold in general, implying that empirical testing of substantive hypotheses as well as policy recommendations based on random utility models will be invalid. On the other hand, our results may remain valid.

In contemporaneous work, a few scholars have also developed identification and inference results under (random) limited attention, trying to connect behavioral theory and econometric methods, as we do in this paper. Three recent examples of this new research area include \citet*{Abaluck_Adams_2017}, \citet*{Dardanoni_Manzini_Mariotti_Tyson_2018}, and \citet*{Barseghyan-Coughlin-Molinari-Teitelbaum_2018_wp}. These papers are complementary to ours insofar different assumptions on the random attention rule and preference(s) are imposed, which lead to different levels of (partial) identification of preference(s) and (random) attention rule(s). For a further discussion on the relationship with these papers, see Section SA.1 of the Supplemental Appendix.

The rest of the paper proceeds as follows. Section \ref{section:setup} presents the basic setup, where our key monotonicity assumption on the decision maker's stochastic attention rule is presented in Section \ref{subsection:monotone}. Section \ref{section:RAM} discusses in detail our random attention model, including the main revealed preference results. Section \ref{section:econometrics} presents our main econometrics methods, including nonparametric (partial) identification, estimation, and inference results. In Section \ref{section:binary}, we consider additional restrictions on the attention rule for binary choice problems, which can help improve our identification and inference results considerably. We also consider random attention filters in Section \ref{section:RA filter}, which are one of the motivating examples of monotonic attention rules. In this case, however, there is no additional identification. Section \ref{section:simulations} summarizes the findings from a simulation study. Finally, Section \ref{section:conclusion} concludes with a discussion of directions for future research. A companion online Supplemental Appendix includes more examples, extensions and other methodological results, omitted proofs, and additional simulation evidence.

\section{Setup}\label{section:setup}

We designate a finite set $\Universe$ to act as the universal set of all mutually exclusive alternatives. This set is thus viewed as the grand alternative space, and is kept fixed throughout. A typical element of $\Universe$ is denoted by $a$ and its cardinality is $|\Universe|=K$. We let $\mathcal{X}$ denote the set of all non-empty subsets of $X$. Each member of $\mathcal{X}$ defines a choice problem. 

\begin{defn}[Choice Rule] A \emph{choice rule} is a map $\ChoiProb: X \times \mathcal{X} \rightarrow [0,1]$ such that for all $ S \in \mathcal{X}$, $\ChoiProb(a | S)\geq 0$ for all $a \in S$, $\ChoiProb(a | S)= 0$ for all $a \notin S$, and $\sum\limits_{a \in S}\ChoiProb(a|S)=1$.
\end{defn} 

Thus, $\ChoiProb(a |S)$ represents the probability that the decision maker chooses alternative $a$ from the choice problem $S$. Our formulation allows both stochastic and deterministic choice rules. If $\ChoiProb(a |S)$ is either $0$ or $1$, then choices are deterministic. For simplicity in the exposition, we assume that all choice problems are potentially observable throughout the main paper, but this assumption is relaxed in Section SA.3 of the Supplemental Appendix to account for cases where only data on a subcollection of choice problems is available.

The key ingredient in our model is probabilistic consideration sets. Given a choice problem $S$, each non-empty subset of $S$ could be a consideration set with certain probability. We impose that each frequency is between $0$ and $1$ and that the total frequency adds up to $1$. Formally, 

\begin{defn}[Attention Rule]\label{Definition: Attention Rule} An \emph{attention rule} is a map $\AttFilter: \mathcal{X} \times \mathcal{X} \rightarrow [0,1]$ such that for all $ S \in \mathcal{X}$, $\AttFilter(T|S)\geq 0$ for all $T \subset S$, $\AttFilter(T | S)= 0$ for all $T \not \subset S$, and $\sum\limits_{T \subset S} \AttFilter(T|S)=1$.
\end{defn} 

Thus, $\AttFilter(T|S)$ represents the probability of paying attention to the consideration set $T \subset S $ when the choice problem is $S$. This formulation captures both deterministic and stochastic attention rules. For example, $\AttFilter(S|S)=1$ represents an agent with full attention. Given our approach, we can always extract the probability of paying attention to a specific alternative: For a given $a \in S$, $\sum\limits_{a \in T \subset S} \AttFilter(T|S) $ is the probability of paying attention to $a$ in the choice problem $S$. The probabilities on consideration sets allow us to derive the attention probabilities on alternatives uniquely.

\subsection{Monotonic Attention}\label{subsection:monotone}

We consider a choice model where a decision maker picks the maximal alternative with respect to her preference among the alternatives she pays attention to. Our ultimate goal is to elicit her preferences from observed choice behavior without requiring any information on consideration sets. Of course, this is impossible without any restrictions on her (possibly random) attention rule. For example, a decision maker's choice can always be rationalized by assuming she only pays attention to singleton sets. Because the consumer never considers two alternatives together, one cannot infer her preferences at all.

We propose a property (i.e., an identifying restriction) on how stochastic consideration sets change as choice problems change, as opposed to explicitly modeling how the choice problem determines the consideration set. We argue below that this nonparametric property is indeed satisfied by many problems of interest and mimics heuristics people use in real life (see examples below and in Section SA.2 of the Supplemental Appendix). This approach makes it possible to apply our method to elicit preference without relying on a particular formation mechanism of consideration sets.

\begin{Assumption}[Monotonic Attention]\label{Assumption: monotonicity} 
For any $a \in S-T$, $ \mu(T|S) \leq \mu(T|S-a).$
\end{Assumption}

Monotonic $\mu$ captures the idea that each consideration set competes for consumers' attention: the probability of a particular consideration set does not shrink when the number of possible consideration sets decreases. Removing an alternative that does not belong to the consideration set $T$ results in less competition for $T$, hence the probability of $T$ being the consideration set in the new choice problem is weakly higher. Our assumption is similar to the regularity condition proposed by \cite{Suppes-Luce_1965_Handbook}. The key difference is that their regularity condition is defined on choice probabilities, while our assumption is defined on attention probabilities.

To demonstrate the richness of the framework and motivate the analysis to follow, we discuss six leading examples of families of monotonic attention rules, that is, attention rules satisfying Assumption \ref{Assumption: monotonicity}. We offer several more examples in Section SA.2 of the Supplemental Appendix. The first example is deterministic (i.e., $\AttFilter(T|S) $ is either $0$ or $1$), but the others are all stochastic.

\begin{enumerate}

\item \citep*[\textsc{Attention Filter};][]{Masatlioglu-Nakajima-Ozbay_2012_AER} A large class of deterministic attention rules, leading to consideration sets that do not change if an item not attracting attention is made unavailable (Attention Filter), was introduced by \citet{Masatlioglu-Nakajima-Ozbay_2012_AER}. A classical example in this class is when a decision maker considers all the items appearing in the first page of search results and overlooks the rest. Formally, let $\Gamma(S)$ be the deterministic consideration set when the choice problem is $S$, and hence $\Gamma(S) \subset S$. Then, $\Gamma$ is an Attention Filter if when $a \notin \Gamma(S)$, then $\Gamma(S - a)=\Gamma(S)$. In our framework, this class corresponds to the case $\AttFilter(T|S)=1$ if $T = \Gamma(S)$, and $0$ otherwise.

\item (\textsc{Random Attention Filters}) Consider a decision maker whose attention is deterministic but utilizes different deterministic attention filters on different occasions. For example, it is well-known that search behavior on distinct platforms (mobile, tablet, and desktop) is drastically different (e.g., the same search engine produces different first page lists depending on the platform, or different platforms utilize different search algorithms). In such cases, while the consideration set comes from a (deterministic) attention filter for each platform, the resulting consideration set is random. Formally, if a decision maker utilizes each attention filter $\Gamma_j$ with probability $\psi_j$, then the attention rule can be written as
\[
\AttFilter(T|S)=\sum_{j}\Indicator(\Gamma_j(S)=T)\cdot \psi_j.
\]

We will pay special attention to this class of attention rules in Section \ref{section:RA filter}.

\item \citep*[\textsc{Independent Consideration};][]{Manzini-Mariotti_2014_ECMA} Consider a decision maker who pays attention to each alternative $a$ with a fixed probability $\gamma(a) \in (0,1)$. $\gamma$ represents the degree of brand awareness for a product, or the willingness of an agent to seriously evaluate a political candidate. The frequency of each set being the consideration set can be expressed as follows: for all $T \subset S$, 
\[
\AttFilter(T|S)= \frac{1}{\beta_S}\prod_{a\in T}\gamma(a)\prod_{a\in S-T}(1-\gamma(a)),
\]
where $\beta_{S}=1-\prod\limits_{a \in S} (1-\gamma(a))$, which represents the probability that the decision maker considers no alternative in $S$, is used to adjust each probability so that they sum up to $1$.

\item \citep[\textsc{Logit Attention};][]{Brady-Rehbeck_2016_ECMA} Consider a decision maker who assigns a positive weight for each non-empty subset of $X$. Psychologically $w_{T}$ is a strength associated with the subset $T$. The probability of considering $T$ in $S$ can be written as follows: 
\[
\AttFilter(T|S)=\frac{w_{T}}{\sum\limits_{\substack{T' \subset S }} w_{T'}}. 
\]
Even though there is no structure on weights in the general version of this model, there are two interesting special cases where weights solely depend on the size of the set. These are $w_{T}=|T|$ and $w_{T}=\frac{1}{|T|}$, which are conceptually different. In the latter, the decision maker tends to have smaller consideration sets, while larger consideration sets are more likely in the former.

\item (\textsc{Dogit Attention}) This example is a generalization of Logit Attention, and is based on the idea of the Dogit model \citep{Gaundry_1979}. A decision maker is captive to a particular consideration set with certain probability, to the extent that she pays attention to that consideration set regardless of the weights of other possible consideration sets. Formally, let
\[
\AttFilter(T|S)= \frac{1}{1+ \sum_{\substack{T' \subset S }} \theta_{T'}} \frac{w_{T}}{\sum_{\substack{T' \subset S }} w_{T'}} + \frac{\theta_{T}}{1+\sum_{\substack{T' \subset S }} \theta_{T'}},
\]
where $\theta_{T} \geq 0$ represents the degree of captivity (impulsivity) of $T$. The ``captivity parameter'' reflects the attachment of a decision maker to a certain consideration set. Since $w_T$ are non-negative, the second term, which is  independent of $w_T$, is the smallest lower bound for $\AttFilter(T|S)$. The larger $\theta_{T}$, the more likely the decision maker is to be captive to $T$ and pays attention to it. When $\theta_T=0$ for all $T$, this model becomes Logit Attention. This formulation is able to distinguish between impulsive and deliberate attention behavior. 

\item (\textsc{Elimination by Aspects}) Consider a decision maker who intentionally or unintentionally focuses on a certain ``aspect'' of alternatives, and then refuses or ignores those alternatives that do not possess that aspect. This model is similar in spirit to \cite{Tversky_1972}. Let $\{j,k,\ell, \dots \}$ be the set of aspects. Let $\omega_j$ represent the probability that aspect $j$ ``draws attention to itself.'' It reflects the salience and/or importance of aspect $j$. All alternatives without that aspect fail to receive attention. Let $B_j$ be the set of alternatives that posses aspect $j$. We assume that each alternative must belong to at least one $B_{j}$ with $\omega_j >0$. If aspect $j$ is the salient aspect, the consideration set is $B_{j} \cap S$ when $S$ is the set of feasible alternatives. The total probability of $T$ being the consideration set is the sum of $\omega_j$ such that $T=B_{j} \cap S$. When there is no alternative in $S$ possessing the salient aspect, a new aspect will be drawn. Formally, the probability of $T$ being the consideration set under $S$ is given by 
\[
\AttFilter(T|S)=\sum_{\substack{B_{j}\cap S=T}} \frac{\omega_j}{\sum_{\substack{B_{k}\cap S\neq\emptyset}}\omega_k}.
\]

\end{enumerate}

These six examples give a sample of different limited attention models of interest in economics, psychology, marketing, and many other disciplines. While these examples are quite distinct from each other, all of them are monotonic attention rules.\footnote{To provide an example where Assumption \ref{Assumption: monotonicity} might be violated, consider a generalization of Independent Consideration of \citet{Manzini-Mariotti_2014_ECMA}. In this generalization, the degree of brand awareness for a product is not only a function of the product but also a function of the context, that is, $\gamma_S(a)$. Then, the frequency of each set being the consideration set is calculated as in Independent Consideration rule. Due to this contextual dependence, 
further restrictions on $\gamma_{S}(a)$ and $\gamma_{S-b}(a)$ are needed to ensure Assumption \ref{Assumption: monotonicity}.}  As a consequence, our revealed preference characterization will be applicable to a wide range of choice rules without committing to a particular attention mechanism, which is not observable in practice and hence untestable. Furthermore, as illustrated by the examples above (and those in Section SA.2 of the Supplemental Appendix), our upcoming characterization, identification, estimation, and inference results nest important previous contributions in the literature.

\section{A Random Attention Model\label{section:RAM}}

We are ready to introduce our random attention model based on Assumption \ref{Assumption: monotonicity}. We assume the decision maker has a strict preference ordering $\succ$ on $X$. To be precise, we assume the preference ordering is an \textit{asymmetric}, \textit{transitive} and \textit{complete} binary relation. A binary relation $\succ$ on a set $X$ is (i) asymmetric, if for all $x, y \in X$, $x \succ y$ implies $y \not\succ x$; (ii) transitive, if for all $x, y,z \in X$, $x \succ y$ and $y \succ z$ imply $x \succ z$; and (iii) complete, if for all $x \neq y \in X $, either $x \succ y$ or $y \succ x$ is true. Consequently, the decision maker always picks the maximal alternative with respect to her preference among the alternatives she pays attention to. Formally,
\begin{defn}[Random Attention Representation]\label{def: random attention representation}
A choice rule $\ChoiProb$ has a \textit{random attention representation} if there exists a preference ordering $\succ $ over $X$ and a monotonic attention rule $\AttFilter$ (Assumption \ref{Assumption: monotonicity}) such that
\[\ChoiProb(a|S)=\sum_{T\subset S}\Indicator(\text{$a$ is $\succ$-best in $T$})\cdot\AttFilter(T|S)\]
for all $a\in S$ and $S\in \mathcal{X}$. 
In this case, we say $\ChoiProb$ is represented by $(\succ, \AttFilter)$. We may also say that $\succ$ represents $\ChoiProb$, which means that there exists some monotonic attention rule $\AttFilter$ such that $(\succ, \AttFilter)$ represents $\ChoiProb$. We also say $\ChoiProb$ is a \textit{Random Attention Model} (RAM).
\end{defn} 

While our framework is designed to model stochastic choices, it captures deterministic choices as well. In classical choice theory, a decision maker chooses the best alternative according to her preferences with probability $1$, hence choice is deterministic. In our framework, this case is captured by a monotone attention rule with $\mu(S|S)=1$. Figure \ref{fig:fig1} gives a graphical representation of RAM.

\begin{figure}[!th]
\centering
\begin{tikzpicture}

\node at (4 , 0)  {$S$};
\node at (9 , 0)  {$T$};
\node at (14, 0)  {$a$ is $\succ$-best in $T$};

\draw (4, 0) circle (0.8);
\draw[dashed] (9, 0) circle (0.6);

\draw[dashed, ->] (5.3, 0)  -- (7.9, 0);
\draw[dashed, ->] (10.1, 0)  -- (12.3, 0);

\node at (6.6 , 0.4)  {$\mu(T|S)$};
\node at (6.6 , -0.4)  {\footnotesize{Attention Rule}};

\node at (11.2 , 0.4)  {$\succ$};
\node at (11.2 , -0.4)  {\footnotesize{Preference}};

\node at (4 , -1.4)  {\footnotesize{Choice Problem}};
\node at (9 , -1.4)  {\footnotesize{Consideration Set}};
\node at (14 , -1.4)  {\footnotesize{Choice}};

\draw[->] (4, -1.8) -- (4, -3)  -- (14, -3) -- (14, -1.8);

\node at (9 , -2.6)  {$\pi(a|S)$};
\node at (9 , -3.4)  {\footnotesize{Choice Rule}};
\end{tikzpicture}
\caption{Illustration of a RAM. \textit{Observable}: choice problem and choice (solid line). \textit{Unobservable}: attention rule, consideration set and preference (dashed line).}
\label{fig:fig1}
\end{figure}
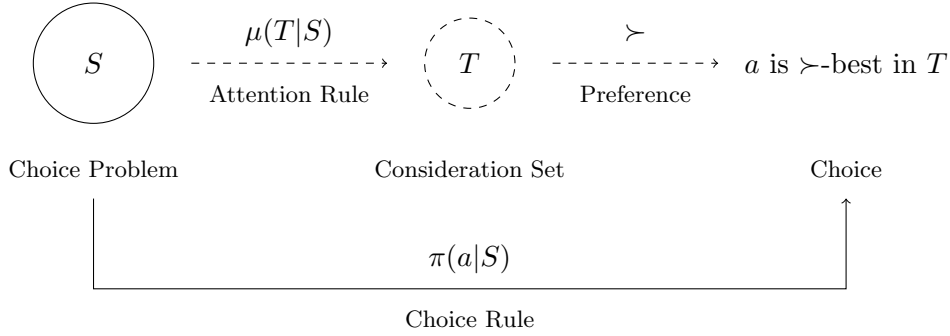

We now derive the implications of our random attention model. They can be used to test the model in terms of observed choice rules/probabilities. In this section, we treat the choice rule as known/observed in order to facilitate the discussion of preference elicitation. In practice, the researcher may only observe a set of choice problems and choices thereof. We discuss econometric implementation in Section \ref{section:econometrics}: even if the choice rule is not directly observed, it is identified (consistently estimable) from choice data.

In the literature, there is a principle called \textit{regularity} \citep[see][]{Suppes-Luce_1965_Handbook}, according to which adding a new alternative should only decrease the probability of choosing one of the existing alternatives. However, empirical findings suggest otherwise. \citet*{Rieskamp_Busemeyer_Mellers_2006_JEL} provide a detailed review of empirical evidence on violations of regularity and alternative theories explaining these violations. Importantly, our model allows regularity violations.

The next example illustrates that adding an alternative to the feasible set can increase the likelihood that an existing alternative is selected. This cannot be accommodated in the Luce (multinomial logit) model, nor in any random utility model. In RAM, the addition of an alternative changes the choice set and therefore the decision maker's attention, which could increase the probability of an existing alternative being chosen.

\begin{example}[Regularity Violation]
Let $X=\{a,b,c\}$ and consider two nested choice problems $\{a,b,c\} $ and $\{a,b\}$. Imagine a decision maker with $a\succ b\succ c$ and the following monotonic attention rule $\AttFilter$. Each row corresponds to a different choice problem and columns represent possible consideration sets. 
\[
\begin{array}{r|ccccccc}
  \AttFilter(T|S) &    T=\{a,b,c\} & \{a,b\} & \{a,c\} & \{b,c\} & \{a\} & \{b\} & \{c\}\\ \hline
  S=\{a,b,c\} & 2/3 & 0 & 0 & 1/6 & 0 & 0 & 1/6 \\
 \{a,b\} & & 1/2 & & & 0 & 1/2 & \\
 \{a,c\} & & & 1/2 & & 0 & & 1/2 \\
 \{b,c\} & & & & 1/2 & & 0 & 1/2 \\
 \end{array}
\]
Then $\ChoiProb(a|\{a,b,c\})=2/3>1/2=\ChoiProb(a|\{a,b\})=\ChoiProb(a|\{a,c\})$. 
\end{example}

This example shows that RAM can explain choice patterns that cannot be explained by the classical random utility model.  Given that the model allows regularity violations, one might think that the model has very limited empirical implications, i.e. that it is too general to have empirical content. However, it is easy to find a choice rule $\ChoiProb$ that lies outside RAM with only three alternatives. Here we provide an example where our model makes very sharp predictions. 

\begin{example}[RAM Violation]\label{example:outside} The following choice rule $\ChoiProb$ is not compatible with our random attention model as long as the decision maker chooses each alternative with positive probability from the set $\{a,b,c\}$, i.e., $\lambda_a  \lambda_b \lambda_c>0$. Each column corresponds to a different choice problem. 
\[
\begin{array}{c|cccc}
 \ChoiProb(\cdot|S) & S=\{a,b,c\} & \{a,b\} & \{a,c\} & \{b,c\} \\ \hline
 a & \lambda_a & 1 & 0 & \\
 b & \lambda_b & 0 & & 1 \\
 c & \lambda_c & & 1 & 0
 \end{array}
\]
We now illustrate that $\ChoiProb$ is not a RAM. Since the choice behavior is symmetric among all binary choices, without loss of generality, assume $a \succ b \succ c$. Given that 
 $c$ is the worst alternative, $\{c\}$ is the only consideration set in which $c$ can be chosen.  Hence  the decision maker must consider the consideration set $\{c\}$ with probability $\lambda_c$ (i.e., $\AttFilter(\{c\}|\{a,b,c\})=\lambda_c$). Assumption \ref{Assumption: monotonicity} implies that $\AttFilter(\{c\}| \{b,c\} )$ must be greater than $\lambda_c >0$. This yields a contradiction since $\ChoiProb(c|\{b,c\})=0$. In sum, given the above binary choices, our model predicts that when the choice set is $\{a,b,c\}$ the decision maker must choose at least one alternative with $0$ probability, which is a stark prediction in probabilistic choice. \end{example}

One might wonder that the model makes a strong prediction due to the cyclical binary choices, i.e., $\ChoiProb(a|\{a,b\})=\ChoiProb(b|\{b,c\})=\ChoiProb(c|\{a,c\})=1$. We can generate a similar prediction where the individual is perfectly rational in the binary choices, i.e., $\ChoiProb(a|\{a,b\})=\ChoiProb(a|\{a,c\})=\ChoiProb(b|\{b,c\})=1$. In this case, our model predicts that the individual cannot chose both $b$ and $c$ with strictly positive probability when the choice problem is $\{a,b,c\}$. Therefore, we obtain similar predictions. Given that RAM has non-trivial empirical content, it is natural to investigate to what extent Assumption \ref{Assumption: monotonicity} can be used to elicit (unobserved) strict preference orderings given (observed) choices of decision makers.

\subsection{Revealed Preference \label{section:revealedpreference}}

In general, a choice rule can have multiple RAM representations with different preference orderings and different attention rules. When multiple representations are possible, we say that $a$ is revealed to be preferred to $b$ if and only if $a$ is preferred to $b$ in all possible RAM representations. This is a very conservative approach as it ensures that we never make false claims about the preference of the decision maker.

\begin{defn}[Revealed Preference]
\label{Definition: revealedpref}
Let $\{(\succ_{j}, \AttFilter_{j})\}_{j=1,\dots,J}$ be all random attention representations of $\ChoiProb$. We say that $a$ is \textit{revealed to be preferred} to $b$ if $a \succ _{j} b$ for all $j$. 
\end{defn} 

We now show how revealed preference theory can still be developed successfully in our RAM framework. If all representations share the same preferences $\succ$ (or if there is a unique representation), then the revealed preference will be equal to $\succ$. In general, if one wants to know whether $a$ is revealed to be preferred to $b$, it would appear necessary to identify all possible $(\succ_{j},\AttFilter_{j})$ representations. However, this is not practical, especially when there are many alternatives. Instead, we shall now provide a handy method to obtain the revealed preference completely.

Our theoretical strategy parallels that of \citet*{Masatlioglu-Nakajima-Ozbay_2012_AER} (MNO) in their study of a deterministic model of inattention. MNO identifies $a$ as revealed preferred to $b$ whenever $a$ is chosen in the presence of $b$, and removing $b$ causes a choice reversal. This particular observation, in conjunction with the structure of attention filters, ensures that the decision maker considers $b$ while choosing $a$. Here, we show that $a$ is revealed preferred to $b$ if removing $b$ causes a regularity violation, that is, $\ChoiProb(a|S) > \ChoiProb(a|S-b)$.  To see this, assume $(\succ,\AttFilter)$ represents $\ChoiProb$ and $\ChoiProb(a|S) > \ChoiProb(a|S-b)$.  By definition, we have \vspace{-0.2cm} 
\begin{alignat*}{2}
\ChoiProb( a |S ) &= \sum_{\substack{  T \subset S,  \\ a\text{ is }{{\succ}} \text{%
-best in }T }} \AttFilter( T|S )\quad &&\ \\
 &=\sum_{\substack{ {  {b\in T \subset S}},  \\ a\text{ is }{{\succ}} \text{-best in }T }}{\AttFilter(T|S)} &&+ \sum_{\substack{ {  {b \notin T \subset S}},  \\ a\text{ is }{{\succ}} \text{%
-best in }T }}{\AttFilter(T|S)} \\
&\leq \sum_{\substack{ b \in T \subset S,  \\ a \text{ is } \succ \text{-best in }T }}  \AttFilter(T|S) &&+ \sum_{\substack{  T \subset S- b,  \\ a\text{ is } \succ  \text{-best in } T }}  \AttFilter ( T| {  {S -b}} )  \qquad (\text{by Assumption \ref{Assumption: monotonicity}})\\
&= \sum_{\substack{ b \in T \subset S,  \\ a\text{ is } \succ \text{-best in }T }}  \AttFilter(T|S) &&+ \quad  {\ChoiProb(a|S -b)}.
\end{alignat*}
Hence, we have the following inequality: $$ {{\ChoiProb(a|S ) - \ChoiProb(a|S -b)}}  \leq \sum_{\substack{ {{b \in T \subset S}},  \\ a\text{ is }{{\succ}} \text{-best in }T}}{\AttFilter(T|S)} $$
Since $\ChoiProb ( a|S ) - \ChoiProb(a|S -b) > 0$, there must exist at least one $T$ such that (i) $b \in T$, (ii) $a\text{ is }{{\succ}} \text{-best in }T$, and (iii) $\AttFilter(T|S) \neq 0$. Therefore, there exists at least one occasion that the decision maker pays attention to $b$ while choosing $a$ (Revealed Preference).  The next lemma summarizes this interesting relationship between regularity violations and revealed preferences. It simply illustrates that the existence of a regularity violation informs us about the underlying preference. 

\begin{lem}\label{Lemma: Removing Dominating}
Let $\ChoiProb$ be a RAM. If $\ChoiProb(a|S) > \ChoiProb(a|S -b)$, then $a$ is revealed to be preferred to $b$.
\end{lem}

Lemma \ref{Lemma: Removing Dominating} allows us to define the following binary relation. For any distinct $a$ and $b$, define:
\begin{center}
    $a \mathsf{P}b$, if there exists $S\in \mathcal{X}$ including $a$ and $b$ such that $\ChoiProb(a|S) > \ChoiProb(a|S-b)$.
\end{center}
By Lemma \ref{Lemma: Removing Dominating}, if $a\mathsf{P}b$ then $a$ is revealed to be preferred to $b$. In other words, this condition is sufficient to reveal preference. In addition, since the underlying preference is transitive, we also conclude that she prefers $a$ to $c$ if $a\mathsf{P}b$ and $b\mathsf{P}c$ for some $b$, even when $a\mathsf{P}c$ is not directly revealed from her choices. Therefore, the transitive closure of $\mathsf{P}$, denoted by $\mathsf{P}_\mathtt{R}$, must also be part of her revealed preference. One may wonder whether some revealed preference is overlooked by $\mathsf{P}_\mathtt{R}$. The following theorem, which is our first main result, shows that $\mathsf{P}_\mathtt{R}$ includes all preference information given the observed choice probabilities, under only Assumption \ref{Assumption: monotonicity}.

\begin{thm}[Revealed Preference]\label{Theorem: Revealed Preference}
Let $\ChoiProb$ be a RAM. Then $a$ is revealed to be preferred to $b$ if and only if $a\mathsf{P}_\mathtt{R}b$. \end{thm}

\begin{proof}
The ``if'' part follows from Lemma \ref{Lemma: Removing Dominating}. To prove the ``only if'' part, we show that given any preference $\succ$ that includes $\mathsf{P}_\mathtt{R}$, there exists a monotonic attention rule $\AttFilter$ such that $(\succ,\AttFilter)$ represents $\ChoiProb$. The details of the construction can be found in the proof of Theorem \ref{Theorem: Characterization}. 
\end{proof}

Theorem \ref{Theorem: Revealed Preference} establishes the empirical content of revealed preferences under monotonic attention only. Our resulting revealed preferences could be incomplete: it may only provide coarse welfare judgments in some cases. At one extreme, there is no preference revelation when there is no regularity violation. This is because the decision maker's behavior can be attributed fully to her preference or to her inattention (i.e., never considering anything other than her actual choice). This highlights the fact that our revealed preference definition is conservative, which guarantees no false claims in terms of revealed preference especially when there are alternative explanations for the same choice behavior. The following example illustrates that we might make misleading inferences if we wrongly believe the decision maker uses a particular attention rule.

\begin{example}[Avoiding Misleading Inference]\label{example: wrong inference} We now describe a typical online customer's search behavior.  For simplicity, there are three products $a$, $b$, and $c$.  She prefers $c$ over $a$ and $a$ over $b$ (not observable). She visits two different search engines: $G$ and $Y$. 85 percent of her search takes place on Engine G across three different platforms: laptop (20\%), tablet (50\%), and smartphone (15\%). Engine G always lists $b$ before $a$ and $a$ before $c$. Due to screen size, Engine G lists up to three, two, and one product information on laptops, tablets, and smartphones, respectively. The rest of her search is on Engine Y (15\%), which has a unique platform. In this engine, $a$ is listed first, if it is available, and clicking $a$'s link will provide information about both $a$ and $c$. If $a$ is not available, $b$ is listed first. In Engine Y, she clicks only one link. (When she uses Engine Y, her consideration set is $\{a,c\}$ when $a$ and $c$ are both available, $\{a\}$ when $a$ is available but not $c$, finally $\{b\}$ when only $b$ and $c$ are available.) Based on her underlying preference, above consideration set formation leads to stochastic choice, the frequencies of which are reported in the following table:
		 
$$\begin{array}{c|cccc}
 \ChoiProb(\cdot|S) & S=\{a,b,c\} & \{a,b\} & \{a,c\} & \{b,c\} \\ \hline
 a & 0.50  & 0.85 & 0.15 & \\
 b & 0.15 & 0.15 &         & 0.30 \\
 c & 0.35 &         & 0.85 & 0.70
 \end{array}$$
 
Assume that we observe the customer's choice data without any knowledge about her underlying search behavior. First, note that above choice data is consistent with the logit attention model of \citet{Brady-Rehbeck_2016_ECMA}.\footnote{For example, letting the preference order be $a\succ b\succ c$ and let weights be given as $w_{\{a\}}=0$, $w_{\{b\}}=1/20$, $w_{\{c\}}=7/20$, $w_{\{a,b\}}=17/60$, $w_{\{a,c\}}=21/340$, $w_{\{b,c\}}=1/10$, $w_{\{a,b,c\}}=79/510$, it is easy to check that this is a logit attention representation of the choice data given above.} In other words, we can apply their revealed preference result for this choice data. Their model, then, concludes that the unique revealed preference is $a \succ b \succ c$, which however is not the true one that has generated the data. Therefore, if we make a mistaken assumption that the customer's behavior is in line with the logit model, we will infer that $c$ is the worst alternative when it is the best product for our customer.  
\end{example}

Example \ref{example: wrong inference} is an example where a specific cnsideration set formation model
leads to wrong conclusions on the revealed preferences. This example highlights the importance of knowledge about the underlying choice procedure when we conduct welfare analysis. In other words, welfare analysis is more delicate a task than it looks. Notice that, in the above example, monotonic attention is satisfied as engines do not change their presentations of first page results when an alternative outside of the first page becomes unavailable. Hence, Theorem \ref{Theorem: Revealed Preference} is applicable. Since $\pi(a|\{a,b,c\})> \pi(a|\{a,c\})$, our model correctly identifies her true preference between $a$ and $b$. However, our model is silent about the relative ranking of $c$. Therefore, while our revealed preference is conservative, it does not make misleading claims.

We now illustrate that Theorem \ref{Theorem: Revealed Preference} could be very useful to understand the attraction effect phenomena. The attraction effect introduced by \citet*{Huber_1982} was the first evidence against the regularity condition. It refers to an inferior product's ability to increase the attractiveness of another alternative when this inferior product is added to a choice set. In a typical attraction effect experiment, we observe $\ChoiProb(a|\{a,b,c\}) > \ChoiProb(a|\{a,b\})$. Assume that we have no information about the alternatives other than the frequency of choices. Then, by simply using observed choice, Theorem \ref{Theorem: Revealed Preference} informs us that the third product $c$ is indeed an inferior alternative compared to $a$ ($a \succ c$). This is exactly how these alternatives are chosen in these experiments. While alternatives $a$ and $b$ are not comparable, alternative $c$, which is also not comparable to $b$, is dominated by $a$. Theorem \ref{Theorem: Revealed Preference} informs us about the nature of products by just observing choice frequencies. 

Our revealed preference result includes the one in \citet*{Masatlioglu-Nakajima-Ozbay_2012_AER} for attention filters (i.e., non-random monotonic attention rules). In their model, $a$ is revealed to be preferred to $b$ if there is a choice problem such that $a$ is chosen and $b$ is available, but it is no longer chosen when $b$ is removed from the choice problem. This means we have $1=\ChoiProb(a|S) > \ChoiProb(a|S- b)=0$. Given Theorem \ref{Theorem: Revealed Preference}, this reveals that $a$ is better than $b$. On the other hand, generalizing this result to non-deterministic attention rules allows for a broader class of empirical and theoretical settings to be analyzed, hence our revealed preference result (Theorem \ref{Theorem: Revealed Preference}) is strictly richer than those obtained in previous work. For example, in a deterministic world with three alternatives, there is no data revealing the entire preference. On the other hand, we illustrate that it is possible to reveal the entire preference in RAM with only three alternatives. This discussion makes clear the connection between deterministic and probabilistic choice in terms of revealed preference. 

\begin{example}[Full Revelation]\label{example:full_revelation_general}
Consider the following stochastic choice with three alternatives:
\begin{equation*}\begin{array}{c|cccc}
 \ChoiProb(\cdot|S) & S=\{a,b,c\} & \{a,b\} & \{a,c\} & \{b,c\} \\ \hline
 a & \lambda & 1-\lambda_b & \lambda_a & \\
 b & 1-\lambda & \lambda_b & & 1-\lambda_c \\
 c & 0 & & 1-\lambda_a & \lambda_c
 \end{array}
\end{equation*}
If  $1- \lambda_b > \lambda > \lambda_a,\lambda_c$, then we can  verify that $\ChoiProb$ has a random attention representation (see Theorem \ref{Theorem: Characterization}). Now we show that in all possible representations of $\ChoiProb$, $a\succ b\succ c$ must hold. By Lemma \ref{Lemma: Removing Dominating}, $\ChoiProb(a|\{a,b,c\}) > \ChoiProb(a|\{a,c\})$ implies that $a$ is revealed to be preferred to $b$. Similarly, $\ChoiProb(b|\{a,b,c\}) > \ChoiProb(b|\{a,b\})$ implies that $b$ is revealed to be preferred to $c$. Hence preference is uniquely identified. \end{example} 

Example \ref{example:full_revelation_general} also illustrates that one can achieve unique identification of preferences by utilizing Assumption \ref{Assumption: monotonicity} even when observed choices cannot be explained by well known random attention models such as the logit attention model of \citet{Brady-Rehbeck_2016_ECMA} and the independent attention model of \citet{Manzini-Mariotti_2014_ECMA}. To see this point, assume that $\max\{1-\lambda_a,\lambda_c\}>0$. One can show that neither \citet{Brady-Rehbeck_2016_ECMA} nor \citet{Manzini-Mariotti_2014_ECMA} can explain observed choices in this example. First, notice that since both models satisfy Assumption \ref{Assumption: monotonicity} and the preference is uniquely revealed as $a\succ b\succ c$ under Assumption \ref{Assumption: monotonicity}, if the observed choice data can be explained by either model, then their revealed preference must also be $a\succ b\succ c$. That is $c$ must be the worst alternative. On the other hand, $c$ is chosen with zero probability in $\{a,b,c\}$. These models then imply that $c$ must also be chosen with zero probability in $\{a,c\}$ and $\{b,c\}$. This contradicts our assumption that $\max\{1-\lambda_a,\lambda_c\}>0$.

\subsection{A Characterization\label{section:Characterization}}

Theorem \ref{Theorem: Revealed Preference} characterizes the revealed preference in our model. However, it is not applicable unless the observed choice behavior has a random attention representation, which motivates the following question: how can we test whether a choice rule is consistent with RAM? It turns out that RAM can be simply characterized by only one behavioral postulate of choice: acyclicity. Our characterization is based on an idea similar to \citet{Houthakker_1950_Economica}. Choices reveal information about preferences. If these revelations are consistent in the sense that there is no cyclical preference revelation, the choice behavior has a RAM representation.

\begin{thm}[Characterization]\label{Theorem: Characterization}
A choice rule $\ChoiProb$ has a random attention representation if and only if $\mathsf{P}$ has no cycle.
\end{thm}

Recall that Example \ref{example:outside} is outside of our model. Theorem \ref{Theorem: Characterization} implies that $\mathsf{P}_\mathtt{R}$ must have a cycle. Indeed, we have $a \mathsf{P} b$ due to the regularity violation $\ChoiProb(a|\{a,b,c\}) =\lambda_a >0 = \ChoiProb(a|\{a,c\})$. Similarly, we have  $b \mathsf{P} c$ by $\ChoiProb(b|\{a,b,c\}) =\lambda_b >0 = \ChoiProb(b|\{a,b\})$) and $c \mathsf{P} a$ by ($\ChoiProb(c|\{a,b,c\}) =\lambda_c >0 = \ChoiProb(c|\{b,c\})$). Since $\mathsf{P}$ has a cycle, Example \ref{example:outside} must be outside of our model. Therefore, Theorem \ref{Theorem: Characterization} provides a very simple test of RAM.

Our characterization result also helps us to understand the relation between our model and random utility models. It is well-known in the literature that any choice rule that has a random utility model representation satisfies regularity. On the other hand, for any choice rule that satisfies regularity, $\mathsf{P}$ will trivially have no cycle. Hence, any choice rule that has a random utility model representation also has a RAM representation. However, in terms of modeling purposes, RAM assumes random attention with a deterministic preference whereas random utility model assumes random preference and deterministic (full) attention.

Before closing this section, we sketch the proof of Theorem \ref{Theorem: Characterization}, and provide a corollary which will be used in the next section for developing econometric methods. The ``only if'' part of Theorem \ref{Theorem: Characterization} follows directly from Lemma \ref{Lemma: Removing Dominating}. For the ``if'' part, we need to construct a preference and a monotonic attention rule representing the choice rule. Given that $\mathsf{P}$ has no cycle, there exists a preference relation $\succ$ including $\mathsf{P}_\mathtt{R}$. Indeed, we illustrate that any such completion of $\mathsf{P}_\mathtt{R}$ represents $\ChoiProb$ by an appropriately chosen $\AttFilter$. The construction of $\AttFilter$ depends on a particular completion of $\mathsf{P}_\mathtt{R}$, and is not unique in general. We then illustrate that the constructed $\AttFilter$ satisfies Assumption \ref{Assumption: monotonicity}. At the last step, we show that $(\succ, \AttFilter)$ represents $\ChoiProb$. In Corollary \ref{Lemma: appendix, triangular}, we provide one specific construction of the attention rule. We first make a definition. 

\begin{defn}[Lower Contour Set; Triangular Attention Rule]\label{Definition: appendix, triangular}
Given a preference ordering $\succ$ of the alternatives in $X$: $a_{1,\succ}\succ a_{2,\succ}\succ \cdots \succ a_{K,\succ}$, a \emph{lower contour set} is defined as $L_{k,\succ} = \{ a_{j,\succ}:\ j\geq k \} = \{ a\in X:\ a\preceq a_{k,\succ} \}$. A \emph{triangular attention rule} is an attention rule which puts weights only on lower contour sets. That is, $\mu(T|S)>0$ implies $T = L_{k,\succ}\cap S$ for some $k$ such that $a_{k,\succ}\in S$. 
\end{defn}

\begin{coro}[Monotonic Triangular Attention Rule Representation]\label{Lemma: appendix, triangular}
Assume $(\succ,\mu)$ is a representation of $\pi$ with $\mu$ satisfying Assumption \ref{Assumption: monotonicity}. Then there is a unique triangular attention rule $\tilde{\mu}$ corresponding to $\succ$, which also satisfies Assumption \ref{Assumption: monotonicity}, such that  $(\succ,\tilde{\mu})$ is a representation of $\pi$. 
\end{coro}

\section{Econometric Methods\label{section:econometrics}}

Theorem \ref{Theorem: Revealed Preference} shows that if the choice probability $\ChoiProb$ is a RAM then preference revelation is possible. Theorem \ref{Theorem: Characterization} gives a falsification result, based on which a specification test can be designed. The challenge for econometric implementation, however, is that our main assumption, monotonic attention, is imposed on the attention rule, and that the attention rule is not identified from a typical choice data and has a much higher dimension than the identified (consistently estimable) choice rule. To circumvent this difficulty, we rely on Corollary \ref{Lemma: appendix, triangular}, which states that if $\ChoiProb$ has a random attention representation $(\succ,\mu)$, then there exists a unique monotonic triangular attention rule $\tilde{\mu}$ such that $(\succ,\tilde{\mu})$ is also a representation of $\pi$. This latter result turns out to be useful for our proposed identification, estimation, and inference methods, as it allows us to construct, for each given preference ordering, a mapping from the identified choice rule to a triangular attention rule, for which we can test whether Assumption \ref{Assumption: monotonicity} holds. This test turns out to be a test on moment inequalities. 

\subsection{Nonparametric Identification}

We first define the set of partially identified preferences, which mirrors Definition \ref{def: random attention representation}, with the only difference that now we fix the choice rule to be identified/estimated from data. More precisely, let $\ChoiProb$ be the underlying choice rule/data generating process. Then a preference $\succ$ is compatible with $\ChoiProb$, denoted by $\succ\ \in\ParSet_{\ChoiProb}$,\footnote{$\ParSet_{\ChoiProb}$ is not the same as $\mathsf{P}_{\mathtt{R}}$ (defined in Section \ref{section:revealedpreference}): $\mathsf{P}_{\mathtt{R}}$ contains all revealed preferences, while $\ParSet_{\ChoiProb}$ is the set of preferences compatible with the choice probability (i.e., all possible completions of $\mathsf{P}_{\mathtt{R}}$). For example, when there is no preference revelation, $\ParSet_{\ChoiProb}$ contains all preference orderings, and $\mathsf{P}_{\mathtt{R}}$ will be empty. For the other extreme that the choice probability is not compatible with our RAM, $\ParSet_{\ChoiProb}$ will be empty and $\mathsf{P}_{\mathtt{R}}$ will involve cycles.} if there exists some monotonic attention rule $\mu$  such that $(\ChoiProb,\succ,\mu)$ is a RAM. 

When $\ChoiProb$ is known, it is possible to employ Theorem \ref{Theorem: Revealed Preference} directly to construct $\ParSet_{\ChoiProb}$. For example, consider the specific preference ordering $a\succ b$, which can be checked by the following procedure. First, check whether $\ChoiProb(b|S)\leq \ChoiProb(b|S-a)$ is violated for some $S$. If so, then we know the preference ordering is not compatible with RAM and hence does not belong to $\ParSet_{\ChoiProb}$ (Lemma \ref{Lemma: Removing Dominating}). On the other hand, if the preference ordering is not rejected in the first step, we need to check along ``longer chains'' (Theorem \ref{Theorem: Revealed Preference}). That is, whether $\ChoiProb(b|S)\leq \ChoiProb(b|S-c)$ and $\ChoiProb(c|T)\leq \ChoiProb(c|T-a)$ are simultaneously violated for some $S$, $T$ and $c$. If so, the preference ordering is rejected (i.e., incompatible with RAM), while if not then a chain of length three needs to be considered. This process goes on for longer chains until either at some step we are able to reject the preference ordering, or all possibilities are exhausted. In practice, additional comparisons are needed since it is rarely the case that only a specific pair of alternatives is of interest. This algorithm, albeit feasible, can be hard to implement in practice, even when the choice probabilities are known. The fact that $\ChoiProb$ has to be estimated makes the problem even more complicated, since it becomes a sequential multiple hypothesis testing problem.

Another possibility is to employ the J-test approach, which stems from the idea that, given the choice rule, compatibility of a preference is equivalent to the existence of an attention rule satisfying monotonicity. To implement the J-test, one fixes the choice rule (identified/estimated from the data) and the preference ordering (the null hypothesis to be tested), and search the space of all monotonic attention rules and check if Definition \ref{def: random attention representation} applies. The J-test procedure can be quite computationally demanding, due to the fact that the space of attention rules has high dimension. We further discuss the J-test approach in Section SA.4.3 of the Supplemental Appendix and how it is related to our proposed procedure.

One of the main purposes of this section is to provide an equivalent form of identification, which (i) is simple to implement, and (ii) remains statistically valid even when applied using estimated choice rules. For ease of exposition, we rewrite the choice rule $\ChoiProb$ as a long vector $\bChoiProb$, whose elements are simply the probability of each alternative $a\in X$ being chosen from a choice problem $S\in\mathcal{X}$. For example, one can label the choice problems as $S_1$, $S_2$, $\cdots$, the alternatives as $a_1$, $a_2$, $\cdots$, $a_K$, and then the vector $\bChoiProb$ simply consists of $\ChoiProb(a_1|S_1)$, $\ChoiProb(a_2|S_1)$, $\cdots$, $\ChoiProb(a_K|S_1)$, $\ChoiProb(a_1|S_2)$, $\ChoiProb(a_2|S_2)$, etc. See Example \ref{Example: construction of R} for a concrete illustration. 

\begin{thm}[Nonparametric Identification]\label{Theorem: moment inequalities}
Given any preference $\succ$, there exists a unique matrix $\bR_\succ$ such that $\succ\,\in\ParSet_{\ChoiProb}$ if and only if $\bR_\succ\bChoiProb\leq \mathbf{0}$.
\end{thm}

\begin{proof}
Recall that $(\ChoiProb,\succ)$ has a RAM representation if and only if there exists a monotonic and triangular attention rule ${\mu}$ such that $\ChoiProb$ is induced by ${\mu}$ and $\succ$ (Corollary \ref{Lemma: appendix, triangular}). With this fact, we are able to construct the constraint matrix $\bR_\succ$ explicitly, and write it as a product, $\bR\bC_\succ$. The first matrix, $\bR$, consists of constraints on the attention rules, and the second matrix, $\bC_\succ$, maps the choice rule back to a triangular attention rule. 

First consider $\bR$. The only restrictions imposed on attention rules are from the monotonicity assumption  (Assumption \ref{Assumption: monotonicity}). Again, we represent a generic attention rule $\mu$ as a long vector $\bmu$. Then each row of $\bR$ will consist of one ``$+1$'', one ``$-1$'' and 0 otherwise. The product $\bR\bAttFilter$ then corresponds to $\mu(T|S)-\mu(T|S-a)$, for all $S$, $T\subset S$, and $a\in S-T$. That is, we use $\bR\bAttFilter\leq 0$ to represent Assumption \ref{Assumption: monotonicity}. Note that $\bR$ does not depend on any preference. 

Next consider $\bC_\succ$. Given some preference $\succ$ and the choice rule $\ChoiProb$, the only possible triangular attention rule that can be constructed is (see Corollary \ref{Lemma: appendix, triangular} and the proof of Theorem \ref{Theorem: Characterization} in the Appendix)
\[\mu(T | S) = \sum_{k:\ a_{k,\succ}\in S} \Indicator(T=S\cap L_{k,\succ})\cdot  \ChoiProb(a_{k,\succ}|S),\]
where $\{L_{k,\succ}:\ 1\leq k\leq K\}$ are the lower contour sets corresponding to the preference ordering $\succ$ (Definition \ref{Definition: appendix, triangular}). The above defines the mapping $\bC_{\succ}$, and represents the triangular attention rule as a linear combination of the choice probabilities. This mapping depends on the preference/hypothesis because the triangular attention rule depends on the preference/hypothesis.

Along the construction, both $\bR$ and $\bC_\succ$ are unique, hence showing $\bR_\succ$ is uniquely determined by the preference $\succ$. 
\end{proof}

This theorem states that in order to decide whether a preference $\succ$ is compatible with the (identifiable) choice rule $\ChoiProb$, it suffices to check a collection of inequality constraints. In particular, it is no longer necessary to consider the sequential and multiple testing problems mentioned earlier, or numerically searching in the high dimensional space of attention rules. Moreover, as we discuss below, given the large econometric literature on moment inequality testing, many techniques can be adapted when Theorem \ref{Theorem: moment inequalities} is applied to estimated choice rules. An algorithmic construction of the constraint matrix $\bR_{\succ}$ is given in Algorithm \ref{Algorithm: constraints}. 

\begin{algorithm}
\caption{Construction of $\bR_\succ$.}
\begin{algorithmic} 
\REQUIRE Set a preference $\succ$.
\STATE $\bR_\succ \leftarrow$ empty matrix
\FOR{$S$ in $\mathcal{X}$} \FOR{$a$ in $S$} \FOR{$b\prec a$ in $S$}
\STATE $\bR_\succ \leftarrow$ add row corresponding to $\ChoiProb(b|S)-\ChoiProb(b|S-a)\leq 0$.
\ENDFOR\ENDFOR\ENDFOR
\end{algorithmic}\label{Algorithm: constraints} 
\end{algorithm}

As can be seen, the only input needed is the preference $\succ$, which we are interested in testing against. Each row of $\bR_{\succ}$ consists of one ``$+1$'', one ``$-1$'', and $0$ otherwise. The constraint matrix $\bR_{\succ}$ is non-random and does not depend on the estimated choice probabilities, but rather determined by the collection of (fixed, known to the researcher) restrictions on the estimable choice probabilities. Next we compute the number of constraints (i.e. rows) in $\bR_\succ$ for the complete data case (i.e., when all choice problems are observed):  
\[
\text{\#row}(\bR_\succ) = \sum_{S\in\mathcal{X}} \sum_{a,b\in S}\Indicator(b\prec a) = \sum_{S\in\mathcal{X},\ |S|\geq 2} \binom{|S|}{2} = \sum_{k=2}^{K} \binom{K}{k}\binom{k}{2},
\]
where $K = |X|$ is the number of alternatives in the grand set $X$. Not surprisingly, the number of constraints increases very fast with the size of the grand set. However, once the matrix $\bR_{\succ}$ has been constructed for one preference $\succ$, the constraint matrices for other preference orderings can be obtained by column permutations of $\bR_{\succ}$. This is particular useful and saves computation if there are multiple hypotheses to be tested, as the above algorithm only needs to be implemented once.

Finally, we illustrate that, in simple examples, the constraint matrix $\bR_{\succ}$ can be constructed intuitively.
 
\begin{example}[$\bR_{\succ}$ with Three Alternatives]\label{Example: construction of R}
	Assume there are three alternatives, $a$, $b$ and $c$ in $X$, then the choice rule is represented by a vector in $\mathbb{R}^9$:
	\[
	\bChoiProb = 
	\Big[ \underbrace{\pi(\cdot|\{a,b,c\})}_{\in\mathbb{R}^3},\ \underbrace{\pi(\cdot|\{a,b\})}_{\in\mathbb{R}^2},\ \underbrace{\pi(\cdot|\{a,c\})}_{\in\mathbb{R}^2},\ \underbrace{\pi(\cdot|\{b,c\})}_{\in\mathbb{R}^2}  \Big]^\Trans,
	\]
	where for ease of presentation trivial cases such as $\pi(a|\{b,c\})=0$ and $\pi(b|\{b\})=1$ are ignored. Now consider the preference/hypothesis $b\succ a\succ c$. From Lemma \ref{Lemma: Removing Dominating}, we can reject $b\succ a$ if $\pi(a|\{a,b,c\})>\pi(a|\{a,c\})$. Therefore, we need the reverse inequality in $\bR_{b\succ a\succ c}$, given by a row:
	\[
	\begin{bmatrix}
	1& 0& 0& 0& 0& -1&  0& 0& 0
	\end{bmatrix}.
	\]
	Similarly, we will be able to reject $a\succ c$ if $\pi(c|\{a,b,c\})>\pi(c|\{b,c\})$, which implies the following row in the matrix $\bR_{b\succ a\succ c}$:
	\[
	\begin{bmatrix}
	0& 0& 1& 0& 0&  0&  0& 0& -1
	\end{bmatrix}.
	\]
	The row corresponding to $b\succ c$ is
	\[
	\begin{bmatrix}
	0& 0& 1& 0& 0&  0& -1& 0& 0
	\end{bmatrix}.
	\]
	Therefore, for this simple problem with three alternatives, we have the following constraint matrix:
	\[
	\bR_{b\succ a\succ c} = \begin{bmatrix}
	1& 0& 0& 0& 0& -1&  0& 0& 0\\
	0& 0& 1& 0& 0&  0&  0& 0& -1\\
	0& 0& 1& 0& 0&  0& -1& 0& 0
	\end{bmatrix}.
	\]
	
	Note that for problems with more than three alternatives, the above reasoning does not work if implemented na\"{i}vely. Consider the case $X=\{a,b,c,d\}$. Then $b\succ a$ can be rejected by $\pi(a|\{a,b,c,d\})>\pi(a|\{a,c,d\})$, $\pi(a|\{a,b,d\})>\pi(a|\{a,d\})$ or $\pi(a|\{a,b,c\})>\pi(a|\{a,c\})$, which correspond to three rows in the constraint matrix. 
	
	Again we emphasize that, to construct $\bR_{\succ}$, one does not need to know the numerical value of the choice rule $\pi$. The matrix $\bR_{\succ}$ contains restrictions jointly imposed by the monotonicity assumption and the preference $\succ$ that is to be tested.
\end{example}

\subsection{Hypothesis Testing}

Given the identification result in Theorem \ref{Theorem: moment inequalities}, we can replace the identifiable choice rule with its estimate to conduct estimation and inference of the (partially identifiable) preferences. We can also conduct specification testing by evaluating whether the identified set $\ParSet_{\ChoiProb}$ is empty. To proceed, we assume the following data structure.

\begin{Assumption}[DGP]\label{Assumption 2: iid sampling}
The data is a random sample of choice problems $Y_i$ and corresponding choices $y_i$, $\{(y_i, Y_i): y_i\in Y_i,\ 1\leq i\leq N\}$, generated by the underlying choice rule $\Prob[y_i=a| Y_i=S]=\ChoiProb(a|S)$, with $\Prob[Y_i=S]\geq \underline{p}>0$ for all $S\in\mathcal{X}$.
\end{Assumption}

We only assume the data is generated from some choice rule $\ChoiProb$. We allow for the possibility that it is not a RAM, since our identification result permits falsifying the RAM representation: $\ChoiProb$ has a RAM representation if and only if $\ParSet_{\ChoiProb}$ is not empty according to Theorem \ref{Theorem: moment inequalities}. In addition, we only assume that the choice problem $Y_i$ and the corresponding selection $y_i\in Y_i$ are observed for each unit, while the underlying (possibly random) consideration set for the decision maker remains unobserved (i.e., the set $T$ in Definition \ref{Definition: Attention Rule} and Figure \ref{fig:fig1}). For simplicity, we discuss the case of ``complete data'' where all choice problems are potentially observable, but in Section SA.3 and SA.4.4 of the Supplemental Appendix we extend our work to the case of incomplete data.

The estimated choice rule is denoted by $\hat{\ChoiProb}$, 
\[\hat{\ChoiProb}(a|S) 
  = \frac{ \sum_{1\leq i\leq N} \Indicator(y_i=a,\ Y_i=S) }{ \sum_{1\leq i\leq N} \Indicator(Y_i=S) },
  \qquad a\in S, \quad S\in \mathcal{X}.
\]
For convenience, we represent $\hat{\ChoiProb}(\cdot|S)$ by the vector $\hat{\bChoiProb}_S$, and its population counterpart by $\bChoiProb_{S}$. The choice rules are stacked into a long vector, denoted by $\hat{\bChoiProb}$ with the population counterpart $\bChoiProb$. 

We consider Studentized test statistics, and hence we introduce some additional notation. Let $\bsigma_{\ChoiProb,\succ}$ be the standard deviation of $\bR_{\succ}\hat{\bChoiProb}$, and $\hat{\bsigma}_{\succ}$ be its plug-in estimate. That is, 
\[
\bsigma_{\ChoiProb,\succ} = \sqrt{\mathrm{diag}\Big( \bR_{\succ} \bOmega_{\ChoiProb} \bR_{\succ}^\Trans \Big)}
\qquad\text{and}\qquad
\hat{\bsigma}_{\succ} = \sqrt{\mathrm{diag}\Big( \bR_{\succ} \hat{\bOmega}\bR_{\succ}^\Trans \Big)},
\]
where $\mathrm{diag}(\cdot)$ denotes the operator that extracts the diagonal elements of a square matrix, or constructs a diagonal matrix when applied to a vector. Here $\bOmega_{\ChoiProb}$ is block diagonal, with blocks given by $\frac{1}{\Prob[Y_i=S]}\bOmega_{\ChoiProb,S}$, and $\bOmega_{\ChoiProb,S}=\mathrm{diag}(\bChoiProb_{S}) - \bChoiProb_{S}\bChoiProb_{S}^\Trans$. The estimator $\hat{\bOmega}$ is simply constructed by plugging in the estimated choice rule. 

Consider the null hypothesis $\mathsf{H}_0:\,\succ\ \in\ParSet_{\ChoiProb}$. This null hypothesis is useful if the researcher believes a certain preference represents the underlying data generating process. It also serves as the basis for constructing confidence sets or for ranking preferences according to their (im)plausibility in repeated sampling (for example, via employing associated p-values). Given a specific preference, the test statistic is constructed as the maximum of the Studentized, restricted sample choice probabilities:
\[
\Tstat(\succ) = \sqrt{N}\cdot \max\Big\{(\bR_{\succ}\hat{\bChoiProb}) \oslash \hat{\bsigma}_\succ,\ 0\Big\},
\]
where $\oslash$ denotes elementwise division (i.e, Hadamard division) for conformable matrices. The test statistic is the largest element of the vector $\sqrt{N}(\bR_{\succ}\hat{\bChoiProb}) \oslash \hat{\bsigma}_\succ$ if it is positive, or zero otherwise. The reasoning behind such construction is straightforward: if the preference is compatible with the underlying choice rule, then in the population we have $\bR_{\succ}\bChoiProb\leq \mathbf{0}$, meaning that the test statistic, $\Tstat(\succ)$, should not be ``too large.''

Other test statistics have been proposed for testing moment inequalities, and usually the specific choice depends on the context. When many moment inequalities can be potentially violated simultaneously, it is usually preferred to use a statistic based on truncated Euclidean norm. In our problem, however, we expect only a few moment inequalities to be violated, and therefore we prefer to employ $\Tstat(\succ)$. Having said this, the large sample approximation results given in Theorem \ref{Theorem: valid critical values} can be adapted to handle other test statistics commonly encountered in the literature on moment inequalities.

The null hypothesis is rejected whenever the test statistic is ``too large,'' or more precisely, when it exceeds a critical value, which is chosen to guarantee uniform size control in large samples. We describe how this critical value leading to uniformly valid testing procedures is constructed based on simulating from multivariate normal distributions. Our construction employs the Generalized Moment Selection (GMS) approach of \citet{Andrews-Soares_2010_ECMA}; see also \citet{Canay_2010_JoE} and \citet{Bugni_2016_ET} for closely related methods. The literature on moment inequalities testing includes several alternative approaches, some of which we discuss briefly in Section SA.4.5 of the Supplemental Appendix.

To illustrate the intuition behind the construction, first rewrite the test statistic $\Tstat(\succ)$ as the following:
\[
\Tstat(\succ) =  \max\Big\{\big(\bR_{\succ}\sqrt{N}(\hat{\bChoiProb}-\bChoiProb) + \sqrt{N}\bR_{\succ}\bChoiProb\big) \oslash \hat{\bsigma}_\succ,\ 0\Big\}.
\]
By the central limit theorem, the first component $\sqrt{N}(\hat{\bChoiProb}-\bChoiProb)$ is approximately distributed as $\mathcal{N}(\mathbf{0},\ \bOmega_{\ChoiProb})$. The second component, $\bR_{\succ}\bChoiProb$, although unknown, is bounded above by zero under the null hypothesis. Motivated by these observations, we approximate the distribution of $\Tstat(\succ)$ by simulation as follows:
\[\Tstat^\star(\succ)
  = \sqrt{N}\cdot \max\Big\{
  (\bR_{\succ}\bz^\star) \oslash \hat{\bsigma}_\succ 
  + \psi_N(\bR_{\succ}\hat{\bChoiProb},\hat{\bsigma}_\succ),\ 0\Big\}.
\]
Here $\bz^\star$ is a random vector simulated from the distribution $\mathcal{N}(\mathbf{0}, \hat{\bOmega}/N)$, and $\sqrt{N}\psi_N(\bR_{\succ}\hat{\bChoiProb},\hat{\bsigma}_\succ)$ is used to replace the unknown moment conditions $(\sqrt{N}\bR_{\succ}\bChoiProb) \oslash \hat{\bsigma}_\succ$. Several choices of $\psi_N$ have been proposed. One extreme choice is $\psi_N(\cdot)=0$, so that the upper bound $0$ is used to replace the unknown $\bR_\succ\bChoiProb$. Such a choice also delivers uniformly valid inference in large samples, and is usually referred to as ``critical value based on the least favorable model.'' However, for practical purposes it is better to be less conservative. In our implementation we employ
\[
\psi_N(\bR_{\succ}\hat{\bChoiProb},\hat{\bsigma}_\succ) = \frac{1}{\kappa_N}\Big(\bR_{\succ}\hat{\bChoiProb} \oslash \hat{\bsigma}_\succ\Big)_- ,
\]
where $(\mathbf{a})_-=\mathbf{a}\odot\Indicator(\mathbf{a}\leq 0)$, with $\odot$ denoting the Hadamard product, the indicator function $\Indicator(\cdot)$ operating element-wise on the vector $\mathbf{a}$, and $\kappa_N$ diverges slowly. That is, the function $\psi_N(\cdot)$ retains the non-positive elements of $(\bR_{\succ}\hat{\bChoiProb} \oslash \hat{\bsigma}_\succ)/\kappa_N$, since under the null hypothesis all moment conditions are non-positive. We use $\kappa_N=\sqrt{\ln N}$, which turns out to work well in the simulations described in Section \ref{section:simulations}. For other choices of $\psi_N(\cdot)$, see \cite{Andrews-Soares_2010_ECMA}. 

In practice, $M$ simulations are conducted to obtain the simulated statistics $\{\Tstat_m^\star(\succ) :\ 1\leq m\leq M \}$. Then, given some $\alpha\in(0,1)$, the critical value is constructed as 
\[
c_{\alpha}(\succ) = \inf\Big\{ t:\ 
\frac{1}{M}\sum_{m=1}^M \Indicator\big(\Tstat^\star_m(\succ) \leq t \big)\geq 1-\alpha \Big\},
\]
and the null hypothesis $\mathsf{H}_0:\ \succ\,\in\ParSet_{\ChoiProb}$ is rejected if and only if $\Tstat(\succ)>c_{\alpha}(\succ)$. Alternatively, one can compute the p-value as
\[
\text{pVal}(\succ) = \frac{1}{M}\sum_{m=1}^M \Indicator\Big(\Tstat^\star_m(\succ) > \Tstat(\succ)\Big).
\]

To justify the proposed critical values, it is important to address uniformity issues. A testing procedure is (asymptotically) uniform among a class of data generating processes, if the asymptotic size does not exceed the nominal level across this class. Testing procedures that are valid only pointwise but not uniformly may yield bad approximations to the finite sample distribution, because in finite samples the moment inequalities could be close to binding.  The following theorem shows that conducting inference using the critical values above is uniformly valid.

\begin{thm}[Uniformly Valid Testing]\label{Theorem: valid critical values}
Assume Assumption \ref{Assumption 2: iid sampling} holds. Let $\Pi$ be a class of choice rules, and $\succ$ a preference, such that: (i) for each $\ChoiProb\in\Pi$, $\succ\,\in\ParSet_\ChoiProb$; and (ii) $\inf_{\ChoiProb\in\Pi}\min(\bsigma_{\ChoiProb,\succ})>0$. Then,
\[
\limsup_{N\to\infty} \sup_{\ChoiProb\in\Pi}\Prob\left[ \Tstat(\succ) > c_\alpha(\succ) \right] \leq \alpha.
\]
\end{thm}

The proof is postponed to Appendix \ref{Appendix: proof of critical value}. The only requirement is that each moment condition is nondegenerate so that the normalized statistics are well-defined in large samples, but no restrictions on correlations among moment conditions are imposed. 

\subsection{Extensions and Discussion}

We discuss some extensions based on Theorem \ref{Theorem: valid critical values}, including how to construct uniformly valid confidence sets via test inversion, and how to conduct uniformly valid specification testing, both based on testing individual preferences. 

\subsubsection*{Confidence Set}

Given the uniformly valid hypothesis testing procedure already developed in Theorem \ref{Theorem: valid critical values}, we can obtain a uniformly valid confidence set for the (partially) identified preferences by test inversion:
\[
\mathscr{C}(\alpha)  = \Big\{\succ \ :\ \Tstat(\succ) \leq c_\alpha(\succ)  \Big\}.
\]
The resulting confidence set $\mathscr{C}(\alpha)$ exhibits an asymptotic uniform coverage rate of at least $1-\alpha$:
\[
\liminf_{N\to \infty}\inf_{\ChoiProb\in\Pi}\min_{\succ\,\in\ParSet_{\ChoiProb}}\Prob\big[  \succ\,\in \mathscr{C}(\alpha) \big] \geq 1-\alpha.
\]
This inference method offers a uniformly valid confidence set for each member of the partially identified set with pre-specified coverage probability, which is a popular approach in the partial identification literature \citep{Imbens-Manski_2004_ECMA}.

\subsubsection*{Testing Model Compatibility: $\mathsf{H}_0:\mathcal{P}\cap\ParSet_{\ChoiProb}\neq\emptyset$}

Given a collection of preferences, an empirically relevant question is whether any of them is compatible with the data generating process---a basic model specification question. That is, the question is whether the null hypothesis $\mathsf{H}_0:\mathcal{P}\cap\ParSet_{\ChoiProb}\neq\emptyset$ should be rejected. If the null hypothesis is rejected, then certain features shared by the collection of preferences is incompatible with the underlying decision theory (up to Type I error). See \citet*{Bugni-Canay-Shi_2015_JoE}, \citet*{Kaido-Molinari-Stoye_2019_ECMA} and references therein for further discussion of this idea and related methods.

For a concrete example, consider the question that whether $a\succ b$ is compatible with the data generating process. As long as there are more than 2 alternatives in the grand set, a question like this can be accommodated by setting $\mathcal{P} = \{ \succ:\ a\succ b \}$. Rejection of this null hypothesis provides evidence in favor of $b$ being preferred to $a$, (up to Type I error). Of course with more preferences included in the collection, it becomes more difficult to reject the null hypothesis.

The test is based on whether the confidence set intersects with $\mathcal{P}$:
\[
\text{$\mathsf{H}_0$ is rejected} \quad \text{if and only if} \quad \mathscr{C}(\alpha)\cap \mathcal{P}=\emptyset. 
\]
We note that, since $\mathscr{C}(\alpha)$ covers elements in the identified set asymptotically and uniformly with probability $1-\alpha$, the above testing procedure will have uniform size control. Indeed, if $\mathcal{P}\cap \ParSet_{\ChoiProb}\neq \emptyset$, there exists some $\succ\,\in \mathcal{P}\cap \ParSet_{\ChoiProb}$, which will be included in $\mathscr{C}(\alpha)$ with at least $1-\alpha$ probability asymptotically.

One important application of this idea is to set $\mathcal{P}$ as the collection of all possible preferences, which leads to a specification testing. Then, the null hypothesis becomes $\mathsf{H}_0:\ParSet_{\ChoiProb}\neq\emptyset$, and is rejected based on the following rule:
\[
\text{$\mathsf{H}_0$ is rejected} \quad \text{if and only if} \quad \mathscr{C}(\alpha)=\emptyset. 
\]
Rejection in this case implies that at least one of the underlying assumptions is violated, and the data generating process cannot be represented by a RAM (up to Type I error).

\section{Incorporating Additional Restrictions}

Our identification and inference results so far are obtained using RAM only, that is, all empirical content of our revealed preference theory comes from the weak nonparametric Assumption \ref{Assumption: monotonicity}. As mentioned before, our model provides a minimum benchmark for preference revelation, which sometimes may not deliver enough empirical content. However, it is easy to incorporate additional (nonparametric) assumptions in specific settings. In this section, we first illustrate one such possibility, where additional restrictions on the attentional rule are imposed for binary choice problems. This will improve our identification and inference results considerably. We then consider random attention filters, which are one of the motivating examples of monotonic attention rules, and show that in this case there is no identification improvement relative to the baseline RAM.

\subsection{Attentive at Binaries}\label{section:binary}

To motivate our approach, a policy maker may want to conclude that $a$ is revealed to be preferred to $b$ if the decision maker chooses $a$ over $b$ ``frequently enough'' in binary choice problems. ``Frequently enough'' is measured by a constant $\phi \geq 1/2$.\footnote{Even when the policy maker is least cautious,  we need  $\pi(a|\{a,b\}) >  \pi(b|\{a,b\})$ to conclude $a$ is strictly better than $b$. This implies $\pi(a|\{a,b\})>1/2$. Hence $\phi$ must be greater than $1/2$.}  For example, when $\phi=2/3$, it means that choosing $a$ twice more often than choosing $b$ implies $a$ is better than $b$. $\phi$ represents how cautious the policy maker is. Denote by
\[a \mathsf{P}^{\phi} b\quad \text{if and only if}\quad \pi(a|\{a,b\})>\phi.\]
To justify $\mathsf{P}^{\phi}$ as preference revelation, the policy maker inherently assumes that  the decision maker pays attention to the entire set ``frequently enough.'' This is captured by the following assumption on the attention rule. 

\begin{Assumption}[$\phi$-Attentive at Binaries]\label{Assumption: attentive at binaries}
For all $a,b\in\Universe$ and $\phi \geq 1/2$,
\[\AttFilter(\{a,b\}|\{a,b\}) \geq  \frac{1-\phi}{\phi} \max \Big\{\ \AttFilter(\{a\}|\{a,b\})\ ,\ \AttFilter(\{b\}|\{a,b\})\ \Big\}.\]
\end{Assumption}
The quantity $\frac{1-\phi}{\phi}$ is a measure of full attention at binaries. When $\frac{1-\phi}{\phi}=0$ (or $\phi=1$), there is no constraint on $\AttFilter(\{a,b\}|\{a,b\})$. In this case, it is possible that the decision maker only considers singleton consideration sets. When $\frac{1-\phi}{\phi}$ gets larger (or $\phi$ gets smaller), the probability of being fully attentive is strictly positive, which creates room for preference revelation. An alternative way to understand Assumption \ref{Assumption: attentive at binaries} is as follows. Take $\phi = \max\{ \pi(a|\{a,b\}),\pi(b|\{a,b\}) \}$, then $\frac{1-\phi}{\phi} \max \{\AttFilter(\{a\}|\{a,b\})\ ,\ \AttFilter(\{b\}|\{a,b\})\}$ is a strict lower bound on the amount of attention that the decision maker has to pay to both options, for revelation to occur. 

We now illustrate that, under Assumption \ref{Assumption: attentive at binaries}, if $\ChoiProb(a|\{a,b\}) > \phi$ then $a$ is revealed to be preferred to $b$. Let $(\succ,\AttFilter)$ be a RAM representation of $\ChoiProb$ where $\AttFilter$ satisfies  Assumption \ref{Assumption: attentive at binaries}.  First, Assumption \ref{Assumption: attentive at binaries} necessitates  that  $\AttFilter(\{a\}|\{a,b\}) $ cannot be higher than $ \phi$. (To see this, assume $\AttFilter(\{a\}|\{a,b\}) > \phi$. By  Assumption \ref{Assumption: attentive at binaries},  we must have  $\AttFilter(\{a,b\}|\{a,b\}) > 1- \phi$, which is a contradiction.) Then, $\ChoiProb(a|\{a,b\}) >  \phi $ indicates that $a$ is chosen over $b$ whenever the decision maker pays attention to $\{a,b\}$ (revealed preference). Therefore, $a \succ b$. 

\begin{example}[Preference Revelation Without Regularity Violation]
To illustrate the extra identification power of Assumption \ref{Assumption: attentive at binaries}, consider the following stochastic choice with three alternatives and take $\phi=1/2$.
\begin{equation*}\begin{array}{c|cccc}
 \ChoiProb(\cdot|S) & S=\{a,b,c\} & \{a,b\} & \{a,c\} & \{b,c\} \\ \hline
 a & 1/3 & 2/3 & 1/2 & \\
 b & 1/3 & 1/3 &  & 2/3 \\
 c & 1/3 & & 1/2 & 1/3
 \end{array}
\end{equation*}
Note that $\ChoiProb$ satisfies the regularity condition, meaning that there is no preference revelation if only monotonicity (Assumption \ref{Assumption: monotonicity}) is imposed on the attention rule. That is, $\mathsf{P}=\mathsf{P}_{\mathtt{R}}=\emptyset$ (Section \ref{section:revealedpreference}). On the other hand, by utilizing Assumption \ref{Assumption: attentive at binaries}, we can infer the preference completely.  Since $\ChoiProb( a |\{a,b\}) > 1/2$ and $\ChoiProb(b|\{b,c\})>1/2$, we must have $a \mathsf{P}^{\phi} b$ and $b \mathsf{P}^{\phi} c$. Notice that $\ChoiProb( a |\{a,c\}) = 1/2$, hence we cannot directly deduce $a \mathsf{P}^{\phi} c$. Since the underlying preference is transitive, we can conclude that the decision maker prefers $a$ to $c$ as $a\mathsf{P}^{\phi}b$ and $b\mathsf{P}^{\phi}c$, even when $a\mathsf{P}^{\phi}c$ is not directly revealed from her choices. Therefore, the transitive closure of $\mathsf{P}^{\phi}$, denoted by $\mathsf{P}^{\phi}_{\mathtt{R}}$, must also be part of the revealed preference. In this example, note that the same conclusion can be drawn as long as the policy maker assumes $\phi<2/3$. 
\end{example}

To accommodate the revealed preference defined in the original model (i.e., to combine Assumption \ref{Assumption: monotonicity} and \ref{Assumption: attentive at binaries}), we now define the following binary relation:
\begin{align*}
&\ a (\mathsf{P}^{\phi}\cup \mathsf{P}) b\quad \text{if and only if}\\ 
&\ \qquad \quad\text{either (i) for some $S\in \Menus$, $\pi(a|S)>\pi(a|S-b)$, or (ii)  $\pi(a|\{a,b\})>\phi$}.
\end{align*}
$\mathsf{P}^{\phi}\cup \mathsf{P}$ includes our original binary relation $\mathsf{P}$, defined under the monotonic attention restriction (Assumption \ref{Assumption: monotonicity}), as well as $\mathsf{P}^{\phi}$, characterized by the new attentive at binary assumption. Therefore, we can infer more. 

The next theorem shows that acyclicity of $\mathsf{P}^{\phi}\cup \mathsf{P}$, or its transitive closure $(\mathsf{P}^{\phi}\cup \mathsf{P})_{\mathtt{R}}$, provides a simple characterization of the model we consider in this subsection. 

\begin{thm}[Characterization]\label{Theorem:monotonicity + binary attention}
For a given $\phi \geq 1/2$, a choice rule $\ChoiProb$ has a random attention representation $(\succ , \AttFilter )$ where $\AttFilter$  satisfies Assumption \ref{Assumption: monotonicity} and \ref{Assumption: attentive at binaries} if and only if $\mathsf{P}^{\phi}\cup \mathsf{P}$ has no cycle.
\end{thm}

For $ \phi<1$, the model characterized by Theorem \ref{Theorem:monotonicity + binary attention} has a higher predictive power (i.e., empirical content) compared to the model characterized by Theorem \ref{Theorem: Characterization}. Hence the model will fail to retain some of its explanatory power. For example,  Example \ref{example:full_revelation_general} with  $\lambda_a, \lambda_b, \lambda_c < 1-\phi$ is outside of the model given here. 

Under the assumption $\phi=1/2$ and  $\pi(a|\{a,b\}) \neq 1/2$ for all $a, b$,  Theorem \ref{Theorem:monotonicity + binary attention} yields that our framework reveals a unique preference while it allows regularity violation. 

\begin{remark}[Acyclic Stochastic Transitivity]
We would like to highlight a close connection between acyclicity of $\mathsf{P}^{\phi}\cup \mathsf{P}$ and the acyclic stochastic transitivity (AST) introduced by \cite{Fishburn_1973}. The model characterized by Theorem \ref{Theorem:monotonicity + binary attention} satisfies a weaker version of AST:
\[\ChoiProb( a_1 |\{a_1,a_2\}) > \phi,   \cdots,  \ChoiProb(a_{k-1}|\{a_{k-1},a_k\})>\phi  \text{ imply } \ChoiProb(a_1|\{a_1,a_k\}) \leq \phi.\]
We call this condition $\phi$-acyclic stochastic transitivity ($\phi$-AST). Note that $\frac{1}{2}$-AST is equivalent to AST. If we only consider binary choice probabilities, acyclicity of $\mathsf{P}^{\phi}\cup \mathsf{P}$ becomes equivalent to $\phi$-AST. Otherwise, our condition is stronger than $\phi$-AST. 
\end{remark}

Now we discuss the econometric implementation. Recall from Section \ref{section:econometrics} that, to test if a specific preference ordering is compatible with the observed (identifiable) choice rule and the monotonicity assumption, we first construct a triangular attention rule and then test whether the triangular attention rule satisfies Assumption \ref{Assumption: monotonicity}. This is formally justified in the proof of Theorem \ref{Theorem: moment inequalities}. 

This line of reasoning can be naturally extended to accommodate Assumption \ref{Assumption: attentive at binaries} in our econometric implementation. Again, the researcher constructs a triangular attention rule based on a specific preference ordering and the identifiable choice rule. She then tests whether the triangular attention rule satisfies Assumption \ref{Assumption: monotonicity} and \ref{Assumption: attentive at binaries}. This is formally justified in the proof of Theorem \ref{Theorem:monotonicity + binary attention}. For testing, only minor changes have to be made when constructing the matrix $\mathbf{R}_\succ$. The precise construction is given in Algorithm \ref{Algorithm: constraints binary}.

\begin{algorithm}
\caption{Construction of $\bR_\succ$.}
\begin{algorithmic} 
\REQUIRE Set a preference $\succ$.
\STATE $\bR_\succ \leftarrow$ empty matrix
\FOR{$S$ in $\mathcal{X}$} 
\FOR{$a$ in $S$} \FOR{$b\prec a$ in $S$}
\STATE $\bR_\succ \leftarrow$ add row corresponding to $\ChoiProb(b|S)-\ChoiProb(b|S-a)\leq 0$.
\ENDFOR\ENDFOR
\IF{$S=\{a,b\}$ is binary and $b\prec a$}
\STATE $\bR_\succ \leftarrow$ add row corresponding to $\frac{1-\phi}{\phi}\ChoiProb(b|S)-\ChoiProb(a|S)\leq 0$
\ENDIF
\ENDFOR

\end{algorithmic}\label{Algorithm: constraints binary} 
\end{algorithm}

We now revisit Example \ref{Example: construction of R} to illustrate what additional (identifying) restrictions are imposed by Assumption \ref{Assumption: attentive at binaries}.

\begin{example}[Example \ref{Example: construction of R}, Continued]\label{Example: construction of R 2}
	Recall that there are three alternatives, $a$, $b$ and $c$ in $X$, and the choice rule is represented by a vector in $\mathbb{R}^9$. For the preference $b\succ a\succ c$, the matrix $\bR_{b\succ a\succ c}$ contains three restrictions if only Assumption \ref{Assumption: monotonicity} is imposed. With our new restriction on the attention rule for binary choice problems, $\bR_{b\succ a\succ c}$ is further augmented:
	\[
	\bR_{b\succ a\succ c} = \begin{bmatrix}
	1& 0& 0& 0& 0& -1&  0& 0& 0\\
	0& 0& 1& 0& 0&  0&  0& 0& -1\\
	0& 0& 1& 0& 0&  0& -1& 0& 0\\
	\hline
	0& 0& 0& \frac{1-\phi}{\phi}& -1&  0&  0& 0& 0\\
	0& 0& 0& 0& 0&  -1&  \frac{1-\phi}{\phi}& 0& 0\\
	0& 0& 0& 0& 0&  0&  0& -1& \frac{1-\phi}{\phi}
	\end{bmatrix},
	\]	
where the first three rows corresponding to restrictions imposed by Assumption \ref{Assumption: monotonicity}, and the last three rows captures our new Assumption \ref{Assumption: attentive at binaries}. 
\end{example}

Assumption \ref{Assumption: attentive at binaries} improves considerably the empirical content of our benchmark RAM (Assumption \ref{Assumption: monotonicity}). However, this assumption is just one of many possible assumptions that could be used in addition to our general RAM. The main takeaway is that our proposed RAM offers a baseline for specific, empirically relevant models of choice under random limited attention. In Section \ref{section:simulations} we compare using simulations the empirical content of our benchmark RAM, which employs only Assumption \ref{Assumption: monotonicity}, and the model that incorporates Assumption \ref{Assumption: attentive at binaries} as well.

\subsection{Random Attention Filter \label{section:RA filter}}

We now consider random attention filters, which are one of the motivating examples of monotonic attention rules. Recall from Section \ref{subsection:monotone} that an attention filter is a deterministic attention rule that satisfies Assumption \ref{Assumption: monotonicity}, and a random attention filter is a convex combination of attention filters, and hence a random attention filter will also satisfy Assumption \ref{Assumption: monotonicity}. For example, the same individual might be utilizing different platforms during her Internet search. Each platform yields a different attention filter, and the usage frequency of each platform is equal to the weight of that attention filter. Random attention filters also give a different interpretation of our model. 

The set of all random attention filters is a strict subset of monotonic attention rules. This is not surprising given that the class of monotonic attention rules is very large. What is (arguably) surprising is the following fact that we are able to show: if $(\pi,\succ,\mu)$ is a RAM with $\mu$ being a monotonic attention rule, there exists a random attention filter $\mu'$ such that $(\pi,\succ,\mu')$ is still a RAM (see Remark \ref{Remark: Triangular Random Attention Filter Representation}). Before presenting this result, however, we observe that $\mu$ and $\mu'$ need not be the same, which means that there are monotonic attention rules that cannot be written as a convex combination of attention filters. 

\begin{example}\label{example: ext_attn} Let $X=\{a_1,a_2,a_3,a_4\}$. Consider a monotonic attention rule $\AttFilter$ such that (i) $\AttFilter(T|S)$ is either $0$ or $0.5$, (ii) $\AttFilter(T|S)=0$ if $|T|>1$, and (iii) if $\AttFilter(\{a_j\}|S)=0$ and $k<j$ then $\AttFilter(\{a_k\}|S)=0$. Then we must have $\AttFilter(\{a_3\}|\{a_1,a_2,a_3,a_4\})=\AttFilter(\{a_4\}|\{a_1,a_2,a_3,a_4\})=0.5$. We now show that $\AttFilter$ is not a random attention filter. 

Suppose $\AttFilter$ can be written as a linear combination of attention filters. Then $\AttFilter(\{a_3\}|\{a_1,a_2,a_3,a_4\})=\AttFilter(\{a_4\}|\{a_1,a_2,a_3,a_4\})=0.5$ implies that only attention filters for which $\Gamma(\{a_1,a_2,a_3,a_4\})=\{a_3\}$ or $\Gamma(\{a_1,a_2,a_3,a_4\})=\{a_4\}$ must be assigned positive probability. On the other hand, $\AttFilter(\{a_2\}|\{a_1,a_2,a_3\})=0.5$ and $\AttFilter(\{a_2\}|\{a_1,a_2,a_4\})=0.5$ imply that for all $\Gamma$ which are assigned positive probability $\Gamma(\{a_1,a_2,a_3\})=\{a_2\}$ whenever $\Gamma(\{a_1,a_2,a_3,a_4\})=\{a_4\}$ and $\Gamma(\{a_1,a_2,a_4\})=\{a_2\}$ whenever $\Gamma(\{a_1,a_2,a_3,a_4\})=\{a_3\}$. To see this, notice that the attention filter property implies $\Gamma(\{a_1,a_2,a_3\})=\{a_3\}$ for all $\Gamma$ with $\Gamma(\{a_1,a_2,a_3,a_4\})=\{a_3\}$ and $\Gamma(\{a_1,a_2,a_4\})=\{a_4\}$ for all $\Gamma$ with $\Gamma(\{a_1,a_2,a_3,a_4\})=\{a_4\}$. But then it must be the case that $\Gamma(\{a_1,a_2\})=\{a_2\}$ for all $\Gamma$ which are assigned positive probability, or that $\AttFilter(\{a_2\}|\{a_1,a_2\})=1$, a contradiction.
\end{example}

We now show that if we restrict our attention to a certain type of monotonic attention rules, then we can show that within that class every attention rule is a random attention filter (i.e., convex combination of deterministic attention filters). Let $\mathcal{MT}({\succ})$ denote the set of all attention rules that are both monotonic (Assumption \ref{Assumption: monotonicity}) and triangular with respect to $\succ$ (Definition \ref{Definition: appendix, triangular} in the Appendix), and let $\mathcal{AF}({\succ})$ denote all attention filters that are triangular with respect to $\succ$. We are now ready to state the main result of this section. 

\begin{thm}[Random Attention Filter]
\label{Theorem: Random Model}
For any $\AttFilter \in \mathcal{MT}({\succ})$, there exists a probability law $\psi$ on $\mathcal{AF}(\succ)$ such that for any $S\in \mathcal{X}$ and $T\subset S$
$$\AttFilter(T|S)=\sum_{\Gamma\in \mathcal{AF}({\succ})}\Indicator(\Gamma(S)=T)\cdot \psi(\Gamma).$$
\end{thm}

\begin{remark}[Triangular Random Attention Filter Representation]\label{Remark: Triangular Random Attention Filter Representation}
Combining this theorem and Corollary \ref{Lemma: appendix, triangular} in the Appendix, we easily reach the following conclusion: If $\ChoiProb$ has a random attention representation $(\succ,\AttFilter)$, then there exists a triangular random attention filter $\mu'$ such that $(\succ,\AttFilter')$ also represents $\ChoiProb$. 
\end{remark}

The proof of Theorem \ref{Theorem: Random Model} is long and hence left to Appendix, but here we provide a sketch of it. First, $\mathcal{MT}({\succ})$ is a compact and convex set, and thus the above theorem can alternatively be stated as follows: The set of extreme points of $\mathcal{MT}({\succ})$ is $\mathcal{AF}(\succ)$. (An attention rule $\AttFilter\in \mathcal{MT}({\succ})$ is an extreme point of $\mathcal{MT}({\succ})$ if it cannot be written as a nondegenerate convex combination of any $\AttFilter',\AttFilter''\in \mathcal{MT}({\succ})$.) Then, Minkowski's Theorem guarantees that every element of $\mathcal{MT}({\succ})$ lies in the convex hull of $\mathcal{AF}(\succ)$.

Obviously, every element of $\mathcal{AF}(\succ)$ is an extreme point of $\mathcal{MT}({\succ})$. We then show that non-deterministic triangular attention rules cannot be extreme points, i.e. given any $\AttFilter \in \mathcal{MT}({\succ}) - \mathcal{AF}(\succ)$ we can construct $\AttFilter',\AttFilter''\in \mathcal{MT}({\succ})$ such that $\AttFilter=\frac{1}{2}\AttFilter'+\frac{1}{2}\AttFilter''$. The key step is to show that both $\AttFilter'$ and $\AttFilter''$ that we construct are monotonic. After this step, we have shown that no $\AttFilter\in \mathcal{MT}({\succ}) - \mathcal{AF}(\succ)$ can be an extreme point, thus concluding the proof.

\section{Simulation Evidence}\label{section:simulations}

This section gives a summary of a simulation study conducted to assess the finite sample properties of our proposed econometric methods. We consider a class of logit attention rules indexed by $\varsigma$:
\begin{align*}
    \AttFilter_{\varsigma}(T|S) &= \frac{w_{T,\varsigma}}{\sum_{T'\subset S} w_{T',\varsigma}},\qquad w_{T,\varsigma} = |T|^{\varsigma},
\end{align*}
where $|T|$ is the cardinality of $T$. Thus the decision maker pays more attention to larger sets if $\varsigma>0$, and pays more attention to smaller sets if $\varsigma<0$. When $\varsigma$ is very small (negative and large in absolute magnitude), the decision maker almost always pays attention to singleton sets, hence nothing will be learned about the underlying preference from the choice data.

Other details on the data generating process used in the simulation study are as follows. First, the grand set $X$ consists of five alternatives, $a_1$, $a_2$, $a_3$, $a_4$ and $a_5$. Without loss of generality, assume the underlying preference is $a_1\succ a_2\succ a_3\succ a_4\succ a_5$. Second, the data consists of choice problems of size two, three, four and five. That is, there are in total $26$ choice problems. Third, given a specific realization of $Y_i$, a consideration set is generated from the logit attention model with $\varsigma=2$, after which the choice $y_i$ is determined by the aforementioned preference. We also report simulation evidence for $\varsigma\in\{0,1\}$ in the Supplemental Appendix. Finally, the observed data is a random sample $\{(y_i,Y_i):\ 1\leq i\leq N\}$, where the effective sample size can be $50$, $100$, $200$, $300$ and $400$. (Effective sample size refers to the number of observations for each choice problem. Because there are $26$ choice problems, the overall sample size is $N\in\{1300, 2600, 5200, 7800, 10400\}$.)

For inference, we employ the procedure introduced in Section \ref{section:econometrics} and test whether a specific preference ordering is compatible with the basic RAM (Assumption \ref{Assumption: monotonicity}). We also incorporate the attentive at binaries assumption introduced in Section \ref{section:binary}. Recall from Assumption \ref{Assumption: attentive at binaries} that $(1-\phi)/\phi$ is a measure of full attention at binaries, and specifying a larger value (i.e., a smaller value of $\phi$) implies that the researcher is more willing to draw information from binary comparisons. Note that with $\phi=1$, imposing Assumption \ref{Assumption: attentive at binaries} does not bring any additional identification power. Before proceeding, we list five hypotheses (preference orderings), and whether they are compatible with our RAM and specific values of $\phi$. 

\begin{center}
\begin{tabular}{c|c|cccccccccc}
\hline
&\multicolumn{11}{c}{$\phi$}\\ 
\cline{2-12}
& 1 & .95 & .90 & .85 & .80 & .75 & .70 & .65 & .60 & .55 & .50\\ \hline
$\mathsf{H}_{0,1}:a_1\succ a_2\succ a_3\succ a_4\succ a_5$ &$\checkmark$&$\checkmark$&$\checkmark$&$\checkmark$&$\checkmark$&$\checkmark$&$\checkmark$&$\checkmark$&$\checkmark$&$\checkmark$&$\checkmark$ \\
$\mathsf{H}_{0,2}:a_2\succ a_3\succ a_4\succ a_5\succ a_1$ &$\checkmark$&$\checkmark$&$\checkmark$&$\checkmark$&$\times$&$\times$&$\times$&$\times$&$\times$&$\times$&$\times$ \\
$\mathsf{H}_{0,3}:a_3\succ a_4\succ a_5\succ a_2\succ a_1$ &$\times$&$\times$&$\times$&$\times$&$\times$&$\times$&$\times$&$\times$&$\times$&$\times$&$\times$ \\
$\mathsf{H}_{0,4}:a_4\succ a_5\succ a_3\succ a_2\succ a_1$ &$\times$&$\times$&$\times$&$\times$&$\times$&$\times$&$\times$&$\times$&$\times$&$\times$&$\times$ \\
$\mathsf{H}_{0,5}:a_5\succ a_4\succ a_3\succ a_2\succ a_1$ &$\times$&$\times$&$\times$&$\times$&$\times$&$\times$&$\times$&$\times$&$\times$&$\times$&$\times$ \\ \hline
\end{tabular}
\end{center}

As can be seen, $\mathsf{H}_{0,1}$ always belongs to the identified set of preferences, as it is the preference ordering used in the underlying data generating process. $\mathsf{H}_{0,2}$, however, may or may not belong to the identified set depending on the value of $\phi$: with $\phi$ close to 0.5, the researcher is confident enough using information from binary comparisons, and she will be able to reject this hypothesis; for $\phi$ close to 1, Assumption \ref{Assumption: attentive at binaries} no longer brings too much additional identification power beyond the monotonic attention assumption, and monotonic attention alone is not strong enough to reject this hypothesis. Indeed, with $\phi=1$ (i.e., Assumption \ref{Assumption: monotonicity} alone), the set of identified preference is $\{ \succ: a_2\succ a_3\succ a_4\succ a_5\}$, which contains $\mathsf{H}_{0,2}$. The other three hypotheses, $\mathsf{H}_{0,3}$, $\mathsf{H}_{0,4}$ and $\mathsf{H}_{0,5}$, do not belong to the identified set even with $\phi=1$.

Overall, our simulation has 5 (different $N$) $\times$ 5 (different preference orderings) $\times$ 11 (different $\phi$) $=275$ designs. For each design, 5,000 simulation repetitions are used, and the five null hypotheses are tested using our proposed method at the 5\% nominal level. Simulation results are summarized in Figure \ref{fig:fig2}. 

\begin{figure}[!ht]
\centering
\subfloat[$\mathsf{H}_{0,1}:a_1\succ a_2\succ a_3\succ a_4\succ a_5$]{
\resizebox{0.4\columnwidth}{!}{\includegraphics{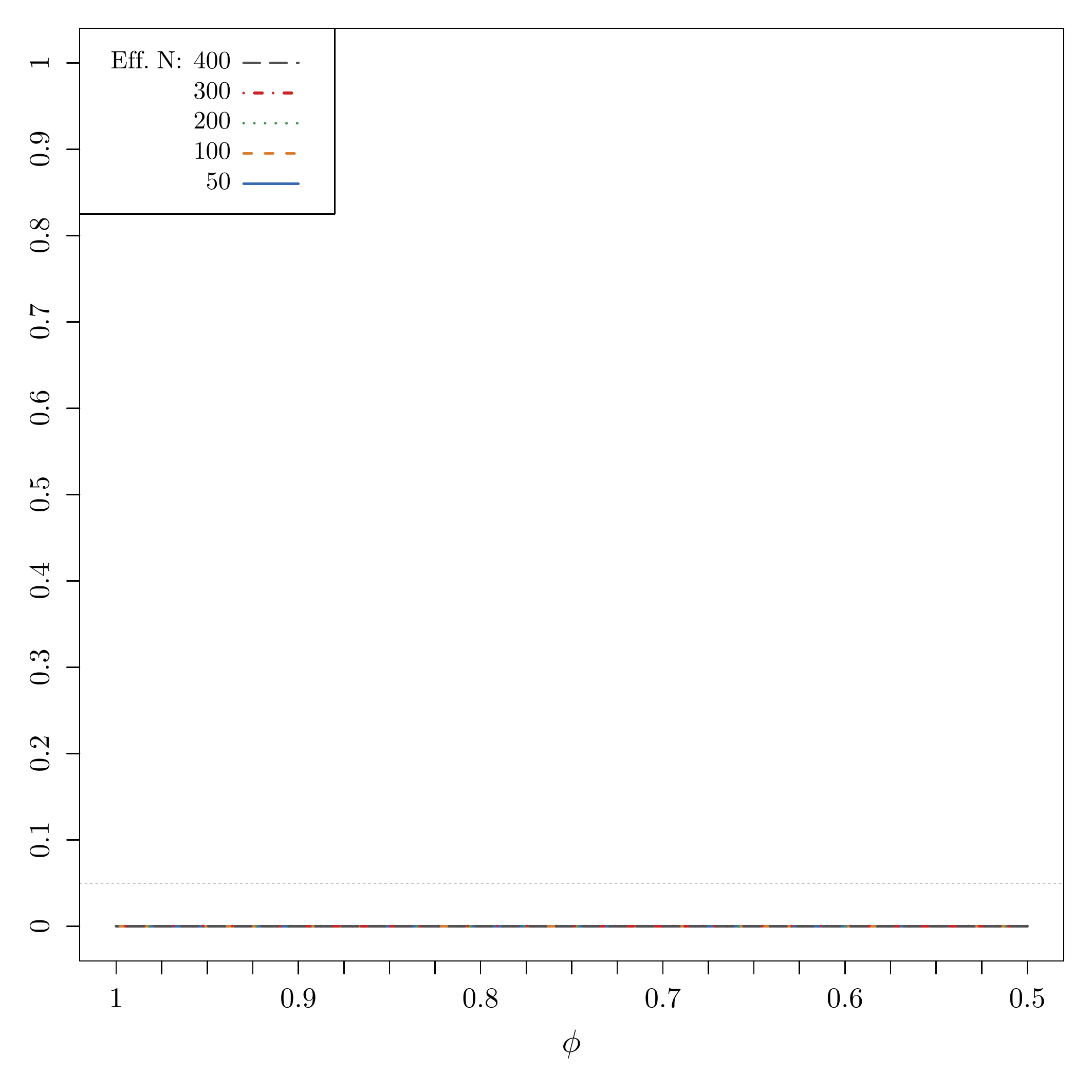}}
}
\subfloat[$\mathsf{H}_{0,2}:a_2\succ a_3\succ a_4\succ a_5\succ a_1$]{
\resizebox{0.4\columnwidth}{!}{\includegraphics{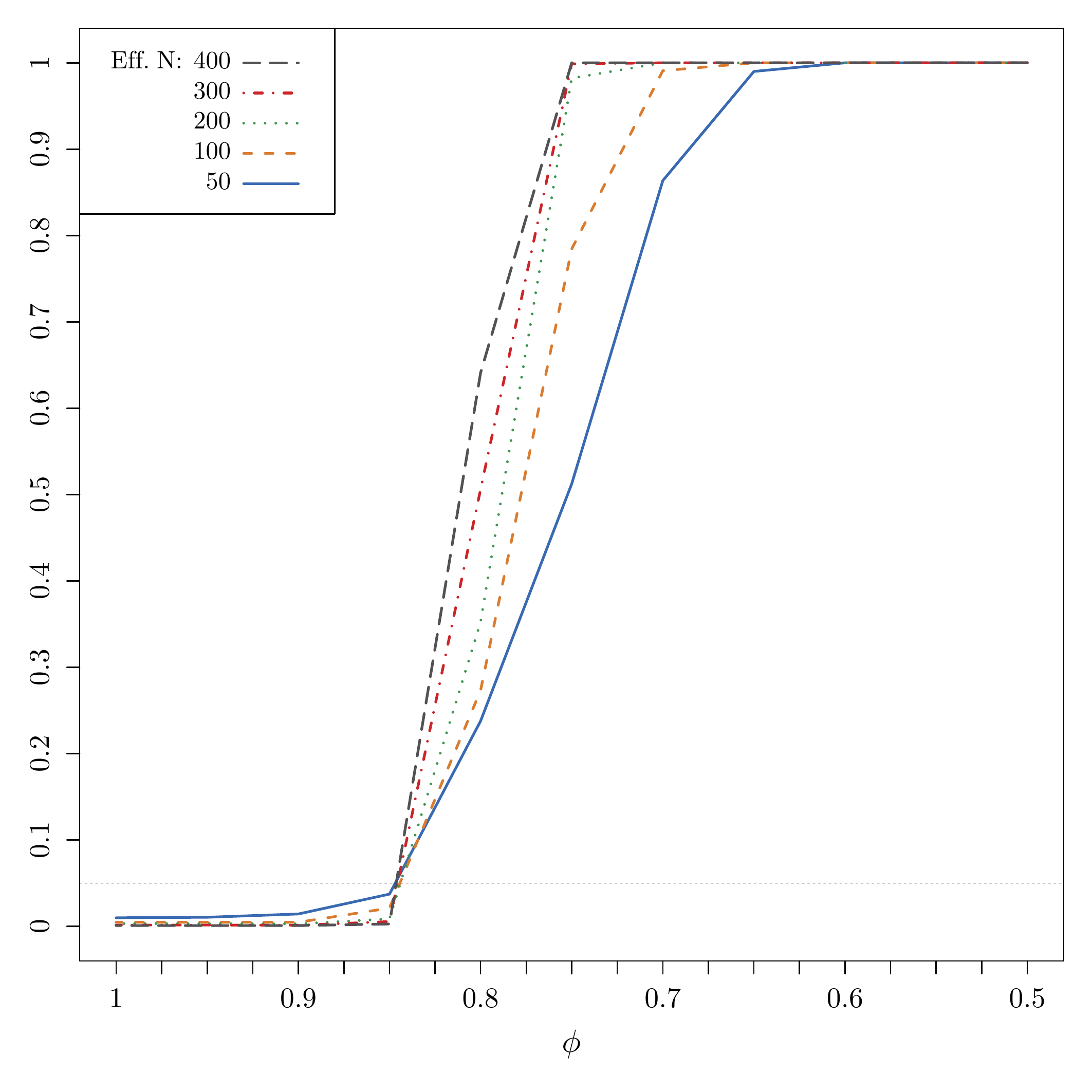}}
}\\
\subfloat[$\mathsf{H}_{0,3}:a_3\succ a_4\succ a_5\succ a_2\succ a_1$]{
\resizebox{0.4\columnwidth}{!}{\includegraphics{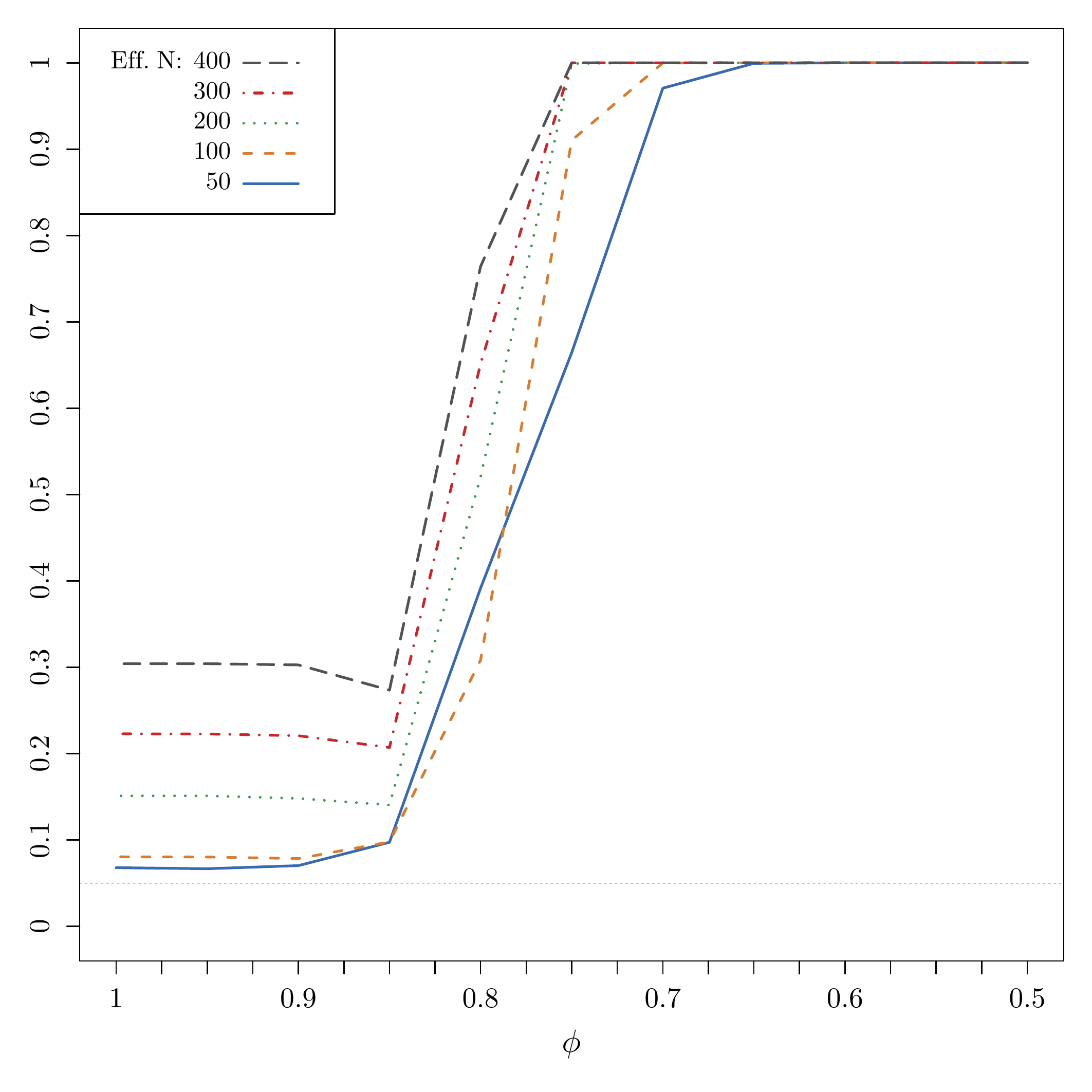}}
}
\subfloat[$\mathsf{H}_{0,4}:a_4\succ a_5\succ a_3\succ a_2\succ a_1$]{
\resizebox{0.4\columnwidth}{!}{\includegraphics{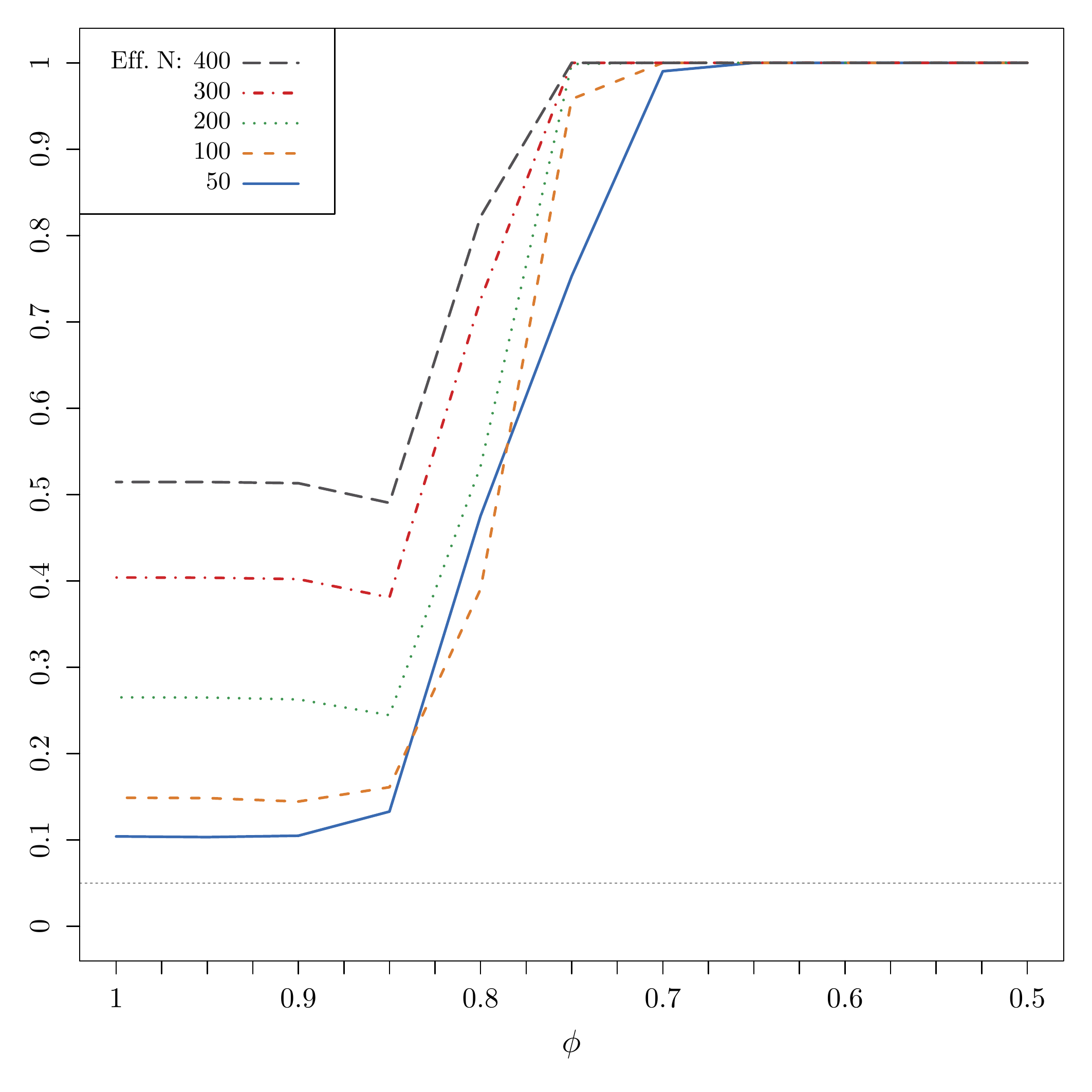}}
}\\
\subfloat[$\mathsf{H}_{0,5}:a_5\succ a_4\succ a_3\succ a_2\succ a_1$]{
\resizebox{0.4\columnwidth}{!}{\includegraphics{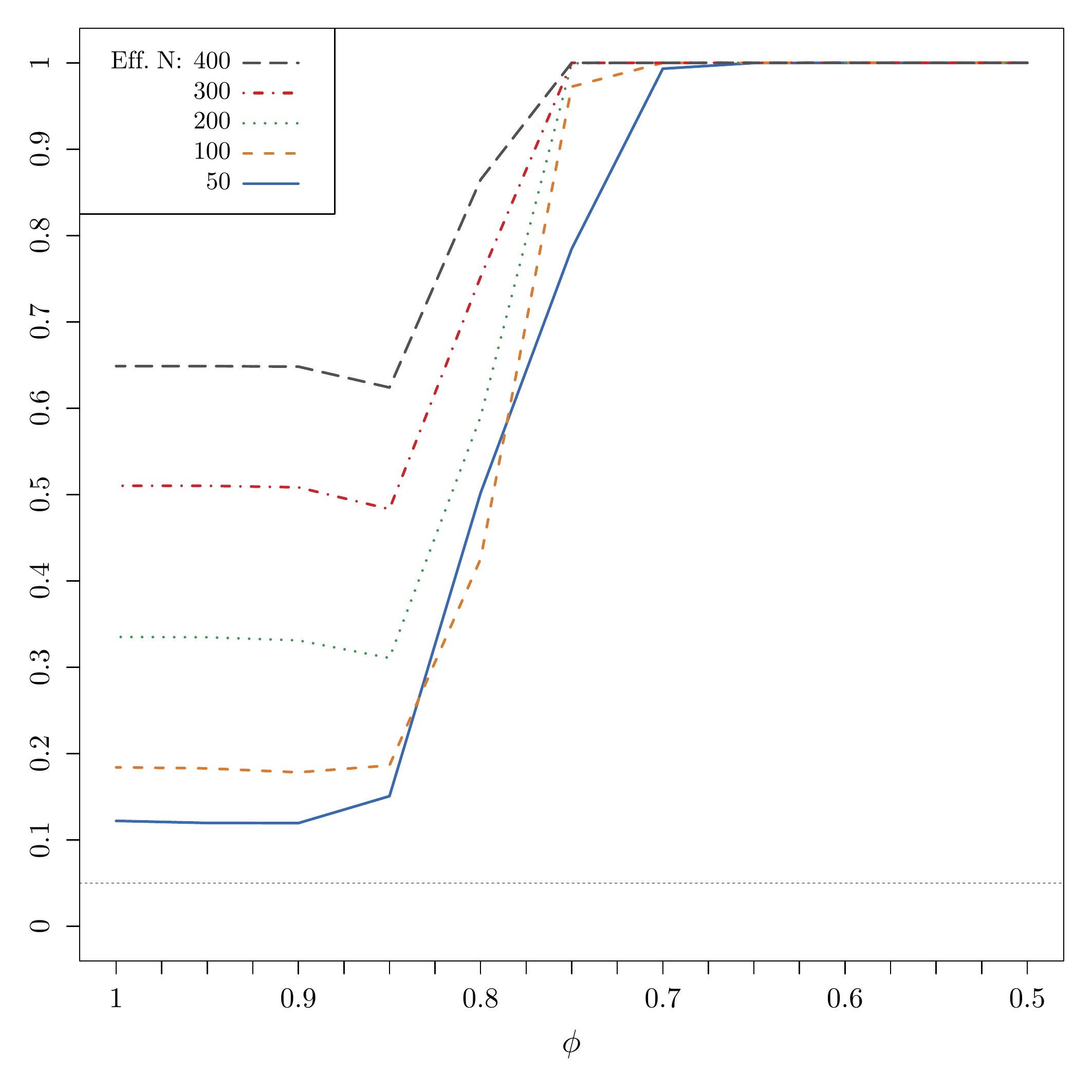}}
}\ \ 
\begin{minipage}{0.4\columnwidth}\footnotesize
\vskip-0.17\textheight Shown in the figure are empirical rejection probabilities testing the five null hypothesis through 5,000 simulations, with nominal size $0.05$. Logit attention rule with $\varsigma=2$ is used, as described in the text. For each simulation repetition, five effective sample sizes are considered $50$, $100$, $200$, $300$ and $400$. 
\end{minipage}
\caption{Empirical Rejection Probabilities}\label{fig:fig2}
\end{figure}

We first focus on $\mathsf{H}_{0,1}$ (panel a). As this preference ordering is compatible with our RAM, one should expect the rejection probability to be less than the nominal level. Indeed, the rejection probability is far below 0.05: this illustrates a generic feature of any (reasonable) procedure for testing moment inequalities---to maintain uniform asymptotic size control, empirical rejection probability is below the nominal level when the inequalities are far from binding. Next consider $\mathsf{H}_{0,2}$ (panel b). For $\phi$ lager than 0.85, the rejection probability is below the nominal size, which is compatible with our theory, because this preference belongs to the identified set when only Assumption \ref{Assumption: monotonicity} is imposed. With smaller $\phi$, the researcher relies more heavily on information from binary comparisons/choice problems, and she is able to reject this hypothesis much more frequently. This demonstrates how additional restrictions on the attention rule can be easily accommodated by our basic RAM, which in turn can bring additional identification power. The other three hypotheses (panel c, d and e) are not compatible with our RAM, and we do see that the rejection probability is much larger than the nominal size even for $\phi=1$, showing that even our basic RAM has non-trivial empirical content in this case.

\section{Conclusion\label{section:conclusion}}

We introduced a limited attention model allowing for a general class of monotonic (and possibly stochastic) attention rules, which we called a Random Attention Model (RAM). We showed that this model nests several important recent contributions in both economic theory and econometrics, in addition to other classical results from decision theory. Using our RAM, we obtained a testable theory of revealed preferences and developed partial identification results for the decision maker's unobserved strict preference ordering. Our results included a precise constructive characterization of the identified set for preferences, as well as uniformly valid inference methods based on that characterization. Furthermore, we showed how additional nonparametric restriction can be easily incorporated into RAM to obtain tigher empirical implications, and more powerful accompying econometric procedures. We found good finite sample performance of our econometric methods in a simulation experiment. Last but not least, we provide the general-purpose \texttt{R} software package \texttt{ramchoice}, which allows other researchers to easily employ our econometric methods in empirical applications.

\bigskip

\appendix

\renewcommand{\thedefn}{A.\arabic{defn}}\setcounter{defn}{0}
\renewcommand{\thethm}{A.\arabic{thm}}\setcounter{thm}{0}
\renewcommand{\thelem}{A.\arabic{lem}}\setcounter{lem}{0}
\renewcommand{\thecoro}{A.\arabic{coro}}\setcounter{coro}{0}
\renewcommand{\theprop}{A.\arabic{prop}}\setcounter{prop}{0}
\renewcommand{\theexampletmp}{A.\arabic{exampletmp}}\setcounter{exampletmp}{0}
\renewcommand{\theremarktmp}{A.\arabic{remarktmp}}\setcounter{remarktmp}{0}
\renewcommand{\thesubsection}{A.\arabic{subsection}}

\singlespacing

\section*{Appendix: Omitted Proofs\label{appendix: proofs}}

This appendix collects proofs that are omitted from the main text to improve the exposition. 

\subsection{Proof of Theorem \ref{Theorem: Characterization}}

Suppose $\ChoiProb$ has a random attention representation $(\succ , \AttFilter )$. Then Lemma \ref{Lemma: Removing Dominating} implies that $\succ $ must include $\mathsf{P}$ so $\mathsf{P}$ must be acyclic. 

For the other direction, suppose that $\mathsf{P}$ has no cycle. Pick any preference $\succ$ that includes $\mathsf{P}_\mathtt{R}$ and enumerate all alternatives with respect to $\succ$: $a_{1,\succ} \succ a_{2,\succ} \succ \cdots \succ a_{\dimUniverse, \succ}$. Let $\{L_{k,\succ}:\ 1\leq k\leq K\}$ be the corresponding lower contour sets (Definition \ref{Definition: appendix, triangular}). Then we specify $\tilde{\AttFilter}$ as
\begin{align*}
\tilde{\mu}(T|S) &= \begin{cases}
\pi(a_{k,\succ}|S) &\qquad \text{if $a_{k,\succ}\in S$ and $T = L_{k,\succ}\cap S$}\\
0 &\qquad \text{otherwise}
\end{cases}.
\end{align*}
It is trivial to verify that $(\succ,\tilde{\mu})$ represents $\pi$, since $(\succ,\tilde{\mu})$ induces the following choice rule:
\begin{align*}
\sum_{T\subset S} \Indicator[\text{$a$ is $\succ$-best in $T$}]\tilde{\mu}(T|S) &= \sum_{a_{k,\succ}\in S} \Indicator[\text{$a$ is $\succ$-best in $L_{k,\succ}\cap S$}]\tilde{\mu}(L_{k,\succ}\cap S|S)\\
&= \sum_{a_{k,\succ}\in S} \Indicator[\text{$a$ is $\succ$-best in $L_{k,\succ}\cap S$}]\pi(a_{k,\succ}|S)\\
&= \sum_{a_{k,\succ}\in S} \Indicator[a = a_{k,\succ}]\pi(a_{k,\succ}|S)\\
&= \pi(a|S),
\end{align*}
which is the same as $\pi$. For the first equality, we use the definition that a triangular attention rule only puts weights on lower contour sets; for the second equality, we apply the definition/construction of $\tilde{\mu}$; the third equality follows from the definition of lower contour sets. 

Now we verify that $\tilde{\mu}$ satisfies Assumption \ref{Assumption: monotonicity}. Assume this is not the case, then it means there exist some $S$, $a_{k,\succ},a_{\ell,\succ} \in S$, such that (i) $L_{k,\succ}\cap S = L_{k,\succ}\cap (S - a_{\ell,\succ})$, and (ii) $\tilde{\mu}(L_{k,\succ}\cap S|S) > \tilde{\mu}(L_{k,\succ}\cap (S - a_{\ell,\succ})|S - a_{\ell,\succ})$. By the definition of lower contour sets, (i) implies $a_{\ell,\succ}\succ a_{k,\succ}$. Then (ii) implies
\begin{align*}
\tilde{\mu}(L_{k,\succ}\cap S|S) &= \pi(a_{k\succ}|S) > \tilde{\mu}(L_{k,\succ}\cap (S - a_{\ell,\succ})|S - a_{\ell,\succ}) = \pi(a_{k,\succ}|S - a_{\ell,\succ}).
\end{align*}
The above, however, implies that $a_{k,\succ}\mathsf{P} a_{\ell,\succ}$, which contradicts the implication of (i) that $a_{\ell,\succ}\succ a_{k,\succ}$. This closes the proof.

\begin{remark}
The previous proof has a nice implication that, a choice rule can be represented by a monotonic attention rule if and only if it can also be represented by a monotonic triangular attention rule. Formally, if $\ChoiProb$ has a random attention representation, $(\succ,\AttFilter)$, then $(\succ,\tilde{\AttFilter})$ also represents $\ChoiProb$ where $\tilde{\AttFilter}$ is monotonic and triangular with respect to $\succ$. Hence, we can focus on monotonic triangular attention rules without loss of generality. This is formally summarized in Corollary \ref{Lemma: appendix, triangular}.
\end{remark}

\subsection{Proof of Theorem \ref{Theorem: valid critical values}}\label{Appendix: proof of critical value}

See Section SA.4.1 of the Supplemental Appendix. 

\subsection{Proof of Theorem \ref{Theorem:monotonicity + binary attention}}

The ``only if'' part is trivial and is omitted. We illustrate the ``if'' part. Assume that $\mathsf{P}^{\phi}\cup \mathsf{P}$ has no cycle (or equivalently, its transitive closure $(\mathsf{P}^{\phi}\cup \mathsf{P})_{\mathtt{R}}$ has no cycle), then there exists some preference ordering that embeds $\mathsf{P}^{\phi}\cup \mathsf{P}$. Fix one such preference $\succ$. With the same argument used in the proof of Theorem \ref{Theorem: Characterization}, we can construct a triangular attention rule $\mu(T|S)$ and show that it satisfies Assumption \ref{Assumption: monotonicity}. 

We then show that $\mu(T|S)$ satisfies Assumption \ref{Assumption: attentive at binaries}. Take binary $S=\{a,b\}$ and assume without loss of generality that $a\succ b$. Then $\mu(\{a,b\}|\{a,b\})=\pi(a|S)$ and $\mu(\{b\}|\{a,b\})=\pi(b|S)$. Violation of Assumption \ref{Assumption: attentive at binaries} implies $\pi(a|\{a,b\})<\frac{1-\phi}{\phi}\pi(b|\{a,b\})$, and equivalently, $\pi(b|\{a,b\})>\phi$. This means that $b\mathsf{P}^{\phi}a$, which violates our definition of $\succ$. 

\subsection{Proof of Theorem \ref{Theorem: Random Model}}

We show that the set of extreme points of $\mathcal{MT}(\succ)$ is $\mathcal{AF}(\succ)$. Clearly, any $\Gamma\in \mathcal{AF}(\succ)$ is an extreme point. Pick a non-deterministic attention rule $\AttFilter\in \mathcal{MT}(\succ)$. We show that $\AttFilter$ cannot be an extreme point. Let $\mathcal{X}_{\AttFilter}\subset \mathcal{X}$ stand for all sets $S\in \mathcal{X}$ for which $\AttFilter(T|S)=1$ for no $T\subset S$. We start by choosing $\epsilon>0$ small enough so that none of the non-binding constraints are affected whenever $\epsilon$ is added to or subtracted from $\AttFilter(T|S)$ for all $T\subset S$ and $S\in \mathcal{X}$. Let $k_{\AttFilter}=\min_{S\in \mathcal{X}_{\AttFilter}}|S|$. Since $\AttFilter$ is not deterministic, such $k_{\AttFilter}$ exists. 

We begin with the following simple observation that given $S$ with $|S|=k_{\AttFilter}$ we can have at most two subsets of $S$ with $\AttFilter(T|S)\in (0,1)$. Moreover, it must be the case that $\AttFilter(S|S)\in (0,1)$. 

\begin{lem}\label{Lemma: appendix, thm6.1}
Let $S$ with $|S|=k_{\AttFilter}$ be given. Then there exist at most two $T\subset S$ such that $\AttFilter(T|S)\in (0,1)$. Furthermore, $\AttFilter(S|S)\in (0,1)$. 
\end{lem}

\begin{proof}
Suppose there exist three such subsets: $T_1$, $T_2$, and $T_3$. Since $\AttFilter$ is triangular the subsets which are considered with positive probability can be ordered by set inclusion. Hence, we can without loss of generality assume $T_1\subset T_2\subset T_3$. But then since $\AttFilter$ is monotonic and $T_1\subset T_2\subset S$ it must be that $\AttFilter(T_1|T_2)\in (0,1)$ and $\AttFilter(T_2|T_2)\in (0,1)$. This contradicts the definition of $k_{\AttFilter}$. Hence there can be at most two subsets $T_1$ and $T_2$ with positive probability. The same contradiction appears as long as $T_2\subsetneq S$. Hence, $T_2=S$. 
\end{proof}

Now for all sets $S\in \mathcal{X}_{\AttFilter}$ with $|S|=k_{\AttFilter}$, we define $\AttFilter'$ and $\AttFilter''$ as follows:

$$\AttFilter'(T|S)=\AttFilter(T|S)+\epsilon,$$ $$\AttFilter'(S|S)=\AttFilter(S|S)-\epsilon,$$  and  $$\AttFilter''(T|S)=\AttFilter(T|S)-\epsilon,$$ $$\AttFilter''(S|S)=\AttFilter(S|S)+\epsilon$$
where $T\subsetneq S$ with $\AttFilter(T|S)\in (0,1)$.

Suppose we have defined $\AttFilter'$ and $\AttFilter''$ for all sets with $|S|\leq l$ and let $S$ with $|S|=l+1$ be given. If there exist no $T\subset S$ and $S_T\subset S$ such that $\AttFilter'(T|S_T)\neq \AttFilter''(T|S_T)$ and $\AttFilter(T|S)=\AttFilter(T|S_T)$, then we set $\AttFilter(T|S)=\AttFilter'(T|S)=\AttFilter''(T|S)$ for all $T\subset S$. Otherwise, pick the smallest $T$ for which such $S_T$ exists. If $\AttFilter'(T|S_T)>\AttFilter''(T|S_T)$, then let $\AttFilter'(T|S)=\AttFilter(T|S)+\epsilon$ and $\AttFilter''(T|S)=\AttFilter(T|S)-\epsilon$ and if $\AttFilter'(T|S_T)<\AttFilter''(T|S_T)$, then let $\AttFilter'(T|S)=\AttFilter(T|S)-\epsilon$ and $\AttFilter''(T|S)=\AttFilter(T|S)+\epsilon$. If $T$ is the only set for which such $S_T$ exists, then let $T'$ be the largest set for which $\AttFilter(T'|S)\in (0,1)$. Otherwise $T'$ denotes the other set for which $S_{T'}$ satisfying the description exists. If $\AttFilter'(T|S_T)>\AttFilter''(T|S_T)$, then let $\AttFilter'(T'|S)=\AttFilter(T'|S)-\epsilon$ and $\AttFilter''(T'|S)=\AttFilter(T'|S)+\epsilon$ and if $\AttFilter'(T|S_T)<\AttFilter''(T|S_T)$, then let $\AttFilter'(T'|S)=\AttFilter(T'|S)+\epsilon$ and $\AttFilter''(T'|S)=\AttFilter(T'|S)-\epsilon$. For all other subsets $\AttFilter$, $\AttFilter'$, and $\AttFilter''$ agree. We proceed iteratively.

\begin{lem}
Suppose there exist $T\subset S$ and $S_T\subset S$ such that $\AttFilter'(T|S_T)\neq \AttFilter''(T|S_T)$ and $\AttFilter(T|S)=\AttFilter(T|S_T)$. Then either $T$ is the smallest set in $S$ satisfying the description or we can set $S_T=T$. 
\end{lem}

\begin{proof}
The claim follows from Lemma \ref{Lemma: appendix, thm6.1} when $|S|=k_{\AttFilter}+1$. Suppose the claim holds whenever $|S|\leq l$. We show that the claim holds when $|S|=l+1$. Let $T\subset S$ and $S_T\subset S$ satisfy the description and suppose $T$ is not the smallest set in $S$ satisfying the description. Since $\AttFilter'(T|S_T)\neq \AttFilter''(T|S_T)$, by construction, either $T$ is the largest set satisfying $\AttFilter(T|S_T)\in (0,1)$ or there exists $S_{S_T}\subset S_T$ such that $\AttFilter'(T|S_{S_T})\neq \AttFilter''(T|S_{S_T})$ and $\AttFilter(T|S_T)=\AttFilter(T|S_{S_T})$. If the first case is true, then since $\AttFilter$ is monotonic, it must be the case that  $\AttFilter(T'|T)=\AttFilter(T'|S_T)$ for all $T'\subset T$, and hence we are done. In the second case, the claim follows from induction. 
\end{proof}

\begin{lem}
For any $S$, there exist either zero or two subsets satisfying $\AttFilter'(T|S)\neq \AttFilter''(T|S)$. Moreover if there are two sets satisfying the description, then $\AttFilter'(T_1|S)> \AttFilter''(T_1|S)$ if and only if $\AttFilter'(T_2|S)< \AttFilter''(T_2|S)$.
\end{lem}
\begin{proof}
The claim is trivial when $|S|=k_{\AttFilter}$. Suppose the claim is true for all $S$ with $|S|\leq l$ and let $S$ with $|S|=l+1$ be given. If there is no $T$ which satisfies the description in the construction, then no subset will be affected. Suppose there exists only one such $T$. We show that there exists $T'\supset T$ such that $\AttFilter(T'|S)\in (0,1)$. To see this notice that by monotonicity property $\AttFilter(T''|S)\leq \AttFilter(T''|S_T)$ for all $T''\subset T$. Since by induction there are two subsets of $S_T$ for which $\AttFilter'(T|S_T)\neq \AttFilter''(T|S_T)$ either $\AttFilter(T''|S)<\AttFilter(T''|S_T)$ for some $T''\subset T$ or there exists $T'''\supset T$ such that $\AttFilter(T'''|S_T)\in (0,1)$. In both cases, $\sum_{T''\subset T}\AttFilter(T''|S)<1$ follows. Hence, there is $T'\supset T$ such that $\AttFilter(T'|S)\in (0,1)$. The construction then guarantees that $\mu'(T'|S)\neq \mu''(T'|S)$ for some $T'\supset T$. Now suppose there are three subsets, $T_1,T_2$, and $T_3$, satisfying the description. Since $\AttFilter$ is triangular, we can without loss of generality assume that $T_1\subset T_2\subset T_3$. By the previous lemma, we can without loss of generality assume $S_{T_2}=T_2$ and $S_{T_3}=T_3$. But then since $\AttFilter$ is monotonic, 3 subsets of $S_{T_3}$ must satisfy the description, a contradiction to induction hypothesis. 

To prove the second part of the claim, notice that the claim is follows from construction if $|S|=k_{\AttFilter}$. Suppose the claim holds whenever $|S|\leq l$ and let $|S|=l+1$ be given. If $T_2=S$, then the claim follows from construction. If $T_2\subsetneq S$, then the claim follows from induction and construction by considering the set $T_2$.
\end{proof}

It is clear that $\AttFilter=\frac{1}{2}\AttFilter'+\frac{1}{2}\AttFilter''$. The previous lemmas also show that both $\AttFilter'$ and $\AttFilter''$ are monotonic. Hence, no $\AttFilter\in \mathcal{MT}(\succ)- \mathcal{AF}(\succ)$ can be an extreme point. This concludes the proof of Theorem \ref{Theorem: Random Model}.

\bibliographystyle{jpe}
\bibliography{Cattaneo-Ma-Masatlioglu-Suleymanov_2019_RAM}

\end{document}


\title{\vspace{-0.25in} 
 Supplemental Appendix to ``A Random Attention Model''\thanks{Financial support from the National Science Foundation through grant SES-1628883 is gratefully acknowledged.}
 \bigskip\bigskip}
\author{
Matias D. Cattaneo\thanks{Department of Operations Research and Financial Engineering, Princeton University.} 
\and 
Xinwei Ma\thanks{Department of Economics, University of California San Diego.}
\and
Yusufcan Masatlioglu\thanks{Department of Economics, University of Maryland.}
\and
Elchin Suleymanov\thanks{Department of Economics, Krannert School of Management, Purdue University.}
}\maketitle

\vspace{0.2in}\begin{abstract}
This Supplemental Appendix contains additional discussion of related literature, other numerical and methodological results, omitted theoretical proofs, and simulation evidence showcasing the finite sample performance of our proposed econometric procedures.
\end{abstract}

\ \\
\ \\

\setlength{\cftsecnumwidth}{4em}
\setlength{\cftsubsecnumwidth}{4em}
\tableofcontents

\vfill

\thispagestyle{empty}
\clearpage

\doublespacing
\setcounter{page}{1}
\pagestyle{plain}

\section{Related Literature}\label{SA-section:related lit}

\citet{Manzini-Mariotti_2014_ECMA} and \citet{Brady-Rehbeck_2016_ECMA} are the two closest related papers to our work. Similar to ours, both models consider data of a single individual. Both provide parametric random attention models, which are described in the third and fourth examples in Section 2.1 of the main paper, respectively. Since their attention rules are monotonic, these models are two interesting special cases of our random attention model. To provide an accurate comparison, we need to introduce an outside/default option, which is required by both models. Thus, we first extend RAM to accommodate an outside option and then offer a detailed comparison between our work and both of these papers.

Let $a^*\notin X$ be the default option. In the model with a default option, we will allow an empty consideration set. Hence, now $\AttFilter(\cdot|S)$ is defined over all subsets of $S$ including the empty set. The default option is always available and can be interpreted as choosing nothing whenever the consideration set is empty. Let $X^* =X \cup \{a^*\}$ and $S^* =S \cup \{a^*\}$ for all $S\in \mathcal{X}$. We require that the choice rule satisfy $\sum_{a\in S^*}\ChoiProb(a|S)=1$ and $\ChoiProb(a|S)\geq 0$ for all $a\in S^*$. We say that a choice rule $\ChoiProb$ has a random attention representation with a default option if there exists a preference ordering $\succ$ on $X$ and a monotonic attention rule $\AttFilter$ such that for each $a\in S$, 
$\ChoiProb(a|S)=\sum_{T\subset S}\Indicator(a \text{ is } \succ \text{-best in } T)\cdot \AttFilter(T|S)$. Thus $\ChoiProb(a^*|S)=\AttFilter(\emptyset|S)$.

An implication of Assumption 1 in the main paper is that for all $S$, $\AttFilter(\emptyset|S)\leq \AttFilter(\emptyset|S-a)$. In terms of choice probabilities, this implies that the default option satisfies regularity. In fact, it is easy to see that regularity on default and acyclicity of $\mathsf{P}$ are necessary and sufficient for $\ChoiProb$ to have a RAM representation with a default option. 

\begin{remark}
A choice rule $\ChoiProb$ has a RAM representation with a default option if and only if it satisfies acyclicity of $\mathsf{P}$ and regularity on default. 
\end{remark}

In \citet*{Manzini-Mariotti_2014_ECMA}, a choice rule has the representation
$$\ChoiProb(a|S)=\gamma(a)\prod_{b\in S: b\succ a}(1-\gamma(b))$$
where $\gamma(a)$ is the fixed consideration probability of alternative $a$. See \citet{Horan_2018} for an axiomatic characterization of this model when there is no default option. \citet{Demirkan_Kimya_2018} study a generalization of this model where $\gamma$ depends on the menu. Independent consideration model is a special case of RAM with a default option if we set 
$$\AttFilter(T|S)=\prod_{a\in T}\gamma(a)\prod_{b\in S-T}(1-\gamma(b)).$$ 

In \cite*{Brady-Rehbeck_2016_ECMA}, a choice rule has the representation
$$\ChoiProb(a|S)=\frac{\sum_{T \subset S} \Indicator(a\in T \text{ is $\succ$-maximal in T})\cdot w(T)}{\sum_{T\subset S}w(T)},$$
where $w(T)$ is the weight of consideration set $T$. This model is also a special case of RAM with an outside option if we set
$$\AttFilter(T|S)=\frac{w(T)}{\sum_{T' \subset S}w(T')}.$$

Another closely related paper on random consideration sets is \citet{Aguiar_2015}. In fact, acyclicity of $\mathsf{P}$ appears as an axiom in his representation. Hence, his model is also a special case of RAM.\medskip

Our paper is also related to the recent literature which attempts to combine heterogeneous preferences with random attention. The most closely related papers in this literature are \citet{Abaluck_Adams_2017}, \citet*{Barseghyan-Coughlin-Molinari-Teitelbaum_2018_wp}, and \citet*{Dardanoni_Manzini_Mariotti_Tyson_2018}. In a general setup, \citet*{Abaluck_Adams_2017} show that, by exploiting asymmetries in cross-partial derivatives, consideration set probabilities and utility can be separately identified from observed choices when there is rich exogenous variation in observed covariates. \citet{Barseghyan-Coughlin-Molinari-Teitelbaum_2018_wp} provide partial identification results when exogenous variation in observed covariates is more restricted. As opposed to our paper, both of these papers consider heterogeneous preferences and do not assume observable variation in menus decision makers face. On the other hand, we adopt abstract choice theory setup and our results do not rely on the existence of observed covariates. Lastly, similar to previous papers, \citet{Dardanoni_Manzini_Mariotti_Tyson_2018} study choices from a fixed menu of alternatives. They consider aggregate choice where individuals might differ both in terms of their consideration capacities and preferences. Their model is complementary to ours as the attention rule each individual utilizes in their model can be considered a special case of RAM.\medskip

We also compare our model with other random choice models even though they are not in the framework of random consideration sets. \citet*{Gul_Natenzon_Pesendorfer_2014_ECMA} consider the following model. There is a collection of attributes, denoted by $Z$. Each attribute $z\in Z$ has a weight $v(z)$, and each alternative $a\in X$ has intensity $i(a,z)$ in attribute $z$. The decision maker first randomly picks an attribute using Luce rule given the weights of all attributes. Then the decision maker picks an alternative using Luce rule given the intensities of all alternatives in attribute $z$. \citet*{Gul_Natenzon_Pesendorfer_2014_ECMA} show that any attribute rule is a random utility model. Since RUM is a subset of RAM, any choice behavior that can be explained by an attribute rule can also be explained by RAM.\medskip

\citet*{Echenique_Saito_2019} consider a general Luce model (GLM) where the decision maker uses Luce rule to choose among alternatives in her (deterministic) consideration set instead of the whole choice set. See \cite{Ahumada_Ulku_2018} for a characterization of GLM and \citet{Horan_2018_TLR} for a related model. We show that RAM and GLM are distinct in terms of observed choices. Recall Example 2 in the main paper, which cannot be explained by RAM. This example can be explained by GLM by assuming that $\lambda_i$ is the utility weight of the alternative $i$, and that whenever an alternative is not chosen that is because the alternative does not belong to the consideration set. Now consider a choice rule where all alternatives are always chosen with positive probability but Luce's IIA is not satisfied. We can construct such an example where acyclicity of $\mathsf{P}$ holds, and hence the choice rule has a RAM representation. However, this choice behavior cannot be explained by GLM as GLM reduces to Luce rule when all alternatives are chosen with positive probability in all menus. \citet*{Echenique_Saito_2019} also consider two interesting special cases of GLM: (i) in two-stage Luce model, consideration sets are induced by an asymmetric, transitive binary relation, (ii) in threshold Luce model, an alternative belongs to the consideration set only if its utility is not too small compared to the utility of the maximal element in the choice set. It can be shown that RAM and two-stage Luce model are distinct, and threshold Luce model is a subset of both models.\medskip

\citet*{Echenique_Saito_Tserenjigmid_2018} propose a model (PALM) which uses violations of Luce's IIA to reveal perception priority of alternatives. For an example of stochastic choice data which can be explained by RAM but not PALM, consider any data where the outside option is never chosen. When the outside option is never chosen, PALM reduces to Luce rule. However, RAM allows for violations of Luce's IIA in the absence of an outside option. On the other hand, RAM with an outside option satisfies regularity on default, but PALM does not necessarily satisfy this property. Hence, RAM with an outside option does not nest PALM.\medskip

\citet*{Fudenberg_Iijima_Strzalecki_2015_ECMA} consider a model (Additive Perturbed Utility-APU) where agents randomize as making deterministic choices can be costly. In their model, choices satisfy regularity. Since any choice rule that satisfies regularity has a RAM representation, RAM includes APU.\medskip

\citet*{Aguiar_Boccardi_Dean_2016_JET} consider a satisficing model where the decision maker searches till she finds an alternative above the satisficing utility level. If there is no alternative above the satisficing utility level, the decision maker picks the best available alternative. They focus on two special cases of this model: (i) Full Support Satisficing Model (FSSM) where in any menu each alternative has a positive probability of being searched first, and (ii) Fixed Distribution Satisficing Model (FDSM). They show that FDSM is a subset of RUM, and hence it is also a subset of RAM. On the other hand, FSSM has no restrictions on observed choices if all alternatives are always chosen with positive probability. Hence, there exist choice rules that can be explained by FSSM but not RAM. Lastly, consider the choice data $\ChoiProb(a|\{a,b,c\})=\ChoiProb(a|\{a,b\})=\ChoiProb(a|\{a,c\})=1$ and $\ChoiProb(b|\{b,c\})=1/2$. FSSM cannot explain this choice behavior even though regularity is satisfied. Hence, FSSM and RAM are distinct, and FDSM is a subset of both.


\section{Other Examples of RAM}\label{SA-section: Examples}

Here we provide more examples of random consideration sets that satisfy our key monotonicity assumption (Assumption 1 in the main paper).

\small\begin{enumerate}
\item  (\textsc{Full Attention}) The decision maker considers everything with probability one: $\AttFilter(T|S)=\Indicator(T=S)$.

\item \citep[\textsc{Top N};][]{Salant_Rubinstein_2008_Restud} The decision maker faces a list of alternatives created by some ordering. She pays attention to the first $N$ elements among available alternatives (e.g., first page of Google search results). If the number of available alternatives is less than $N$, she pays attention to the entire set. Formally, let $S(k,R)$ denote the set of first $k$ elements in $S$ according ordering $R$ provided that $k \leq  |S|$. ($S(|S|,R)$ is equal to $S$.) In our framework: $\AttFilter(T|S)=\Indicator(T=S(\min\{|S|, N\},R))$.

\item (\textsc{Satisficing Consideration Set}) The decision maker observes alternatives sequentially from a pre-determined list. The order of alternatives is unknown to the decision maker in the beginning of the search and uncovers them during the search process. The decision maker stops searching upon finding a satisfactory alternative \citep{Simon_1955_QJE}. If there is no such alternative, she searches the entire budget set. Formally, given the list $L$, $RS_{L}(S)$ denotes the range of search (the consideration set) when the budget set is $S$. In our framework: $\AttFilter(T|S)=\Indicator(T=RS_{L}(S))$.

\item (\textsc{At Most $k$ Alternatives}) The decision maker considers at most $k$ alternatives for any decision problem. If there are more alternatives than $k$, she considers only subsets including exactly $k$ alternatives with equal probability. If there are less alternatives than $k$, she considers everything. In our framework:
\[\AttFilter(T|S)=\left\{\begin{array}{cc}
                          1 & \text{if }  |S| \leq  k \text{ and } T = S \\
                          \binom{|S|}{k}^{-1}  & \text{if }  |S| > k \text{ and } |T| = k \\
                          0 & \text{otherwise}
                          \end{array}
                   \right.\]

\item (\textsc{Uniform Consideration}) The decision maker considers any subset of the feasible set with equal probabilities. That is,  for all $T\subset S$, $\AttFilter(T|S)=1/(2^{|S|}-1)$.

\item \citep[\textsc{Fixed Correlated Consideration};][]{Barbera_Grodal_JME_2011, Aguiar_2017_EL} The decision maker pays attention to each alternative with a fixed probability but the consideration of alternatives is potentially correlated. Formally, let $\omega$ be a probability distribution over $\mathcal{X}$. Then each alternative $a\in X$ is considered with a fixed probability $\sum_{T\in \mathcal{X}}\Indicator(a\in T)\cdot\omega(T)$ for all $S\ni a$. In our framework:
\[\AttFilter(T|S)=\sum_{T'\in \mathcal{X}}\Indicator(T'\cap S=T)\cdot \omega(T').\]

\item (\textsc{Ordered Logit Attention}) This is another generalization of the logit attention example. The decision maker ranks all subsets in terms of their attention priority, and she only considers subsets which are maximal with respect to that ordering.  When there are several best subsets, the decision maker considers each of them with certain frequency as in the Logit Attention example. Thus, consideration sets are constructed in the spirit of standard maximization paradigm. Formally, let $\trianglerighteq$  be a complete and transitive priority order over subsets $\mathcal{X}$. $S\trianglerighteq T$ reads as ``$S$ has a higher attention priority than $T$". The case when $S$ and $T$ have the same attention priority is denoted by $S \bowtie T$.  Formally,  
\[\AttFilter (T|S)=\left\{\begin{array}{ll}
                           \frac{w_{T}}{\sum_{\small{T' \subset S}}\Indicator(T'\bowtie T)\cdot w_{T'}} &
                             \text{if } T \text{ is } \trianglerighteq \text{-best} \text{ in } S \\
                           0 & \text{otherwise}
                           \end{array}
                    \right.\]

\item (\textsc{Elimination by Aspects, General}) The example is similar to the one given in the main text, except that when the decision maker picks an irrelevant aspect, she selects a subset at random drawn from the uniform distribution. Formally, 
\[\AttFilter (T|S)=\sum_{\substack{B_{i}\cap S=T }} \omega(i) + \frac{1}{2^{|S|}-1}  {\sum_{\substack{B_{k} \cap  S = \emptyset }}\omega(k)}.\]
The probability that the decision maker selects an irrelevant aspect is $\sum_{\substack{k: B_{k}\cap S=\emptyset }} \omega(k)$. In this case, $T$ is randomly chosen, which is reflected by the number $\frac{1}{2^{|S|}-1}$. We can generate similar examples. For instance, if the initial screening is not successful (choosing an irrelevant aspect), the decision maker may consider all the alternatives. Formally, 
\[\AttFilter(T|S)=\left\{\begin{array}{ll}
                          \sum_{\substack{B_{i}\cap S=T }} \ \ \omega(i)  & \text{if } T \neq S\\
                          \sum_{\substack{B_{k}\cap S=S, \emptyset }} \omega(k)  & \text{if } T = S
                          \end{array}
                   \right.\]

\item (\textsc{Stochastic Satisficing}) Suppose the satisficer faces multiple lists. The probability that the decision maker faces list $L$ is denoted by $p(L)$. As opposed to Example 3, consideration sets are stochastic. Formally,
\[\AttFilter(T|S)=\sum_{{\small{\substack{L:T=RS_{L}(S)}}}} p(L)\]
where $RS_{L}(S)$ is the range of search when the budget set is $S$ and $L$ is the list.

\item (\textsc{$1/N$ Rule}) The decision maker utilizes $N$ different sources of recommendations with equal probabilities. Given a fixed source $s$, she considers only top $k_{s}$ alternatives according to ordering $R_{s}$ that source $s$ is using, which could be different from her preference. For example, she utilizes either Google or Yahoo with equal probability. Once she decides which one to look at, she pays attention to only the products appearing in the first page of the corresponding search result. Formally, 
\[\AttFilter(T|S)=\frac{\big|\{s\|T=S(\min\{k_{s},|S|\},R_{s})\}\big|}{N}\]
where $S(k,R)$ denotes the first $k$ alternatives in $S$ according to $R$.

\item \citep[\textsc{Random Product Network};][]{Masatlioglu_Suleymanov_2017} Consider a decision maker faced with a product network $\mathcal{N}$. If $(a,b)\in \mathcal{N}$, then the decision maker who considers $a$ is recommended alternative $b$ (or alternatively, $b$ is linked to $a$). The decision maker starts search from a random starting point. Given a realized starting point in the product network, the decision maker considers all alternatives which are directly or indirectly linked to that starting point. Formally, let $\eta$ be a probability distribution on $X$. Then
$$\AttFilter(T|S)=\sum_{a\in S}\Indicator(T=N_a(S)) \cdot \frac{\eta(a)}{\sum_{b\in S}\eta(b)}$$
where $N_a(S)$ denotes all alternatives which are directly or indirectly linked to $a$ in $S$. 

\item \citep*[\textsc{Path Dependent Consideration};][]{Suleymanov_2018} This example is similar to the one above, except that now the decision maker starts searching from a fixed starting point $a^*\notin X$ (the default option) and takes a random path on a network. Let $X^*=X\cup \{a^*\}$, and $\mathcal{P}_{a^*}$ stands for all possible paths in $X^*$ with the initial node $a^*$. When the choice set is $X$, the decision maker takes the path $\rho\in \mathcal{P}_{a^*}$ with probability $\gamma(\rho)$. Given $\rho\in \mathcal{P}_{a^*}$ and $S\in \mathcal{X}$, $\rho_S\in \mathcal{P}_{a^*}$ is the subpath of $\rho$ in $S$ with the same initial node $a^*$. For any $\rho$, let $V(\rho)$ be the vertices of the path $\rho$ excluding the initial node. Then
$$\AttFilter(T|S)=\sum_{\rho \in \mathcal{P}_{a^*}}\Indicator(T=V(\rho_S)) \cdot \gamma(\rho).$$
This model is a subset of RAM with a default option. 

\end{enumerate}

\section{Limited Data}\label{SA-section: Limited data}

We discuss how our results can be adapted to handle limited data, that is, settings where not all possible choice probabilities or choice problems are observed. To investigate the implications of random attention models with limited data, assume that we observe choices from the set of choice problems $\mathcal{S}$. Let $\ChoiProb_{\text{obs}}$ denote the observed choice behavior. We say that $\ChoiProb_{\text{obs}}$ is consistent with the random attention model if there exists $\ChoiProb$ defined on the entire domain $\mathcal{X}$ such that $\ChoiProb(a|S)=\ChoiProb_{\text{obs}}(a|S)$ for all $S\in  \mathcal{S}$ and $\ChoiProb$ is a RAM. We call such $\ChoiProb$ an extension of $\ChoiProb_{\text{obs}}$. As \cite{DeClippel_Rozen_2014} point out, it is possible that $\ChoiProb_{\text{obs}}$ satisfies the acyclicity of $\mathsf{P}$ even though it is inconsistent with RAM. Here we provide an example with stochastic choices in which the same problem occurs.

\begin{example}
\label{Example: limited_data1} Let $\mathcal{S}=\{ \{a,b,c,d\}, \{b,c,d\},\{a,c\}\}$. Consider the following choice rule. 

\begin{equation*}\begin{array}{c|ccc}
                 \ChoiProb_{\text{obs}}(\cdot|S)  & S=\{a,b,c,d\}  & \{b,c,d\} & \{a,c\} \\ \hline
                 a & 1/4 & & 1/5  \\
                 b & 1/4 & 1/5 & \\
                 c & 1/4 & 3/5 & 4/5  \\
				 d & 1/4 & 1/5 &
                 \end{array}
\end{equation*}

In this example, we observe choices only from 3 choice problems instead of 15 potential choice problems. We show that these observations are sufficient to conclude that $\ChoiProb_{\text{obs}}$ is not consistent with RAM. Suppose there exists a RAM $\ChoiProb$ that extends $\ChoiProb_{\text{obs}}$. Then the observation $\ChoiProb(a|\{a,b,c,d\})>\ChoiProb(a|\{a,c\})$ tells us that either $a\mathsf{P} b$ or $a\mathsf{P} d$. To see this, notice that at least one of $\ChoiProb(a|\{a,b,c,d\})>\ChoiProb(a|\{a,c,d\})$ and $\ChoiProb(a|\{a,c,d\})>\ChoiProb(a|\{a,c\})$ must hold or we get a contradiction. On the other hand, from $\ChoiProb(b|\{a,b,c,d\})>\ChoiProb(b|\{b,c,d\})$ and $\ChoiProb(d|\{a,b,c,d\})>\ChoiProb(d|\{b,c,d\})$ we learn that $d\mathsf{P} a$ and $b\mathsf{P} a$. Hence, even though acyclicity of $\mathsf{P}$ is satisfied on the limited domain, there does not exist an extension of $\ChoiProb_{\text{obs}}$ that is RAM.
\end{example}

From the example above we can note the following: if $\ChoiProb_{\text{obs}}(a|S)>\ChoiProb_{\text{obs}}(a|S-A)$, then it must be that $a \mathsf{P}b$ for some $b\in A$. Now suppose $\ChoiProb_{\text{obs}}(a|S)+\ChoiProb_{\text{obs}}(b|S')>1$ and $\{a,b\}\subset S\cap S'$ where $a\neq b$. Then if $\ChoiProb$ is an extension of $\ChoiProb_{\text{obs}}$ it has to be the case that either $\ChoiProb(a|S)>\ChoiProb(a|S\cap S')$ or $\ChoiProb(b|S)>\ChoiProb(b|S\cap S')$. Hence, either $a \mathsf{P} c$ for some $c\in S- S'$ or $b \mathsf{P} d$ for some $d\in S'-S$. This is exactly the probabilistic analog of the condition in \cite{DeClippel_Rozen_2014}. However, this condition is not enough in probabilistic domain. The next example illustrates this point. 

\begin{example}
\label{Example: limited_data2} Consider the following  choice rule. 

\begin{equation*}\begin{array}{c|ccccc}
                 \ChoiProb_{\text{obs}}(\cdot|S)  & S=\{a,b,c,d\} & \{a,b,c,e\} & \{a,b,d\} & \{a,c,d\} & \{b,c,e\} \\ \hline
                 a & 1/4 & 2/3 & 1/2 & 1/2 &  \\
                 b & 1/4 & 1/6 & 1/2 &     & 5/6 \\
                 c & 1/4 & 0   &     & 1/2 & 1/12\\
				 d & 1/4 &     & 0   &  0  &  \\
				 e &     & 1/6 & & &  1/12
                 \end{array}
\end{equation*}

First, $\ChoiProb_{\text{obs}}(d|\{a,b,d\})<\ChoiProb_{\text{obs}}(d|\{a,b,c,d\})$ and $\ChoiProb_{\text{obs}}(d|\{a,c,d\})<\ChoiProb_{\text{obs}}(d|\{a,b,c,d\})$ imply that $d\mathsf{P}c$ and $d\mathsf{P}b$. Furthermore, $\ChoiProb_{\text{obs}}(e|\{b,c,e\})<\ChoiProb_{\text{obs}}(e|\{a,b,c,e\})$ implies that $e\mathsf{P}a$. Now consider the set $\{a,b,c\}$ and notice that $\ChoiProb_{\text{obs}}(a|\{a,b,c,e\})+\ChoiProb_{\text{obs}}(b|\{a,b,c,d\})+\ChoiProb_{\text{obs}}(c|\{a,b,c,d\})>1$. Thus if we had observations on the choice problem $\{a,b,c\}$ one of the following would have been true: (i) the probability that $a$ is chosen decreases when $e$ is removed from the choice problem $\{a,b,c,e\}$, (ii) the probability that $b$ is chosen decreases when $d$ is removed from the choice problem $\{a,b,c,d\}$, or (iii) the probability that $c$ is chosen decreases when $d$ is removed from the choice problem $\{a,b,c,d\}$. Hence, one of the following must be true: $a\mathsf{P}e$, $b\mathsf{P}d$, or $c\mathsf{P}d$. Since we have a contradiction in all cases, $\ChoiProb_{\text{obs}}$ is inconsistent with RAM. 
\end{example}

We generalize the intuition from this example. Suppose there exists a collection of pairs $(a_i,S_i)_{i=1}^m$ such that $\{a_1,\dots, a_m\}\subset \bigcap_{i=1}^m S_i$ and $\sum_{i=1}^m\ChoiProb_{\text{obs}}(a_i,S_i)>1$ where $a_i$ are all distinct. Now in the choice problem $\bigcap_{i=1}^m S_i$ the probability that $a_i$ is chosen must decrease for at least one $i\in\{1,\dots,m\}$. From here we can conclude that $a_i \mathsf{P} b_i$ for some $b_i\in S_i- (S_1\cap \cdots \cap S_{i-1}\cap S_{i+1}\cap \cdots \cap S_{m})$ for some $i\in \{1,\dots,m\}$. Hence, the existence of an acyclic $\mathsf{P}$ that satisfies this condition is necessary for $\ChoiProb_{\text{obs}}$ to be consistent with RAM. Theorem \ref{Theorem: Limited Data} shows that it is also sufficient.\medskip

\begin{thm}
\label{Theorem: Limited Data}
A choice rule $\ChoiProb_{\text{obs}}$ is consistent with RAM if and only if there exists an acyclic binary relation $\mathsf{P}$ on $X$ which satisfies the following: for any collection $(a_i,S_i)_{i=1}^m$ with distinct $a_i$ such that $\{a_1,\dots, a_m\}\subset \bigcap_{i=1}^m S_i$ and $\sum_{i=1}^m\ChoiProb_{\text{obs}}(a_i,S_i)>1$, $a_i \mathsf{P} b_i$ for some $b_i\in S_i- (S_1\cap \cdots \cap S_{i-1}\cap S_{i+1}\cap \cdots \cap S_{m})$ for some $i\in \{1,\dots,m\}$.
\end{thm}

\begin{proof}
Let $\succ$ be a transitive completion of $\mathsf{P}$. We reorder the alternatives in $X$ such that $a_{1,\succ}\succ \cdots \succ a_{K,\succ}$. We define $\AttFilter$ as follows. For any $S\in \mathcal{S}$, 
\begin{equation*} 
\AttFilter(T|S)=\left\{ 
\begin{array}{cc}
\ChoiProb_{\text{obs}}( a_{k,\succ}|S) & \text{if } \exists \ k \  s.t.  \ T =L_{k,\succ} \cap S\\ 
0 & \text{otherwise}%
\end{array}
\right.\end{equation*} 
where $L_{k,\succ}=\{a_{k,\succ},\dots,a_{K,\succ}\}$. For any $S\in \mathcal{X}-\mathcal{S}$, if there is $S'\in \mathcal{S}$ with $S'\supset S$, then 
\begin{equation*} 
\AttFilter(T|S)=
\underset{S'\in \mathcal{S}: S'\supset S}{\max} \AttFilter(T|S') \quad \text{ if } T\subsetneq S
\end{equation*} 
and $\AttFilter(S|S)=1-\sum_{T \subsetneq S}\AttFilter(T|S)$. Finally, for $S\in \mathcal{X}-\mathcal{S}$, if there is no $S'\in \mathcal{S}$ with $S'\supset S$, then $\AttFilter(S|S)=1$ and $\AttFilter(T|S)=0$ for all $T\subsetneq S$. 

It is easily seen that $(\succ,\AttFilter)$ represents $\ChoiProb_{\text{obs}}$. We first need to show that $\AttFilter(\cdot|S)$ is a probability distribution. The only case we need to check is when $S\in \mathcal{X}-\mathcal{S}$ and there exists $S'\supset S$ with $S'\in \mathcal{S}$. We need to show that $\AttFilter(S|S)\geq 0$ or that $\sum_{T \subsetneq S}\AttFilter(T|S)\leq 1$. Suppose $\sum_{T \subsetneq S}\AttFilter(T|S)> 1$. By definition, for each $T\subsetneq S$ such that $\AttFilter(T|S)>0$, there exists a pair $(a_{k_T},S_T)$ such that $S_T\in \mathcal{S}$ with $S_T\supset S$ and $\AttFilter(T|S)=\AttFilter(T|S_T)=\ChoiProb_{\text{obs}}(a_{k_T}|S_T)$ where $T=L_{a_{k_T},\succ}\cap S_T$. Then, $\sum_{T \subsetneq S}\ChoiProb_{\text{obs}}(a_{k_T}|S_T)>1$.  Notice that since $\AttFilter$ is triangular, $a_{k_T}$ are distinct. By definition of $\mathsf{P}$, there exist $T\subsetneq S$ and an alternative $b_{k_T}$ in $S_T-S$ such that $a_{k_T}\mathsf{P}b_{k_T}$. But this is a contradiction as $b_{k_T}\notin T$ and $T=L_{a_{k_T},\succ}\cap S_T$.

We now need to show that $\AttFilter$ defined as above is monotonic. We have a few cases to consider. 

\textbf{Case 1:} $S, S-b \in \mathcal{S}$. Suppose $\AttFilter(T|S)>\AttFilter(T|S-b)$ where $b\notin T$. Since $S\in \mathcal{S}$ and $\AttFilter(T|S)>0$ it must be that $T=L_{k,\succ}\cap S$ for some $k$ and $\AttFilter(T|S)=\ChoiProb_{\text{obs}}(a_{k,\succ}|S)$. Since $b\notin T$ and $T=L_{k,\succ}\cap S$ we must have $b\succ a_{k,\succ}$. Therefore, it must be the case that $T=L_{k,\succ}\cap (S-b)$. By definition, $\AttFilter(T|S-b)=\ChoiProb_{\text{obs}}(a_{k,\succ}|S-b)$. But then we have $\ChoiProb_{\text{obs}}(a_{k,\succ}|S)>\ChoiProb_{\text{obs}}(a_{k,\succ}|S-b)$ which implies that $a_{k,\succ}\mathsf{P} b$, a contradiction. 

\textbf{Case 2:} $S\in \mathcal{S}$, $S-b\notin \mathcal{S}$. Let $T$ with $b\notin T$ given. Since $S\in \mathcal{S}$ it must be that either $\AttFilter(T|S)=0$ in which case monotonicity is trivial or $T=L_{k,\succ}\cap S$ for some $k$ and $\AttFilter(T|S)=\ChoiProb_{\text{obs}}(a_{k,\succ}|S)$. First, suppose $T\subsetneq S-b$. Now $S\supset S-b$, and hence by definition, $\AttFilter(T|S-b)=\underset{S'\in \mathcal{S}: S'\supset S-b}{\max}\AttFilter(T|S')\geq \AttFilter(T|S)$. This establishes that the claim holds for all $T\subsetneq S-b$. Notice that if $\AttFilter(T|S-b)=\AttFilter(T|S)$ for all $T\subsetneq S-b$, then $\AttFilter(S-b|S-b)\geq \AttFilter(S-b|S)$ also follows. 

Suppose $\AttFilter(S-b|S)>\AttFilter(S-b|S-b)$. Then since $\AttFilter(S-b|S)>0$ and $\AttFilter$ is triangular we must have that $b$ is $\succ$ maximal in $S$. Thus $\AttFilter(S|S)=\ChoiProb_{\text{obs}}(b|S)$. Furthermore, by the argument in the previous paragraph, there exists $T\subsetneq S-b$ such that  $\AttFilter(T|S-b)>\AttFilter(T|S)$. Suppose there exists only one such $T$. (A similar argument will work if there is more than one such $T$.) By definition, there exists $S_T\in \mathcal{S}$ with $S_T\supset S-b$ such that $\AttFilter(T|S_T)>\AttFilter(T|S)$. Thus there exists $a_{k_T}$ such that $\AttFilter(T|S-b)=\AttFilter(T|S_T)=\ChoiProb_{\text{obs}}(a_{k_T}|S_T)$ and $T=L_{a_{k_T},\succ}\cap S_T=L_{a_{k_T},\succ}\cap (S-b)$ where the second equality follows from the fact that $T\subset S-b$. Notice that $a_{k_T}\in S-b$ and since $b$ is $\succ$ maximal in $S$ we have $b\succ a_{k_T}$. Hence we have $T=L_{a_{k_T},\succ}\cap S$ which by definition implies  $\AttFilter(T|S)=\ChoiProb_{\text{obs}}(a_{k_T}|S)<\ChoiProb_{\text{obs}}(a_{k_T}|S_T)$. Now since by assumption $\mu(T'|S)=\mu(T'|S-b)$ for all $T'\subsetneq S-b$ with $T'\neq T$, by using the definition of $\mu$, $\AttFilter(S-b|S)>\AttFilter(S-b|S-b)$ implies that $\ChoiProb_{\text{obs}}(b|S)+\ChoiProb_{\text{obs}}(a_{k_T}|S)<\ChoiProb_{\text{obs}}(a_{k_T}|S_T)$. Consider the collection $\{(a_{k_T},S_T)\}\cup \{(a_i,S)| a_i\in S \text{ and } a_i\neq b, a_i\neq a_{k_T}\}$. Since  $\ChoiProb_{\text{obs}}(a_{k_T}|S_T)+(1-\ChoiProb_{\text{obs}}(a_{k_T}|S)-\ChoiProb_{\text{obs}}(b|S))>1$, the observed choice probabilities summed over this collection adds up to greater than one. By definition of $\mathsf{P}$, either there exists an alternative in $a_i\in S-b$ such that $a_i\mathsf{P}b$ or there exists an alternative $c\in S_T-(S-b)$ such that $a_{k_T}\mathsf{P}c$. The first case leads to a contradiction since $b$ is $\succ$ maximal in $S$. The second case leads to a contradiction since $T=L_{a_{k_T},\succ}\cap S_T=L_{a_{k_T},\succ}\cap (S-b)\not \ni c$. 

\textbf{Case 3:} $S\notin \mathcal{S}$, $S-b \in \mathcal{S}$. If there exists no $S'\supset S$ such that $S'\in \mathcal{S}$, then monotonicity property is trivial as $\AttFilter(S|S)=1$. Hence, suppose there is $S'\supset S$ such that $S'\in \mathcal{S}$ and there exists $T\subset S-b$ such that $\AttFilter(T|S)>\AttFilter(T|S-b)$. Let $a_{k_T}$ and $S_T$ be such that $S_T\supset S$ and  $\AttFilter(T|S)=\ChoiProb_{\text{obs}}(a_{k_T}|S_T)$ where $T=L_{a_{k_T},\succ}\cap S_T$. Also notice that since $a_{k_T}\in T\subset S-b$, $a_{k_T}\neq b$. By definition of $\AttFilter$, it must be that $\ChoiProb_{\text{obs}}(a_{k_T}|S_T)>\ChoiProb_{\text{obs}}(a_{k_T}|S-b)$. Now consider the collection $\{(a_{k_T},S_T)\}\cup \{(a_i,S-b)| a_i\in S-b, a_i\neq a_{k_T}\}$. If we add the choice probabilities over this collection, they will add up to greater than one. Hence, by definition of $\mathsf{P}$, there exists $c\in S_T-(S-b)$ such that $a_{k_T}\mathsf{P}c$. But this is a contradiction as $T=L_{a_{k_T},\succ}\cap S_T$ and $T\subset S-b$.  

\textbf{Case 4:} $S\notin \mathcal{S}$, $S-b\notin \mathcal{S}$. If there is no $S'\supset S$ such that $S'\in \mathcal{S}$, then the claim is trivial. Suppose there exists such $S'$. Consider $T\subsetneq S-b$. By definition,
$$\AttFilter(T|S)=\underset{S'\supset S: S'\in \mathcal{S}}\AttFilter(T|S')\leq \underset{S''\supset S-b: S''\in \mathcal{S}}\AttFilter(T|S'')=\AttFilter(T|S-b).$$
Hence, the claim holds for all $T\subsetneq S-b$. We need to show that the claim also holds when $T=S-b$. Notice that if $\AttFilter(T|S)=\AttFilter(T|S-b)$ for all $T\subsetneq S-b$, then $\AttFilter(S-b|S)\leq \AttFilter(S-b|S-b)$ follows immediately. Hence suppose there is at least one $T$ such that $\AttFilter(T|S)<\AttFilter(T|S-b)$. 

Now if $b$ is not $\succ$ maximal in $S$, then $\AttFilter(S-b|S)=0$ and monotonicity is trivial. So suppose $b$ is $\succ$ maximal and $\AttFilter(S-b|S)>\AttFilter(S-b|S-b)$. For each $T\subsetneq S-b$ such that $\AttFilter(T|S-b)>0$, let $a_{k_T}$ and $S_T$ be such that $\AttFilter(T|S-b)=\ChoiProb_{\text{obs}}(a_{k_T}|S_T)$. Let $a_{k_{S-b}}$ and $S_{S-b}$ be such that $\AttFilter(S-b|S)=\ChoiProb_{\text{obs}}(a_{k_{S-b}}|S_{S-b})$. Consider the collection $\{(a_{k_T},S_T)|\: T\subset S-b\}$. Since $\AttFilter(S-b|S)>\AttFilter(S-b|S-b)$, if we sum choice probabilities over this collection, they add up to greater than one. This implies that $a_{k_T}\mathsf{P} c$ for some $T\subset S-b$ and for some $c\notin S-b$. But this is a contradiction since only sets of the form $T=L_{a_{k_T},\succ}\cap S_T$ are considered with positive probability and $c\notin T$.
\end{proof}

\section{Estimation and Inference}

Here we first provide the proof of Theorem 4 in the main paper. We then discuss the computational cost of our inference procedure, following which we connect our inference approach to the ``J-test'' method. Then, we demonstrate how the matrix $\bR_{\succ}$ can be constructed in a limited data scenario. Finally we discuss other possible ways to construct the critical value.

\subsection{Proof of Theorem 4}\label{SA-subsection: proof of econometrics}

For completeness and notational clarity, we first present two lemmas, which follow from the central limit theorem. 

\begin{lem}\label{lemma: CLT}
Under Assumption 2 in the main paper, 
\[
\sqrt{N_S}\Big( \hat{\bChoiProb}_S - \bChoiProb_{S} \Big) \rightsquigarrow \mathcal{N}\Big(\mathbf{0},\ \bOmega_{\ChoiProb,S}\Big),\qquad \bOmega_{\ChoiProb,S} = \mathrm{diag}(\bChoiProb_{S}) - \bChoiProb_{S}\bChoiProb_{S}^\Trans,
\]
where $N_S= \sum_{1\leq i\leq N} \Indicator(Y_i=S)$ is the effective sample size of menu $S$, $\mathcal{N}(\mathbf{0},\bOmega_{\ChoiProb,S})$ is the $|S|$-dimensional Gaussian distribution with mean zero and covariance $\bOmega_{\ChoiProb,S}$, and $\mathrm{diag}$ is the operator constructing diagonal matrices. 
\end{lem}

The asymptotic distribution is degenerate for two reasons. First, the choice rule, by construction, has to sum up to 1. Second, it is possible that some of the alternatives in $S$ is never chosen (either in the sample or the population). 

\begin{lem}
Under Assumption 2 in the main paper, 
\[
\sqrt{N}\Big( \hat{\bChoiProb} - \bChoiProb \Big) \rightsquigarrow \mathcal{N}\Big(\mathbf{0},\ \bOmega_{\ChoiProb}\Big),
\]
where the asymptotic variance $\bOmega_{\ChoiProb}$ is block diagonal, with blocks given by $\frac{1}{\Prob[Y_i=S]}\bOmega_{\ChoiProb,S}$, $S\in\mathcal{X}$. (Block diagonality of the asymptotic variance follows directly from the fact that across choice problems, choice rules are estimated with independent samples, hence are independent.)
\end{lem}

Recall that $\Tstat(\succ) = \sqrt{N}\max( (\bR_{\succ}\hat{\bChoiProb}) \oslash \hat{\bsigma}_\succ )_+$. The next proposition shows that its distribution is approximated by the infeasible statistic $\sqrt{N}\max( \bR_{\succ}(\bz^\star + \bChoiProb) \oslash \hat{\bsigma}_\succ )_+$. For clarity, we do not seek uniformity here. Later we will show how the same argument can be used to demonstrate distributional approximation uniformly in a class of DGPs (see Proposition \ref{Proposition: uniform distributional approx}). 

\begin{prop}\label{Proposition: distributional approx}
Suppose Assumption 2 in the main paper holds and $\min(\bsigma_{\ChoiProb,\succ})>0$. Then
\[
\Big| \Prob\Big[\Tstat(\succ) \leq t \Big] - \Prob^\star\Big[\sqrt{N}\max( \bR_{\succ}(\bz^\star + \bChoiProb) \oslash \hat{\bsigma}_\succ )_+ \leq t  \Big] \Big| \to_\Prob 0,\qquad \forall t\neq 0.
\]
\end{prop} 

\begin{remark}
We exclude $t=0$ since the limiting distribution may have a point mass at the origin. 
\end{remark}

\begin{proof}
Let $f$ be a bounded Lipschitz function. Without loss of generality assume its Lipschitz constant is $1$, and $2\cdot\Vert f\Vert_\infty=c$. For convenience, denote $\Expectation[f(\cdot)]$ by $\Expectation_f[\cdot]$. 

By the central limit theorem, it is possible to construct a (sequence of) random vector $\tilde{\bz} \sim \mathcal{N}(\mathbf{0}, \bOmega_{\ChoiProb}/N)$ such that $|\sqrt{N}(\hat{\bChoiProb}-\bChoiProb)-\sqrt{N}\tilde{\bz}|=\Op(1/\sqrt{N})$ (we postpone the proof to the end). Further, let $\bw^\star\sim \mathcal{N}(\mathbf{0},\bI/\sqrt{N})$ with suitable dimension such that $\hat{\bOmega}^{1/2}\bw^\star\sim\bz^\star$. Then consider bounds on the following quantities. 

\begin{align*}
&\ \left| \Expectation_f\Big[ \sqrt{N}\max\Big(\bR_{\succ}\hat{\bChoiProb}\oslash \hat{\bsigma}_\succ \Big)_+ \Big] -  \Expectation_f\Big[ \sqrt{N}\max\Big(\bR_{\succ}\hat{\bChoiProb}\oslash\bsigma_{\ChoiProb,\succ}\Big)_+ \Big]\right|\\
&\leq \Expectation\Big[\Big| \sqrt{N}\max\Big(\bR_{\succ}\hat{\bChoiProb}\oslash \hat{\bsigma}_\succ \Big)_+-\sqrt{N}\max\Big(\bR_{\succ}\hat{\bChoiProb}\oslash\bsigma_{\ChoiProb,\succ}\Big)_+ \Big|\wedge c\Big]\\
&\leq \Expectation\Big[\Big\Vert \sqrt{N}\Big(\bR_{\succ}\hat{\bChoiProb}\oslash \hat{\bsigma}_\succ \Big)_+-\sqrt{N}\Big(\bR_{\succ}\hat{\bChoiProb}\oslash\bsigma_{\ChoiProb,\succ}\Big)_+ \Big\Vert_\infty\wedge c\Big]\\
&= \Expectation\Big[\Big\Vert \sqrt{N}\Big(\bR_{\succ}\hat{\bChoiProb}\oslash \hat{\bsigma}_\succ - \bR_{\succ}\hat{\bChoiProb}\oslash\bsigma_{\ChoiProb,\succ} \Big)_+ \odot \Indicator\Big(\bR_{\succ}\hat{\bChoiProb}\geq 0\Big) \Big\Vert_\infty\wedge c\Big]\\
&\leq \Expectation\Big[\Big\Vert \sqrt{N}\Big(\bR_{\succ}(\hat{\bChoiProb}-\bChoiProb)\oslash \hat{\bsigma}_\succ - \bR_{\succ}(\hat{\bChoiProb}-\bChoiProb)\oslash\bsigma_{\ChoiProb,\succ} \Big)_+  \Big\Vert_\infty\wedge c\Big]\\
&\qquad + \Expectation\Big[\Big\Vert \sqrt{N}\Big(\bR_{\succ}\bChoiProb\oslash \hat{\bsigma}_\succ - \bR_{\succ}\bChoiProb\oslash\bsigma_{\ChoiProb,\succ} \Big)_+ \odot \Indicator\Big(\bR_{\succ}\hat{\bChoiProb}\geq 0\Big) \Big\Vert_\infty\wedge c\Big].
\end{align*}
The first inequality uses Lipschitz property of $f$ and the fact that the whole term is bounded by $2\cdot\Vert f\Vert_\infty=c$. The second inequality uses basic property of the max operator. The third inequality follows from triangle inequality of the norm $\Vert\cdot\Vert_\infty$. 

We further split in order to control the denominator:
\begin{align*}
&\ \text{previous display }\\
&\leq \Expectation\Big[\Big\Vert \sqrt{N}\Big(\bR_{\succ}(\hat{\bChoiProb}-\bChoiProb)\oslash \hat{\bsigma}_\succ - \bR_{\succ}(\hat{\bChoiProb}-\bChoiProb)\oslash\bsigma_{\ChoiProb,\succ} \Big)_+  \Big\Vert_\infty\Indicator\Big( \min(\hat{\bsigma}_\succ) \geq \min(\bsigma_{\ChoiProb,\succ})/2\Big)\wedge c\Big]\\
&\qquad + \Prob\Big[ \min(\hat{\bsigma}_\succ) < \min(\bsigma_{\ChoiProb,\succ})/2 \Big]\cdot c\\
&\qquad + \Expectation\Big[\Big\Vert \sqrt{N}\Big(\bR_{\succ}\bChoiProb\oslash \hat{\bsigma}_\succ - \bR_{\succ}\bChoiProb\oslash\bsigma_{\ChoiProb,\succ} \Big)_+ \odot \Indicator\Big(\bR_{\succ}\hat{\bChoiProb}\geq 0\Big) \Big\Vert_\infty\wedge c\Big]\\
&= O\Big(\frac{1}{\sqrt{N}}\Big) + \Expectation\Big[\Big\Vert \sqrt{N}\Big(\bR_{\succ}\bChoiProb\oslash \hat{\bsigma}_\succ - \bR_{\succ}\bChoiProb\oslash\bsigma_{\ChoiProb,\succ} \Big)_+ \odot \Indicator\Big(\bR_{\succ}\hat{\bChoiProb}\geq 0\Big) \Big\Vert_\infty\wedge c\Big]\\
&\leq O\Big(\frac{1}{\sqrt{N}}\Big) + \Expectation\Big[\Big\Vert \sqrt{N}\Big(\bR_{\succ}\bChoiProb\oslash \hat{\bsigma}_\succ - \bR_{\succ}\bChoiProb\oslash\bsigma_{\ChoiProb,\succ} \Big)_+ \odot \Indicator\Big(\bR_{\succ}\bChoiProb\geq -\frac{a_n}{\sqrt{N}}\Big) \Big\Vert_\infty\wedge c\Big]\\
&\qquad + \Expectation\Big[\Big\Vert \sqrt{N}\Big(\bR_{\succ}\bChoiProb\oslash \hat{\bsigma}_\succ - \bR_{\succ}\bChoiProb\oslash\bsigma_{\ChoiProb,\succ} \Big)_+ \odot \Indicator\Big(\bR_{\succ}\hat{\bChoiProb}\geq 0,\ \bR_{\succ}\bChoiProb< -\frac{a_n}{\sqrt{N}}\Big) \Big\Vert_\infty\wedge c\Big]\\
&\leq O\Big(\frac{1}{\sqrt{N}}\Big) + \Expectation\Big[\Big\Vert \sqrt{N}\Big(\bR_{\succ}\bChoiProb\oslash \hat{\bsigma}_\succ - \bR_{\succ}\bChoiProb\oslash\bsigma_{\ChoiProb,\succ} \Big)_+ \odot \Indicator\Big(\bR_{\succ}\bChoiProb\geq -\frac{a_n}{\sqrt{N}}\Big) \Big\Vert_\infty\wedge c\Big]\\
&\qquad + \Expectation\Big[\Big\Vert \Indicator\Big(\bR_{\succ}\hat{\bChoiProb}\geq 0,\ \bR_{\succ}\bChoiProb< -\frac{a_n}{\sqrt{N}}\Big) \Big\Vert_\infty\wedge c\Big]\\
&= O\Big( \frac{1}{\sqrt{N}} + \frac{a_n}{\sqrt{N}} + \frac{1}{a_n} \Big),
\end{align*}
with $a_n\to \infty$ chosen so that the last line vanishes. 

\begin{align*}
&\ \Big| \Expectation_f\Big[ \sqrt{N}\max\Big(\bR_{\succ}\hat{\bChoiProb}\oslash\bsigma_{\ChoiProb,\succ}\Big)_+ \Big]- \Expectation_f\Big[ \sqrt{N}\max\Big(\bR_{\succ}(\tilde{\bz}+\bChoiProb)\oslash\bsigma_{\ChoiProb,\succ}\Big)_+ \Big] \Big|\\
&\leq \Expectation\Big[\Big\Vert \sqrt{N}\Big(\bR_{\succ}\hat{\bChoiProb}\oslash \bsigma_{\ChoiProb,\succ} \Big)_+-\sqrt{N}\Big(\bR_{\succ}(\tilde{\bz} + \bChoiProb)\oslash\bsigma_{\ChoiProb,\succ}\Big)_+ \Big\Vert_\infty\wedge c\Big]\\
&\leq \Expectation\Big[\sqrt{N}\Big\Vert \bR_{\succ}\hat{\bChoiProb}\oslash \bsigma_{\ChoiProb,\succ} -\bR_{\succ}(\tilde{\bz} + \bChoiProb)\oslash\bsigma_{\ChoiProb,\succ} \Big\Vert_\infty\wedge c\Big]\\
&= \Expectation\Big[\sqrt{N}\Big\Vert \bR_{\succ}(\hat{\bChoiProb}-\bChoiProb)\oslash \bsigma_{\ChoiProb,\succ} -\bR_{\succ}\tilde{\bz} \oslash\bsigma_{\ChoiProb,\succ} \Big\Vert_\infty\wedge c\Big]\\
&= \Op\Big(\frac{1}{\sqrt{N}}\Big). 
\end{align*}
The first inequality uses Lipschitz continuity of $f$ and property of the max operator. The second inequality drops positivity. The rate in the last line comes from the coupling requirement of $\tilde{\bz}$. 

\[
\Big| \Expectation_f\Big[ \sqrt{N}\max\Big(\bR_{\succ}(\tilde{\bz}+\bChoiProb)\oslash\bsigma_{\ChoiProb,\succ}\Big)_+ \Big]-\Expectation_f^\star\Big[ \sqrt{N}\max\Big(\bR_{\succ}(\bOmega_{\ChoiProb}^{1/2}\bw^\star+\bChoiProb)\oslash\bsigma_{\ChoiProb,\succ}\Big)_+ \Big] \Big|=0,
\]
since the two terms have the same distribution. 

\begin{align*}
&\ \Big| \Expectation_f^\star\Big[ \sqrt{N}\max\Big(\bR_{\succ}(\bOmega_{\ChoiProb}^{1/2}\bw^\star+\bChoiProb)\oslash\bsigma_{\ChoiProb,\succ}\Big)_+ \Big]-\Expectation_f^\star\Big[ \sqrt{N}\max\Big(\bR_{\succ}(\hat{\bOmega}^{1/2}\bw^\star+\bChoiProb)\oslash\bsigma_{\ChoiProb,\succ}\Big)_+ \Big] \Big|\\
&\leq \Expectation^\star\Big[ \Big\Vert \sqrt{N}\Big(\bR_{\succ}(\bOmega_{\ChoiProb}^{1/2}\bw^\star+\bChoiProb)\oslash\bsigma_{\ChoiProb,\succ}\Big)_+ -\sqrt{N}\Big(\bR_{\succ}(\hat{\bOmega}^{1/2}\bw^\star+\bChoiProb)\oslash\bsigma_{\ChoiProb,\succ}\Big)_+ \Big\Vert_\infty \wedge c \Big]\\
&\leq \Expectation^\star\Big[ \sqrt{N}\Big\Vert \bR_{\succ}(\bOmega_{\ChoiProb}^{1/2}\bw^\star+\bChoiProb)\oslash\bsigma_{\ChoiProb,\succ}-\bR_{\succ}(\hat{\bOmega}^{1/2}\bw^\star+\bChoiProb)\oslash\bsigma_{\ChoiProb,\succ} \Big\Vert_\infty \wedge c \Big]\\
&= \Expectation^\star\Big[ \sqrt{N}\Big\Vert \bR_{\succ}\bOmega_{\ChoiProb}^{1/2}\bw^\star\oslash\bsigma_{\ChoiProb,\succ}-\bR_{\succ}\hat{\bOmega}^{1/2}\bw^\star\oslash\bsigma_{\ChoiProb,\succ} \Big\Vert_\infty \wedge c \Big]\\
&= \Op\Big(\frac{1}{\sqrt{N}}\Big).
\end{align*}
The first inequality uses Lipschitz continuity of $f$ and property of the max operator. The second inequality drops positivity. The rate in the last line comes from the fact $\sqrt{N}\Vert \hat{\bOmega} - \bOmega_{\ChoiProb}\Vert_\infty = \Op(1)$. Also, by construction $\bw^\star = O_{\Prob^\star}(1/\sqrt{N})$.  

\begin{align*}
&\ \Big| \Expectation_f^\star\Big[ \sqrt{N}\max\Big(\bR_{\succ}(\hat{\bOmega}^{1/2}\bw^\star+\bChoiProb)\oslash\bsigma_{\ChoiProb,\succ}\Big)_+ \Big] - \Expectation_f^\star\Big[ \sqrt{N}\max\Big(\bR_{\succ}(\hat{\bOmega}^{1/2}\bw^\star+\bChoiProb)\oslash\hat{\bsigma}_{\succ}\Big)_+ \Big] \Big|\\
&\leq \Expectation^\star\Big[ \Big\Vert \sqrt{N}\Big(\bR_{\succ}(\hat{\bOmega}^{1/2}\bw^\star+\bChoiProb)\oslash\bsigma_{\ChoiProb,\succ}\Big)_+ -\sqrt{N}\Big(\bR_{\succ}(\hat{\bOmega}^{1/2}\bw^\star+\bChoiProb)\oslash\hat{\bsigma}_{\succ}\Big)_+ \Big\Vert_\infty \wedge c \Big]\\
&= \Expectation^\star\Big[ \Big\Vert \sqrt{N}\Big(\bR_{\succ}(\hat{\bOmega}^{1/2}\bw^\star+\bChoiProb)\oslash\bsigma_{\ChoiProb,\succ}-\bR_{\succ}(\hat{\bOmega}^{1/2}\bw^\star+\bChoiProb)\oslash\hat{\bsigma}_{\succ}\Big)_+\odot\Indicator\Big(\bR_{\succ}(\hat{\bOmega}^{1/2}\bw^\star+\bChoiProb)\geq 0\Big) \Big\Vert_\infty \wedge c \Big]\\
&\leq \Expectation^\star\Big[ \Big\Vert \sqrt{N}\Big(\bR_{\succ}\hat{\bOmega}^{1/2}\bw^\star\oslash\bsigma_{\ChoiProb,\succ}-\bR_{\succ}\hat{\bOmega}^{1/2}\bw^\star\oslash\hat{\bsigma}_{\succ}\Big)_+ \Big\Vert_\infty \wedge c \Big]\\
&\qquad + \Expectation^\star\Big[ \Big\Vert \sqrt{N}\Big(\bR_{\succ}\bChoiProb\oslash\bsigma_{\ChoiProb,\succ}-\bR_{\succ}\bChoiProb\oslash\hat{\bsigma}_{\succ}\Big)_+\odot\Indicator\Big(\bR_{\succ}(\hat{\bOmega}^{1/2}\bw^\star+\bChoiProb)\geq 0\Big) \Big\Vert_\infty \wedge c \Big]\\
&= \Op\left(\frac{1}{\sqrt{N}} + \frac{a_n}{\sqrt{N}} + \frac{1}{a_n}\right).
\end{align*}
The first inequality uses Lipschitz continuity of $f$ and property of the max operator. The second inequality applies triangle inequality to the norm $\Vert\cdot\Vert_\infty$. The rest is essentially the same as that of the first part. 

The only missing part is to show the existence of the coupling variable $\tilde{\bz}$. Since the choice probabilities are averages of indicators, Corollary 4.1 of \citet*{Chen-Goldstein-Shao_2010_Book} implies the following non-asymptotic bound on the Wasserstein metric:
\[
\inf \left\{\Expectation\left[ \left| X - Y \right| \right]:\ X\sim \sqrt{N}(\hat{\bChoiProb}-\bChoiProb),\ Y\sim \mathcal{N}(\mathbf{0}, \bOmega_{\ChoiProb})\right\}\leq \text{Const.}\sqrt{\frac{1}{\min_{S\in\mathcal{X}}N_S}},
\]
where the infimum is taken over all joint distributions with the given marginals, and the constant in the above display is universal. By Assumption 2 in the main paper, the rate on the RHS is proportional to $\sqrt{1/N}$. Existence of the coupling variable follows from the bounds on the Wasserstein metric. 

Hence we showed that 
\[
\left| \Expectation_f\Big[ \sqrt{N}\max\Big(\bR_{\succ}\hat{\bChoiProb}\oslash \hat{\bsigma}_\succ \Big)_+ \Big] -  \Expectation_f^\star\Big[ \sqrt{N}\max\Big(\bR_{\succ}(\hat{\bOmega}^{1/2}\bw^\star+\bChoiProb)\oslash\hat{\bsigma}_{\succ}\Big)_+ \Big]\right| = \op(1). 
\]
\end{proof}

Now we show how the previous result can be generalized to be uniform among a class of distributions. The main argument used in Proposition \ref{Proposition: distributional approx} remains almost unchanged.

\begin{prop}\label{Proposition: uniform distributional approx}
Under the assumptions of Theorem 4 in the main paper,
\[
\sup_{\ChoiProb\in\Pi}\Prob\left[\Big| \Prob\Big[\Tstat(\succ) \leq t \Big] - \Prob^\star\Big[\sqrt{N}\max( \bR_{\succ}(\bz^\star + \bChoiProb) \oslash \hat{\bsigma}_\succ )_+ \leq t  \Big] \Big| > \epsilon\right] \to 0,\qquad \forall \epsilon>0,\ \forall t\neq 0.
\]
\end{prop}

\begin{proof}
The proof remains almost the same, while extra care is needed for the coupling argument, which we demonstrate below. We would like to bound the following quantity \emph{uniformly}:
\[
\Big| \Expectation_f\Big[ \sqrt{N}\max\Big(\bR_{\succ}\hat{\bChoiProb}\oslash\bsigma_{\ChoiProb,\succ}\Big)_+ \Big]-\Expectation_f^\star\Big[ \sqrt{N}\max\Big(\bR_{\succ}(\bOmega_{\ChoiProb}^{1/2}\bw^\star+\bChoiProb)\oslash\bsigma_{\ChoiProb,\succ}\Big)_+ \Big] \Big|,
\]
where $f$ is a bounded function with Lipschitz constant 1. By the coupling argument in the proof of Proposition \ref{Proposition: distributional approx}, it is possible to construct, for each $\ChoiProb\in\Pi$, a random variable $\tilde{\bz}_\ChoiProb\sim \mathcal{N}(\mathbf{0},\bOmega_{\ChoiProb}/N)$ (the covariance matrix $\bOmega_{\ChoiProb}$ depends on $\ChoiProb$, indicated by the subscript), such that 
\[
\Expectation\left[\left| \sqrt{N}(\hat{\bChoiProb}-\bChoiProb)  - \sqrt{N}\tilde{\bz}_\ChoiProb \right|\right] \leq \text{Const.}\sqrt{\frac{1}{\min_{S\in\mathcal{X}} N_S}}. 
\]
Then we can bound the aforementioned quantity by
\begin{align*}
&\ \Big| \Expectation_f\Big[ \sqrt{N}\max\Big(\bR_{\succ}\hat{\bChoiProb}\oslash\bsigma_{\ChoiProb,\succ}\Big)_+ \Big]-\Expectation_f^\star\Big[ \sqrt{N}\max\Big(\bR_{\succ}(\bOmega_{\ChoiProb}^{1/2}\bw^\star+\bChoiProb)\oslash\bsigma_{\ChoiProb,\succ}\Big)_+ \Big] \Big|\\
&\leq \Big| \Expectation_f\Big[ \sqrt{N}\max\Big(\bR_{\succ}\hat{\bChoiProb}\oslash\bsigma_{\ChoiProb,\succ}\Big)_+ \Big]-\Expectation_f\Big[ \sqrt{N}\max\Big(\bR_{\succ}(\tilde{\bz}_\ChoiProb+\bChoiProb)\oslash\bsigma_{\ChoiProb,\succ}\Big)_+ \Big]\Big|\\
&\qquad + \Big|\Expectation_f\Big[ \sqrt{N}\max\Big(\bR_{\succ}(\tilde{\bz}_\ChoiProb+\bChoiProb)\oslash\bsigma_{\ChoiProb,\succ}\Big)_+ \Big]-\Expectation_f^\star\Big[ \sqrt{N}\max\Big(\bR_{\succ}(\bOmega_{\ChoiProb}^{1/2}\bw^\star+\bChoiProb)\oslash\bsigma_{\ChoiProb,\succ}\Big)_+ \Big] \Big|\\
&= \Big| \Expectation_f\Big[ \sqrt{N}\max\Big(\bR_{\succ}\hat{\bChoiProb}\oslash\bsigma_{\ChoiProb,\succ}\Big)_+ \Big]-\Expectation_f\Big[ \sqrt{N}\max\Big(\bR_{\succ}(\tilde{\bz}_\ChoiProb+\bChoiProb)\oslash\bsigma_{\ChoiProb,\succ}\Big)_+ \Big]\Big|,
\end{align*}
where the second line comes from triangle inequality, and the equality uses the fact that $\tilde{\bz}_\ChoiProb$ and $\bOmega_{\ChoiProb}^{1/2}\bw^\star$ have the same distribution. Hence we have the bound:
\[
\text{previous display} \precsim \sqrt{\frac{1}{\min_{S\in\mathcal{X}} N_S}},
\]
with the estimate in the above display being uniformly valid over $\ChoiProb\in\Pi$. Hence the proof reduces to bound the probability (choose some $\epsilon>0$ small enough):
\begin{align*}
&\ \sup_{\ChoiProb\in\Pi} \Prob\left[ \min_{S\in\mathcal{X}} N_S < \epsilon N \right] = \sup_{\ChoiProb\in\Pi} \Prob\left[  \frac{N_S}{N} < \epsilon,\ \exists S \right] \leq \sup_{\ChoiProb\in\Pi}\sum_{S\in\mathcal{X}}\Prob\left[  \frac{N_S}{N} < \epsilon \right]\\
&= \sup_{\ChoiProb\in\Pi}\sum_{S\in\mathcal{X}}\Prob\left[  \frac{N_S-\Prob[Y_i=S]N}{N} < \epsilon-\Prob[Y_i=S] \right]\\
&\leq  \sup_{\ChoiProb\in\Pi}\sum_{S\in\mathcal{X}}\exp\left( -(\epsilon-\Prob[Y_i=S])^2\frac{N }{2} \right) \to 0.
\end{align*}
For the inequality in the first line, we use union bound. The last line uses Hoeffding's inequality, which is valid for any $\epsilon$ smaller than $\inf_{\ChoiProb\in\Pi}\min_{S\in\mathcal{X}}\Prob[Y_i=S]\geq \underline{p}$ (see Assumption 2 in the main paper). Hence we demonstrated that 
\[
\sqrt{\frac{1}{\min_{S\in\mathcal{X}} N_S}} \asymp_\Prob \sqrt{\frac{1}{N}},\qquad \text{uniformly in $\Pi$}.
\]
\end{proof}

Now we demonstrate how Theorem 4 in the main paper follows from the previous proposition. 

Recall that in constructing $\Tstat^\star(\succ)$, we use $\kappa_N^{-1} (\bR_\succ \hat{\bChoiProb})_-$ to replace the unknown $\bR_\succ \bChoiProb$. This unknown quantity is bounded above by $0$ uniformly in $\Pi$, since this class only contains choice rules that are compatible with $\succ$. At the same time, $\kappa_N^{-1}(\bR_\succ \hat{\bChoiProb})_-$ converges to 0 in probability uniformly for the class $\Pi$. Therefore asymptotically $\Tstat^\star(\succ)$ stochastically dominates $\Tstat(\succ)$ uniformly in $\Pi$, which proves Theorem 4 in the main paper.  

\subsection{Computation}

The number of rows in the constraint matrix is 
\[
\text{\#row}(\bR_\succ) = \sum_{S\in\mathcal{X}} \sum_{a,b\in S}\Indicator(b\prec a) = \sum_{S\in\mathcal{X},\ |S|\geq 2} \binom{|S|}{2} = \sum_{k=2}^{K} \binom{K}{k}\binom{k}{2},
\]
where $K = |X|$ is the number of alternatives in the grand set $X$. Not surprisingly, the number of constraints increases very fast with the size of the grand set. However, we note that as long as the matrix $\bR_{\succ}$ has been constructed for one preference $\succ$, constraint matrices for other preference orderings can be obtained by column permutations of $\bR_{\succ}$. As a result, even though the construction of $\bR_{\succ}$ might be computationally demanding with many alternatives available (i.e., when $K$ is large), this step (Algorithm 1 in the main paper) only needs to be implemented once. This is particular useful and saves computation if there are multiple hypotheses to be tested. 

We showcase how the proposed inference procedure is computationally attractive, especially when inference is conducted for many preferences. Table \ref{Table: Computing Time} records the computing time needed using our companion \texttt{R} package \texttt{ramchoice}. Note that the time reflects not only the time needed to construct the constraint matrices, but also 2,000 simulations from a multivariate normal distribution to obtain critical values. 

\begin{table}[!ht]
\centering
\renewcommand{\arraystretch}{1.1}
\caption{Computing Time ($K=6$ and $K!=720$)}\label{Table: Computing Time}
\begin{tabular}{rr|rr}
\hline
Num. of & Time & Num. of & Time \\
Preferences  & \phantom{aaaaa}(seconds) & Preferences  & \phantom{aaaaa}(seconds)\\ \hline
1		& 1.117			    & 50 	    & 21.185 \\ 
5		& 2.668			    & 100	    & 40.422 \\
10		& 4.584	    		& 400	    & 195.907 \\ 
20		& 8.948 			& 720       & 407.177 \\ \hline
\end{tabular}

{\footnotesize
\begin{flushleft}
There are $K=6$ alternatives, leading to potentially $K!=720$ preference orderings. All choice problems are observable in the data. For each preference, there are 240 inequality constraints in the matrix $\bR_\succ$. The sample size is $N=12,600$, and $M=2,000$ simulations are used to construct critical value. System: MacOS 10.13.1, R 3.4.1. 
\end{flushleft}}
\end{table}

Our general-purpose implementation executes very fast even when $720$ different preference orderings are tested, and given that there are 240 constraints for testing one single preference. The full execution takes less than $7$ minutes. Moreover, since the major computing time comes from constructing the constraint matrix $\bR_\succ$ when testing many preferences, which involves looping over all choice problems and alternatives, it is possible to further speed up our general-purpose implementation by employing low-level programming languages such as \texttt{C++}, which are faster for simple for-loop structures.

\subsection{Connection to J-tests}

The testing methods we proposed are connected to a class of inference strategies based on projection residuals, sometimes also known as J-tests. To describe this alternative inference approach, recall that any preference induces a surjective map from attention rules to choice rules (Definition 3 in the main paper), denoted by $\tilde{\bC}_\succ$. Then, by our definition of partially identified preferences, $\succ\ \in \ParSet_{\ChoiProb}$ if and only if $\bChoiProb$ belongs to $\{\tilde{\bC}_\succ\bmu:\ \bR\bmu\leq 0\}$, where $\bR$ represents the monotonicity assumption imposed on attention rules. Therefore, inference can be based on
\[
\succ\ \in\ParSet_{\ChoiProb}\quad \text{if and only if}\quad \inf_{\bmu:\ \bR\bmu\leq 0} ( \bChoiProb - \tilde{\bC}_\succ \bmu )^\prime \bW ( \bChoiProb - \tilde{\bC}_\succ \bmu )=0,
\]
where $\bW$ is some positive definite weighting matrix. Hence, a preference compatible with $\ChoiProb$ is equivalent to a zero residual from projecting $\bChoiProb$ to the corresponding set $\{\tilde{\bC}_\succ\bmu:\ \bR\bmu\leq 0\}$. For example, this strategy is used by \cite{Kitamura-Stoye_2018_ECMA} in random utility models.

To see the connection of J-tests to moment inequality testing, observe that if the mapping defined by $\tilde{\bC}_\succ$ is invertible, then the above reduces to
\[
\succ\ \in\ParSet_{\ChoiProb}\quad \text{if and only if}\quad \bR\tilde{\bC}_\succ^{-1}\bChoiProb\leq 0.
\]
Such reduction may not be feasible analytically, or it may be numerically prohibitive. With a careful inspection of our problem, we showed that it is without loss of generality to focus on triangular attention rules, which can be uniquely constructed from preferences and choice rules. That is, we showed how the inversion $\tilde{\bC}_\succ^{-1}$ is constructed, which is denoted by $\bC_\succ$ in Theorem 3 in the main paper and its proof. Moreover, we provided an algorithm which constructs the constraint matrix directly, avoiding the detour of forming it as the product of two matrices. 

Compared to directly testing inequality constraints as we propose, employing a J-test inference procedure has two potential drawbacks. First, the J-test statistic is essentially the (weighted) Euclidean norm of the projection residual, which may suffer from low power if only a few inequalities are violated. Second, constructing the projection residual requires numerical minimization, which can be computationally costly especially when the dimension of $\bmu$ is nontrivial. This is apparent from Table \ref{Table: Computing Time}: testing a single preference takes about one second and testing for all $720$ preference orderings takes about $7$ minutes with our procedure, while employing the J-test can easily take a prohibitive amount of time because of the costly numerical optimization step over a possibly high-dimensional parameter space and the fact that this numerical optimization has to be done multiple times to construct a critical value. For example, in the same setting of Table \ref{Table: Computing Time}, employing the J-test with $2,000$ bootstraps takes about $90$ minutes for just one single preference when employing the quadratic programming procedure \texttt{quadprog} in \texttt{Matlab} (R2016b).

\subsection{Implementation with Limited Data}\label{SA-subsection: imcomplete data econometrics}

In some cases not all choice problems $S\in\mathcal{X}$ may be observed for two reasons. First, some choice problems are ruled out \textit{a priori} in the population due to, for instance, capacity or institutional constraints. Second, certain choice problems may not be available in a given finite sample due to sampling variability. Even if all choice problems are observed in the data, some may have only a few appearances, and usually are dropped from empirical analysis to avoid dealing with small (effective) sample sizes. We descibe how our econometric methods (and assumptions) can also be adapted to situations of limited data. 

Recall that $\mathcal{S}\subset \mathcal{X}$ denotes the collection of all observable choice problems. From an econometric perspective, this can also be seen as the collection of realized choice problems in the data. Assumption 2 in the main paper now takes the following form. 

\begin{Assumption}[DGP with Limited Data]\label{Assumption 2: iid sampling limited data}
The data is a random sample of choice problems $Y_i$ and corresponding choices $y_i$, $\{(y_i, Y_i):\ y_i\in Y_i,\ 1\leq i\leq N\}$, generated by the underlying choice rule $\Prob[y_i=a| Y_i=S]=\ChoiProb(a|S)$, and that $\Prob[Y_i=S]\geq \underline{p}>0$ for all $S\in\mathcal{S}$.
\end{Assumption}

The most critical concern is on Assumption 1 in the main paper, which directly affects how the constraint matrix $\bR_\succ$ is constructed. We state a seemingly ``stronger'' version of that assumption. 

\begin{Assumption}[Monotonic Attention with Limited Data]\label{Assumption: monotonicity limited}
For any $A \subset S-T$, $ \mu(T|S) \leq \mu(T|S-A).$
\end{Assumption}

When $\mathcal{S}=\mathcal{X}$, this assumption is equivalent to Assumption 1 in the main paper. To see this fact, pick an arbitrary $A\subset S-T$ and let $a_1,a_2,\cdots$ be an enumeration of elements of $A$. Then, $\mu(T|S)-\mu(T|S-A)$ can be written as $\mu(T|S) - \mu(T|S-a_1)  + \sum_{j\geq 1} [\mu(T|S-a_1-\cdots-a_j) - \mu(T|S-a_1-\cdots-a_j-a_{j+1})]$, where by Assumption 1 in the main paper each summand is nonpositive, hence Assumption 1 in the main paper implies Assumption \ref{Assumption: monotonicity limited}. The other direction is obvious: Assumption \ref{Assumption: monotonicity limited} implies Assumption 1 in the main paper.

If $\mathcal{S}$ does not contain all possible choice problems, however, Assumption \ref{Assumption: monotonicity limited} provides a proper notion of monotonicity. To see the extra identification and statistical power implied by Assumption \ref{Assumption: monotonicity limited}, two examples are provided after we give a proper definition of compatible preferences and the algorithm of constructing the constraint matrix in this context.  

With limited data, unfortunately, there are two ways to generalize the previous definition, which are not equivalent. Recall that $\ChoiProb_{\text{obs}}$ denotes the observed choice rule, defined only on $\mathcal{S}$, while we reserve the notation $\ChoiProb$ to a choice rule that is defined on $\mathcal{X}$. Following is the first version. 

\begin{defn}[Compatible Preferences with Limited Data, I]\label{Definition: identified preferences I}
Let $\ChoiProb_{\text{obs}}$ be the underlying choice rule/data generating process. A preference $\succ$ is compatible with $\ChoiProb_{\text{obs}}$ on $\mathcal{S}$, denoted by $\succ\ \in\ParSet_{\ChoiProb_\text{obs}}$, if $(\ChoiProb_{\text{obs}},\succ)$ is a RAM on $\mathcal{S}$. 
\end{defn}
The other version is the following:
\begin{defn}[Compatible Preferences with Limited Data, II]\label{Definition: identified preferences II}
Let $\ChoiProb_{\text{obs}}$ be the underlying choice rule/data generating process. A preference $\succ$ is compatible with $\ChoiProb_{\text{obs}}$ on $\mathcal{X}$, denoted by $\succ\ \in\Xi_{\ChoiProb_\text{obs}}$, if $\ChoiProb_{\text{obs}}$ has an extension $\ChoiProb$ to $\mathcal{X}$ such that $(\ChoiProb,\succ)$ is a RAM on $\mathcal{X}$. 
\end{defn}
When $\mathcal{S}=\mathcal{X}$, the two definitions agree, and also agree with the definition given in the main paper. If $\mathcal{S}$ is a proper subset of $\mathcal{X}$, however, $\ParSet_{\ChoiProb_{\text{obs}}}$ can be larger than $\Xi_{\ChoiProb_\text{obs}}$. Depending on the goal of analysis, both can be of interest, while $\Xi_{\ChoiProb_\text{obs}}$ is much more difficult to characterize empirically. 

The following algorithm generalizes the one given in the main text to the limited data scenario. The constraint matrix provided in this algorithm can be used to characterize $\ParSet_{\ChoiProb_{\text{obs}}}$. 

\renewcommand{\thealgorithm}{\arabic{algorithm}'}\setcounter{algorithm}{0}
\begin{algorithm}
\caption{Construction of $\bR_\succ$ with Limited Data.}
\begin{algorithmic} 
\REQUIRE Set a preference $\succ$.
\STATE $\bR_\succ \leftarrow$ empty matrix
\FOR{$S$ in $\mathcal{S}$} \FOR{$T$ in $\mathcal{S}$ and $T$ $\subset$ $S$} \FOR{$a\prec S-T$}
\STATE $\bR_\succ \leftarrow$ add row corresponding to $\ChoiProb_{\text{obs}}(a|S)-\ChoiProb_{\text{obs}}(a|T)\leq 0$.
\ENDFOR\ENDFOR\ENDFOR
\end{algorithmic}
\end{algorithm}

Consider now two examples with limited data. 

\begin{example}[Limited Data and Identification]
Assume $\mathcal{S}=\{  \{a,b,c,d\},\{a,b\} \} $, and we are interested in the null hypothesis that $c\succ b\succ d\succ a$. Assumption 1 in the main paper does not impose any constraint, since it only requires comparing choice problems that differ by one element. On the other hand, \ref{Assumption: monotonicity limited} implies the constraint that $\pi(a|\{a,b,c,d\})\leq \pi(a|\{a,b\})$, since violating this constraint is equivalent to $a\succ c$ or $a\succ d$, which is incompatible with our hypothesis. 
\end{example}

\begin{example}[Limited Data and Statistical Power]
Consider $\mathcal{S}=\{ \{a,b,c,d\}, \{a,b,c\}, \{a,b\} \}$. This is one example of limited data, but due to the special structure of $\mathcal{S}$, constraints implied by the two assumptions are equivalent \textit{in population}. Consider the preference $a\succ c\succ d\succ b\succ$, then Assumption 1 in the main paper gives two constraints, (i) $\pi(b|\{a,b,c,d\})\leq \pi(b|\{a,b,c\})$ and (ii) $\pi(b|\{a,b,c\})\leq \pi(b|\{a,b\})$. Assumption \ref{Assumption: monotonicity limited}, however, will give the \textit{extra} condition (iii) $\pi(b|\{a,b,c,d\})\leq \pi(b|\{a,b\})$. In population, this extra constraint is redundant, since it is not possible to violate (iii) without violating at least one of (i) and (ii). 

When applied to finite sample, however, the extra condition (iii) is no longer redundant, and may help to improve statistical power. To see this, assume (i) is violated by $\delta>0$, a small margin, so is (ii). Then it is hard to reject any of them with finite sample. On the other hand, (iii) combines (i) and (ii), hence is violated by $2\delta$, which is much easier to detect in a finite sample.  
\end{example}

Test statistics and critical values are constructed in a similar way based on $\bR_\succ \hat{\bChoiProb}_{\text{obs}}$, hence not repeated here. We provide an analogous result of Theorem 4 in the main paper. 

\begin{thm}[Validity of Critical Values with Limited Data I]\label{Theorem: valid critical values limited data}
Assume Assumption \ref{Assumption 2: iid sampling limited data} holds. Let $\Pi_{\text{obs}}$ be a class of choice rules restricted to $\mathcal{S}$, and $\succ$ a preference, such that: (i) for each $\ChoiProb_{\text{obs}}\in\Pi_{\text{obs}}$, $\succ\,\in\ParSet_{\ChoiProb_{\text{obs}}}$; and (ii) $\inf_{\ChoiProb_{\text{obs}}\in\Pi_{\text{obs}}}\min(\bsigma_{\ChoiProb_{\text{obs}},\succ})>0$. Then
\[
\limsup_{N\to\infty} \sup_{\ChoiProb_{\text{obs}}\in\Pi_{\text{obs}}}\Prob\left[ \Tstat(\succ) > c_\alpha(\succ) \right] \leq \alpha.
\]
\end{thm}

One natural question is whether it is possible to make any claim on $\Xi_{\ChoiProb_\text{obs}}$ with limited data. Indeed, the three tests remain valid when applied to $\Xi_{\ChoiProb_\text{obs}}$ (i.e. controls size uniformly), which can be easily seen from the fact that $\Xi_{\ChoiProb_\text{obs}}\subset \ParSet_{\ChoiProb_{\text{obs}}}$. We provide the following theorem. Also, for a choice rule $\ChoiProb$ defined on $\mathcal{X}$, $\ChoiProb_{\text{obs}}$ denotes its restriction to $\mathcal{S}$. 

\begin{thm}[Validity of Critical Values with Limited Data II]\label{Theorem: valid critical values limited data II}
Assume Assumption \ref{Assumption 2: iid sampling limited data} holds. Let $\Pi$ be a class of choice rules on $\mathcal{X}$, and $\succ$ a preference, such that: (i) for each $\ChoiProb\in\Pi$, $\succ\,\in\Xi_{\ChoiProb_\text{obs}}$; and (ii) $\inf_{\ChoiProb_{\text{obs}}\in\Pi_{\text{obs}}}\min(\bsigma_{\ChoiProb_{\text{obs}},\succ})>0$. Then
\[
\limsup_{N\to\infty} \sup_{\ChoiProb\in\Pi}\Prob\left[ \Tstat(\succ) > c_\alpha(\succ) \right] \leq \alpha.
\]
\end{thm}
In Theorem \ref{Theorem: valid critical values limited data II}, we only need positive variance for moment conditions corresponding to the observed components of the choice rule. The reason is simple: In constructing the tests and critical values, we never use the unobserved part $\bChoiProb-\bChoiProb_{\text{obs}}$. 

What is lost from Theorem \ref{Theorem: valid critical values limited data} to Theorem \ref{Theorem: valid critical values limited data II}? With limited data, there may exist $\succ\,\in\ParSet_{\ChoiProb_{\text{obs}}}-\Xi_{\ChoiProb_\text{obs}}$ which is not rejected by the test statistic $\Tstat(\succ)$ asymptotically.

\subsection{Other Critical Values\label{appendix: other critical values}}

There are many proposals for constructing critical values in the literature of testing moment inequalities \citep[and references therein]{Canay_Shaikh_2017_ChBook,Ho_Rosen_2017_ChBook,Molinari_2019_Handbook}. Here we discuss some of them for completeness. Throughout, let $\bz^\star$ denote a random vector independent of the original data and $\bz^\star\thicksim\mathcal{N}(\mathbf{0},\ \hat{\bOmega}/N)$. To save notation, define $(x)_+$ as the positive parts of $x$, and $(x)_-$ as the negative parts multiplied by $-1$ (truncation above at zero).

\subsubsection*{Plug-in Method}

The first method simply plugs-in an estimate of $\bR_{\succ}\bChoiProb$ subject to the non-positivity constraint. Define
\[
\Tstat^\star_{\mathtt{PI}}(\succ) = \sqrt{N}\cdot \max\Big((\bR_{\succ}\bz^\star + (\bR_{\succ}\hat{\bChoiProb})_-) \oslash \hat{\bsigma}_\succ\Big)_+,
\]
then the critical value with level $\alpha$ with the plug-in method is
\[
c_{\alpha,\mathtt{PI}}(\succ) = \inf\Big\{ t:\ \Prob^\star\left[\Tstat^\star_{\mathtt{PI}}(\succ) \leq t\right]\geq 1-\alpha \Big\},
\]
where $\Prob^\star$ denotes probability operator conditional on the data, and in practice, it is replaced by the simulated average $M^{-1}\sum_{m=1}^M\Indicator(\cdot)$, where recall that $\Tstat^\star_{\mathtt{PI}}(\succ)$ is simulated $M$ times. 

We note the critical values obtained here are \emph{not} uniformly valid in the sense of Theorem 4 in the main paper. 

\subsubsection*{Least Favorable Model}

The critical values are nondecreasing in the centering, hence a conservative method is to consider the least favorable model, $\bR_{\succ}\bChoiProb=0$, which assumes all the moment inequalities are binding. That is, 
\[
\Tstat^\star_{\mathtt{LF}}(\succ)  = \sqrt{N}\cdot \max\Big((\bR_{\succ}\bz^\star) \oslash \hat{\bsigma}_\succ\Big)_+,
\]
and
\[
c_{\alpha,\mathtt{LF}}(\succ) = \inf\Big\{ t:\ \Prob^\star\left[\Tstat^\star_{\mathtt{LF}}(\succ) \leq t\right]\geq 1-\alpha \Big\}.
\]

Undoubtedly, using such critical value may severely decrease the power of the test, especially when one or two moment restrictions are violated and the rest are far from binding. Next we introduce two methods seeking to improve power property of the test. The first one relies on the idea of moment selection, and the third one replaces the unknown moment conditions with upper bounds.

\subsubsection*{Two-step Moment Selection} 

To illustrate this method, we first represent $\bR_{\succ}\hat{\bChoiProb}$ by its individual components, as $\{\br_{\succ, \ell}^\Trans\hat{\bChoiProb}:\ 1\leq \ell\leq L\}$, where $\br_{\succ, \ell}^\Trans$ is the $\ell$-th row of $\bR_{\succ}$, so that there are in total $L$ moment inequalities. The first step is to conduct moment selection. Let $0<\beta<\alpha/3$, and the following set of indices of ``active moment restrictions'':
\[
\mathcal{L} = \left\{ \ell:\ \sqrt{N}\cdot\frac{\br_{\succ, \ell}^\Trans\hat{\bChoiProb}}{\hat{\sigma}_{\succ,\ell}} \geq -2\cdot c_{\beta,\mathtt{LF}}(\succ)  \right\},
\]
then
\begin{align*}
\Tstat^\star_{\mathtt{MS}}(\succ) &= \sqrt{N}\cdot \max_{\ell\in\mathcal{L}}\left(\frac{\br_{\succ, \ell}^\Trans \bz^\star }{\hat{\sigma}_{\succ,\ell}} \right)_+,
\end{align*}
and the critical value is computed as
\begin{align*}
c_{\alpha,\mathtt{MS}}(\succ) &= \inf\Big\{ t:\ \Prob^\star\left[\Tstat^\star_{\mathtt{MS}}(\succ) \leq t\right]\geq 1-\alpha+2\beta \Big\}.
\end{align*}

Constructing the coupling statistic $\Tstat^\star(\mathcal{P}_\forall)$ requires more delicate work. Note that simply taking maximum of individual $\Tstat^\star_{\mathtt{MS}}(\succ)$ does not work, mainly due to the two-step nature of the construction. More precisely, the first step moment selection controls error probability to be $\beta$ for each individual preference, but not jointly, and the reason is that we used $c_{\beta,\mathtt{LF}}(\succ)$ for moment selection, which is too small jointly for a collection of preferences. 

For a collection of preferences, the correct moment selection is the following:
\[
\mathcal{L}_{\succ,\mathcal{P}} = \left\{ \ell:\ \sqrt{N}\cdot\frac{\br_{\succ, \ell}^\Trans\hat{\bChoiProb}}{\hat{\sigma}_{\succ,\ell}} \geq -2\cdot c_{\beta,\mathtt{LF}}(\mathcal{P}_\forall)  \right\},\qquad \text{for each}\ \succ\in\mathcal{P}.
\]
Now we use a much larger critical value for moment selection: $c_{\beta,\mathtt{LF}}(\mathcal{P}_\forall)$ instead of $c_{\beta,\mathtt{LF}}(\succ)$. Then the coupling statistic is
\[
\Tstat^\star_{\mathtt{MS}}(\mathcal{P}_\forall) = \sqrt{N}\cdot \max_{\succ\in\mathcal{P}} \max_{\ell\in\mathcal{L}_{\succ,\mathcal{P}}}\left(\frac{\br_{\succ, \ell}^\Trans \bz^\star }{\hat{\sigma}_{\succ,\ell}} \right)_+,
\]
then the critical value is computed accordingly as the $1-\alpha$ quantiles:
\begin{align*}
c_{\alpha,\mathtt{MS}}(\mathcal{P}_\forall) &= \inf\Big\{ t:\ \Prob^\star\left[\Tstat^\star_{\mathtt{MS}}(\mathcal{P}_\forall) \leq t\right]\geq 1-\alpha+2\beta \Big\}.
\end{align*}

See, for example, \citet*{Chernozhukov-Chetverikov-Kato_2019_RESTUD} and references therein.

\subsubsection*{Two-step Moment Upper Bounding}

This method uses a first step to construct a confidence region for $\bR_{\succ}\bChoiProb$, and the upper bound of such region is used as a conservative estimate for $\bR_{\succ}\bChoiProb$. Let $0<\beta<\alpha$, and
\[
\Tstat^\star_{\mathtt{UB}}(\succ) = \sqrt{N}\cdot \max\Big((\bR_{\succ}\bz^\star + (\bR_{\succ}\hat{\bChoiProb} + c_{\beta,\mathtt{LF}}(\succ) \hat{\bsigma}_\succ/\sqrt{N})_-) \oslash \hat{\bsigma}_\succ\Big)_+,
\]
then 
\[
c_{\alpha,\mathtt{UB}}(\succ) = \inf\Big\{ t:\ \Prob^\star\left[\Tstat^\star_{\mathtt{UB}}(\succ) \leq t\right]\geq 1-\alpha+\beta \Big\}.
\]
Note that in the first step, we use the critical values from the least favorable method to construct upper bounds on the moment inequalities $\bR_{\succ}\hat{\bChoiProb} + c_{\beta,\mathtt{LF}}(\succ) \hat{\bsigma}/\sqrt{N}$, which is guaranteed to have coverage $1-\beta$. Then the significance level in the second step is adjusted to account for errors incurred in the first step. See, e.g., \citet*{Romano-Shaikh-Wolf_2014_ECMA} for further details. In particular, they recommend to use $\beta/\alpha=0.1$. 

Same as the moment selection method, the upper bound (or confidence region for $\bR_{\succ}\bChoiProb$) constructed above only controls error probability for individual preference, but not jointly for a collection of preferences. Hence we need to make further adjustments. The solution is almost the same: replace the critical value used for constructing upper bounds by one that controls error probability jointly for the collection $\mathcal{P}$:
\begin{align*}
\Tstat^\star_{\mathtt{UB}}(\mathcal{P}_\forall) &= \sqrt{N}\cdot \max_{\succ\in\mathcal{P}}\max\Big((\bR_{\succ}\bz^\star + (\bR_{\succ}\hat{\bChoiProb} + c_{\beta,\mathtt{LF}}(\mathcal{P}_\forall) \hat{\bsigma}_\succ/\sqrt{N})_-) \oslash \hat{\bsigma}_\succ\Big)_+,
\end{align*}
then the critical value is computed accordingly as the $1-\alpha$ quantiles:
\begin{align*}
c_{\alpha,\mathtt{UB}}(\mathcal{P}_\forall) &= \inf\Big\{ t:\ \Prob^\star\left[\Tstat^\star_{\mathtt{UB}}(\mathcal{P}_\forall) \leq t\right]\geq 1-\alpha+\beta \Big\}.
\end{align*} 

\section{Additional Simulation Results}

We first recall the simulation setup. We consider a class of logit attention rules indexed by $\varsigma$:
\[
    \AttFilter_{\varsigma}(T|S) = \frac{w_{T,\varsigma}}{\sum_{T'\subset S} w_{T',\varsigma}},\qquad w_{T,\varsigma} = |T|^{\varsigma},
\]
where $|T|$ is the cardinality of $T$. In this Supplemental Appendix, we give simulation evidence for $\varsigma\in\{0,1\}$. 

As in the main paper, we list five hypotheses (preference orderings), and whether they are compatible with our RAM model and specific values of $\phi$. 

\begin{center}
\begin{tabular}{c|c|cccccccccc}
\hline
$\varsigma=0$&\multicolumn{11}{c}{$\phi$}\\ 
\cline{2-12}
& 1 & .95 & .90 & .85 & .80 & .75 & .70 & .65 & .60 & .55 & .50\\ \hline
$\mathsf{H}_{0,1}:a_1\succ a_2\succ a_3\succ a_4\succ a_5$ &$\checkmark$&$\checkmark$&$\checkmark$&$\checkmark$&$\checkmark$&$\checkmark$&$\checkmark$&$\checkmark$&$\checkmark$&$\checkmark$&$\checkmark$ \\
$\mathsf{H}_{0,2}:a_2\succ a_3\succ a_4\succ a_5\succ a_1$ &$\checkmark$&$\checkmark$&$\checkmark$&$\checkmark$&$\checkmark$&$\checkmark$&$\checkmark$&$\times$&$\times$&$\times$&$\times$ \\
$\mathsf{H}_{0,3}:a_3\succ a_4\succ a_5\succ a_2\succ a_1$ &$\checkmark$&$\checkmark$&$\checkmark$&$\checkmark$&$\checkmark$&$\checkmark$&$\checkmark$&$\times$&$\times$&$\times$&$\times$ \\
$\mathsf{H}_{0,4}:a_4\succ a_5\succ a_3\succ a_2\succ a_1$ &$\checkmark$&$\checkmark$&$\checkmark$&$\checkmark$&$\checkmark$&$\checkmark$&$\checkmark$&$\times$&$\times$&$\times$&$\times$ \\
$\mathsf{H}_{0,5}:a_5\succ a_4\succ a_3\succ a_2\succ a_1$ &$\checkmark$&$\checkmark$&$\checkmark$&$\checkmark$&$\checkmark$&$\checkmark$&$\checkmark$&$\times$&$\times$&$\times$&$\times$ \\ \hline
\end{tabular}
\end{center}

\begin{center}
\begin{tabular}{c|c|cccccccccc}
\hline
$\varsigma=1$&\multicolumn{11}{c}{$\phi$}\\ 
\cline{2-12}
& 1 & .95 & .90 & .85 & .80 & .75 & .70 & .65 & .60 & .55 & .50\\ \hline
$\mathsf{H}_{0,1}:a_1\succ a_2\succ a_3\succ a_4\succ a_5$ &$\checkmark$&$\checkmark$&$\checkmark$&$\checkmark$&$\checkmark$&$\checkmark$&$\checkmark$&$\checkmark$&$\checkmark$&$\checkmark$&$\checkmark$ \\
$\mathsf{H}_{0,2}:a_2\succ a_3\succ a_4\succ a_5\succ a_1$ &$\checkmark$&$\checkmark$&$\checkmark$&$\checkmark$&$\checkmark$&$\checkmark$&$\times$&$\times$&$\times$&$\times$&$\times$ \\
$\mathsf{H}_{0,3}:a_3\succ a_4\succ a_5\succ a_2\succ a_1$ &$\checkmark$&$\checkmark$&$\checkmark$&$\checkmark$&$\checkmark$&$\checkmark$&$\times$&$\times$&$\times$&$\times$&$\times$ \\
$\mathsf{H}_{0,4}:a_4\succ a_5\succ a_3\succ a_2\succ a_1$ &$\times$&$\times$&$\times$&$\times$&$\times$&$\times$&$\times$&$\times$&$\times$&$\times$&$\times$ \\
$\mathsf{H}_{0,5}:a_5\succ a_4\succ a_3\succ a_2\succ a_1$ &$\times$&$\times$&$\times$&$\times$&$\times$&$\times$&$\times$&$\times$&$\times$&$\times$&$\times$ \\ \hline
\end{tabular}
\end{center}

With $\varsigma=0$ being the data generating process, we have the identified set to be the set of all preferences, and for $\varsigma=1$, it is $\{ \succ:a_3\succ a_4\succ a_5 \}$. Simulation results are collected in Figure \ref{fig:figSA-1} and \ref{fig:figSA-2}.

\begin{figure}[!ht]
\centering
\subfloat[$\mathsf{H}_{0,1}:a_1\succ a_2\succ a_3\succ a_4\succ a_5$]{
\resizebox{0.4\columnwidth}{!}{\includegraphics{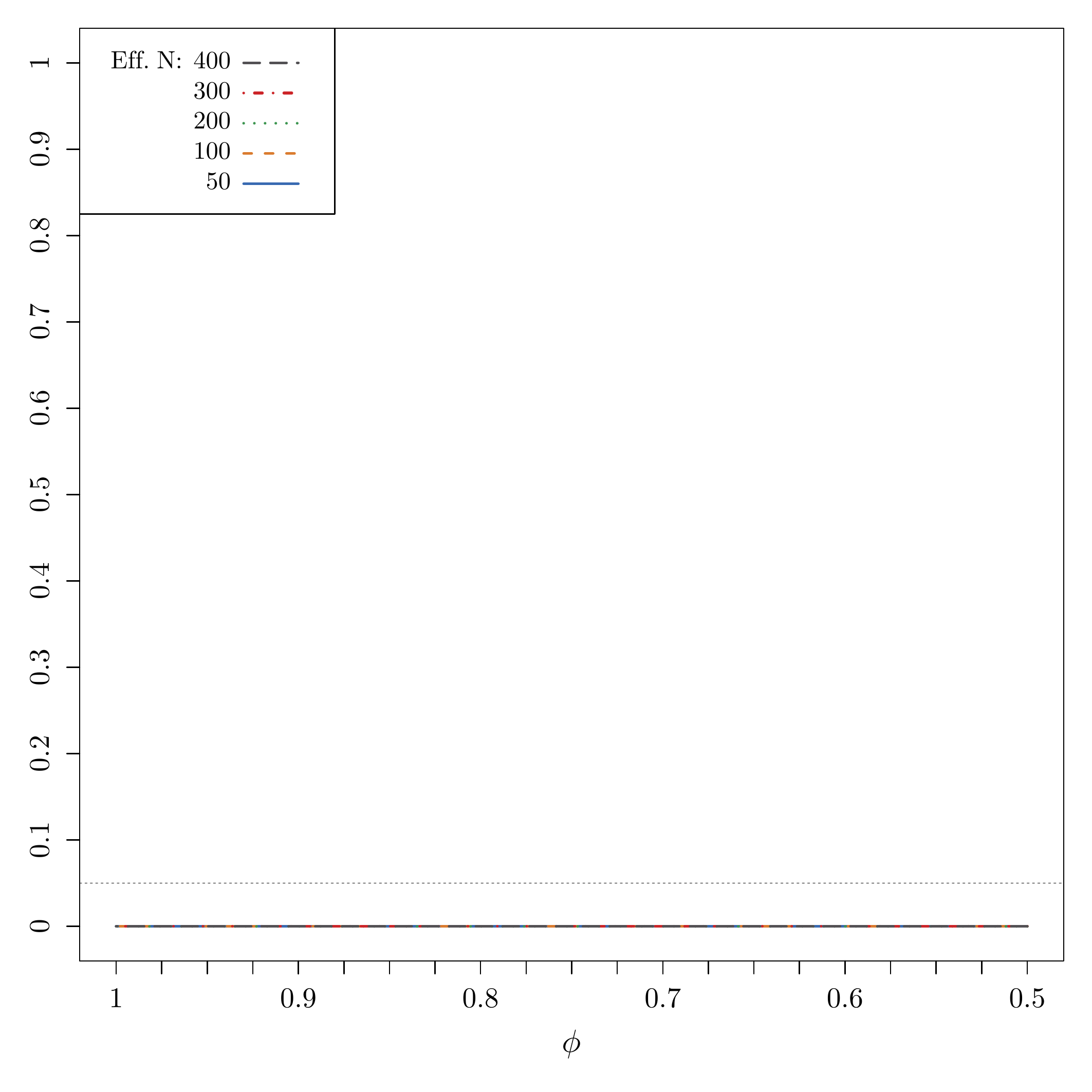}}
}
\subfloat[$\mathsf{H}_{0,2}:a_2\succ a_3\succ a_4\succ a_5\succ a_1$]{
\resizebox{0.4\columnwidth}{!}{\includegraphics{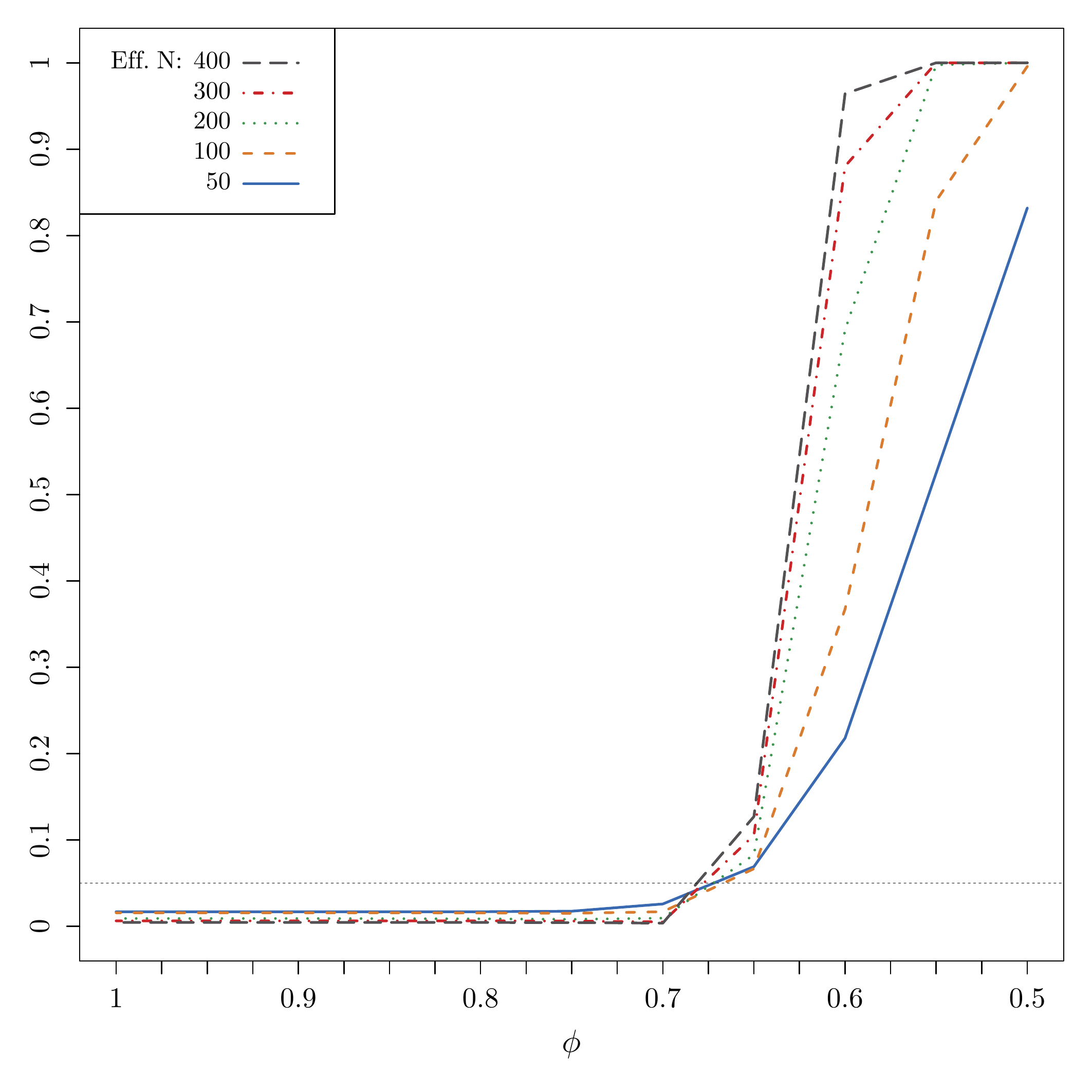}}
}\\
\subfloat[$\mathsf{H}_{0,3}:a_3\succ a_4\succ a_5\succ a_2\succ a_1$]{
\resizebox{0.4\columnwidth}{!}{\includegraphics{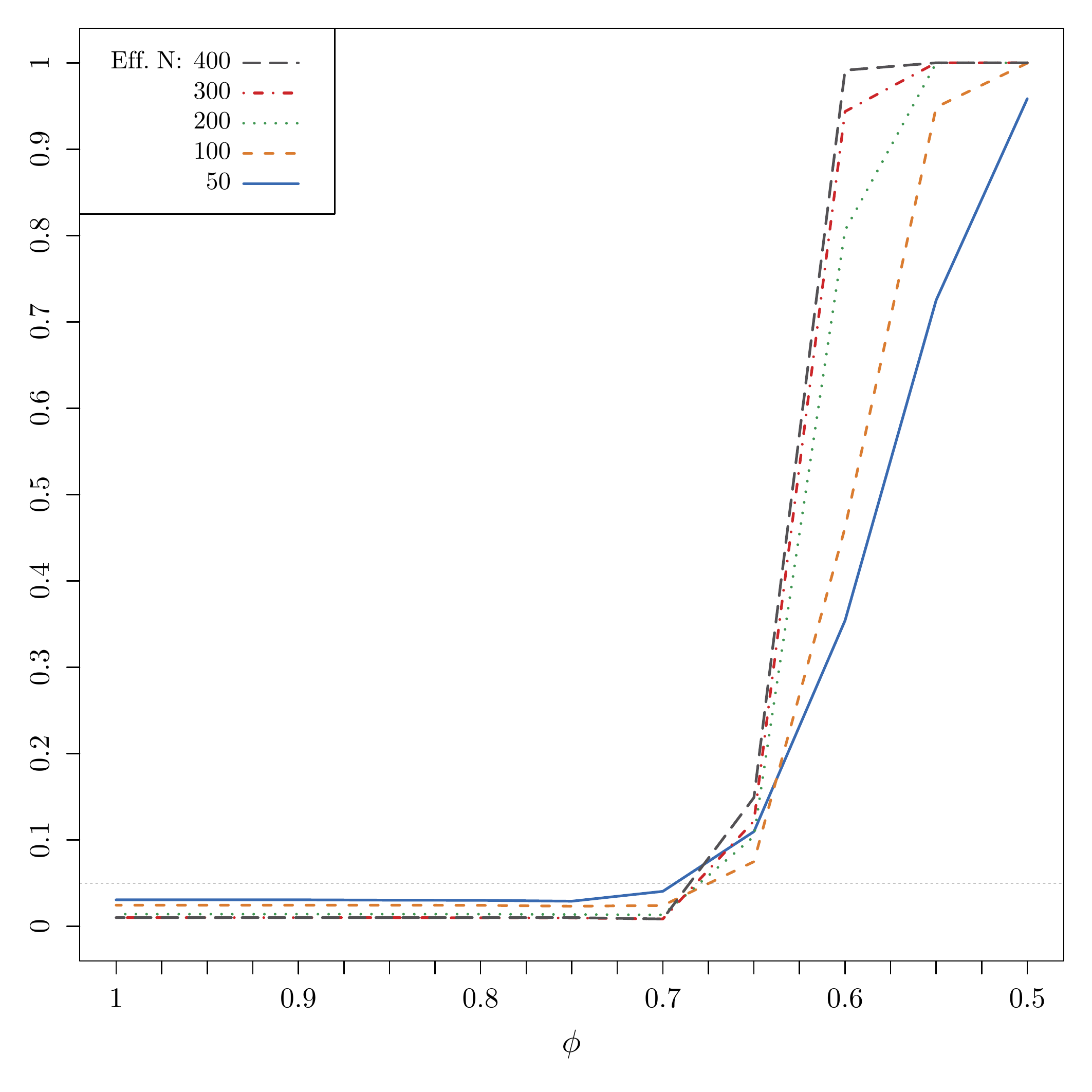}}
}
\subfloat[$\mathsf{H}_{0,4}:a_4\succ a_5\succ a_3\succ a_2\succ a_1$]{
\resizebox{0.4\columnwidth}{!}{\includegraphics{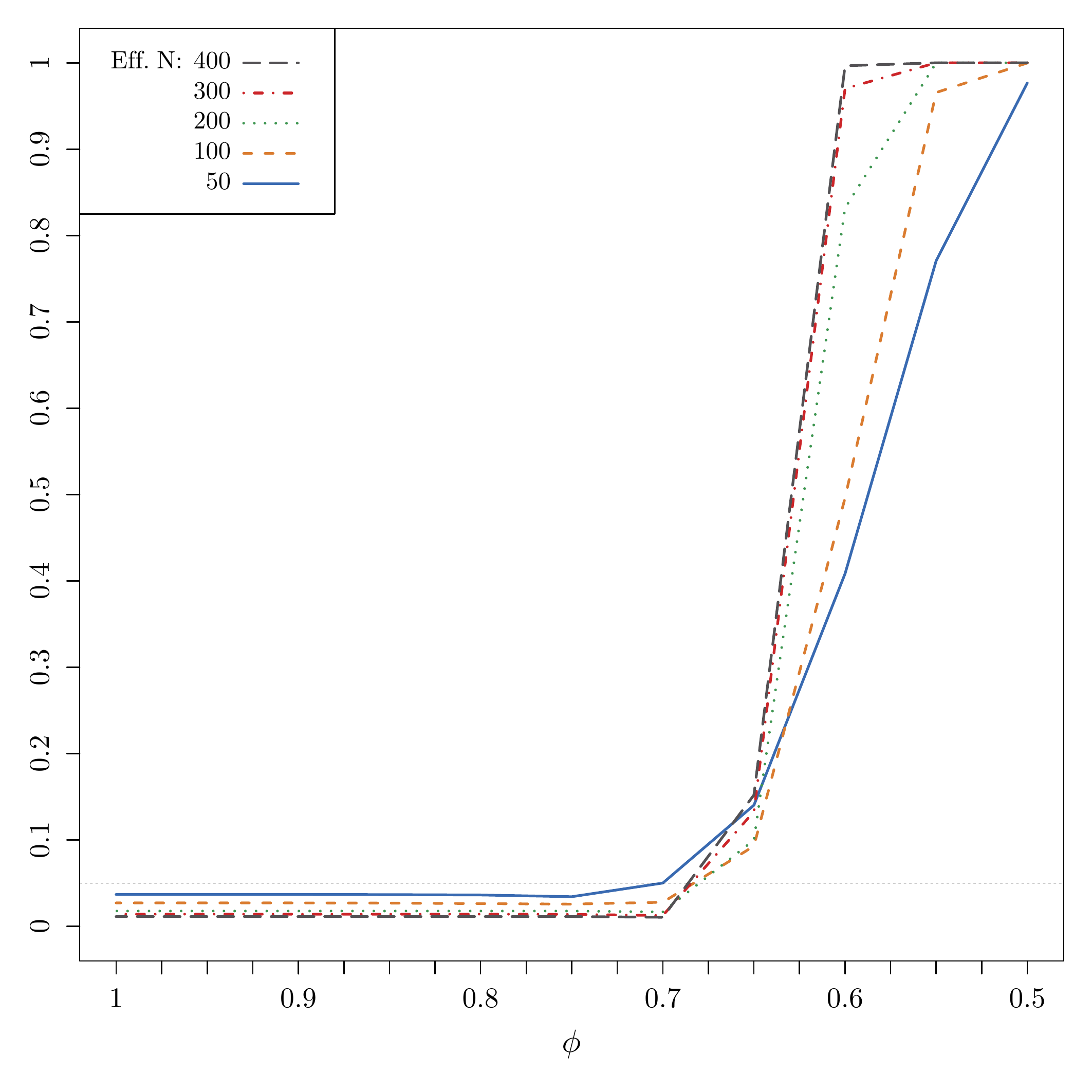}}
}\\
\subfloat[$\mathsf{H}_{0,5}:a_5\succ a_4\succ a_3\succ a_2\succ a_1$]{
\resizebox{0.4\columnwidth}{!}{\includegraphics{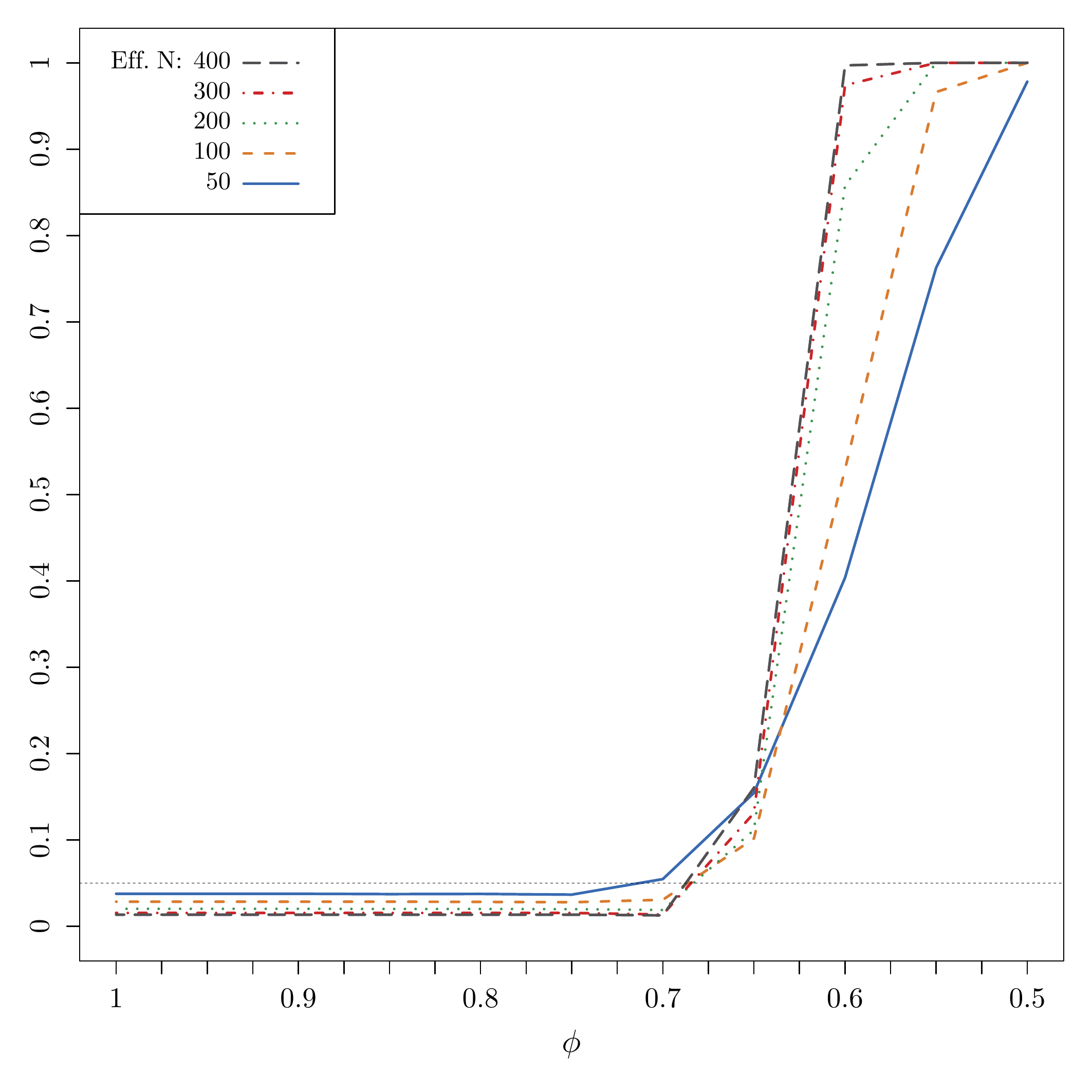}}
}\ \ 
\begin{minipage}{0.4\columnwidth}\footnotesize
\vskip-0.17\textheight Shown in the figure are empirical rejection probabilities testing the five null hypothesis through 5,000 simulations, with nominal size $0.05$. Logit attention rule with $\varsigma=0$ is used, as described in the text. For each simulation repetition, five effective sample sizes are considered $50$, $100$, $200$, $300$ and $400$. 
\end{minipage}
\caption{Empirical Rejection Probabilities ($\varsigma=0$)}\label{fig:figSA-1}
\end{figure}

\begin{figure}[!ht]
\centering
\subfloat[$\mathsf{H}_{0,1}:a_1\succ a_2\succ a_3\succ a_4\succ a_5$]{
\resizebox{0.4\columnwidth}{!}{\includegraphics{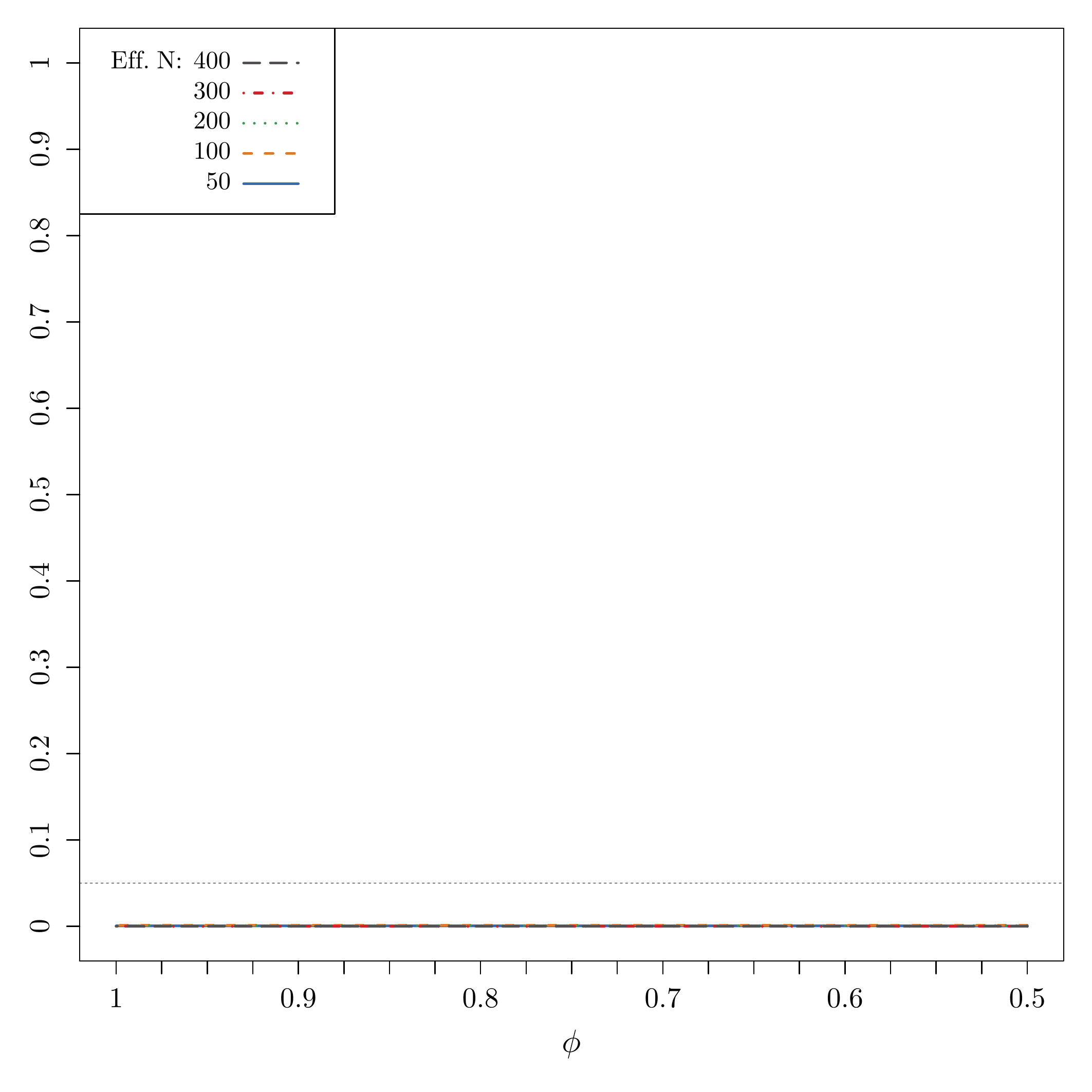}}
}
\subfloat[$\mathsf{H}_{0,2}:a_2\succ a_3\succ a_4\succ a_5\succ a_1$]{
\resizebox{0.4\columnwidth}{!}{\includegraphics{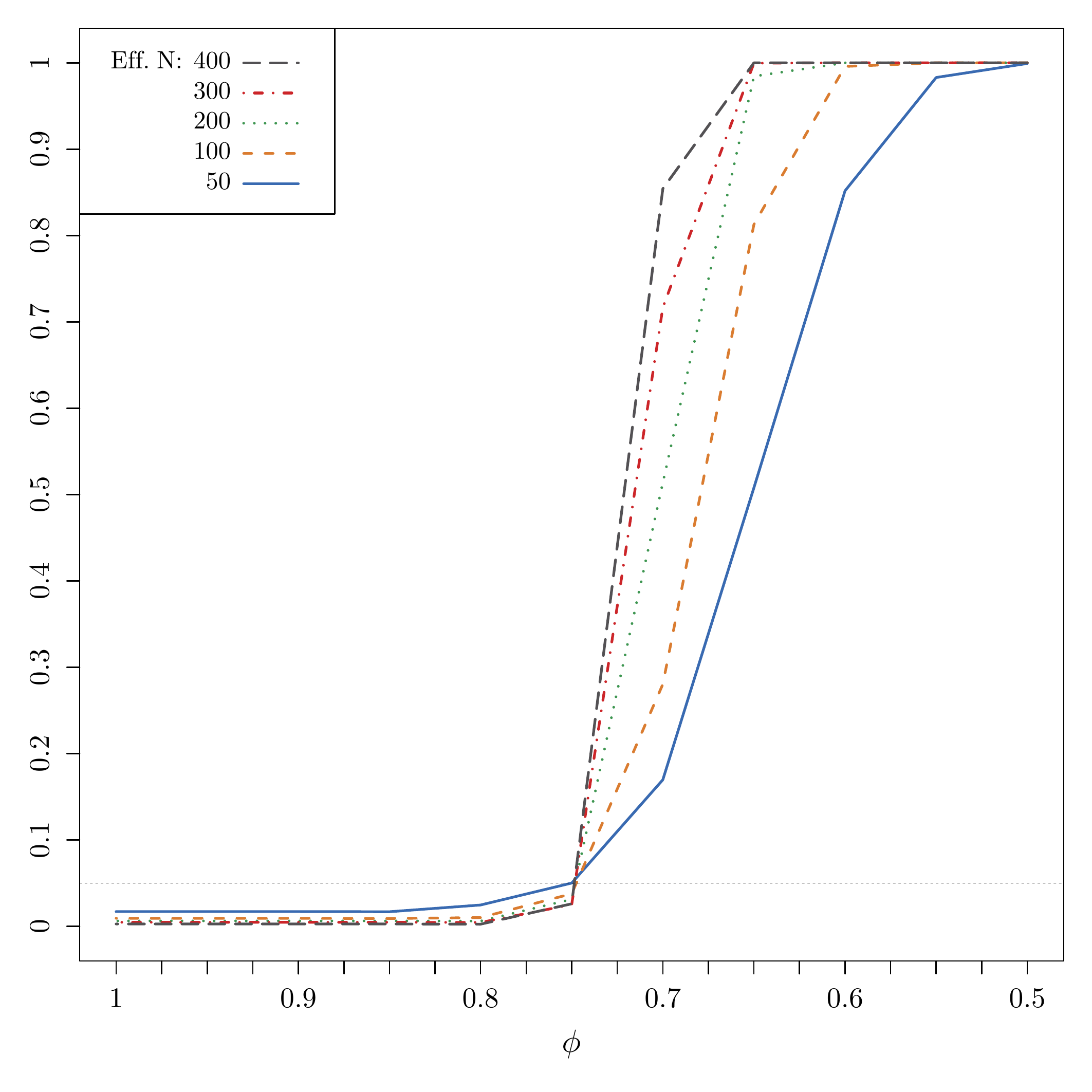}}
}\\
\subfloat[$\mathsf{H}_{0,3}:a_3\succ a_4\succ a_5\succ a_2\succ a_1$]{
\resizebox{0.4\columnwidth}{!}{\includegraphics{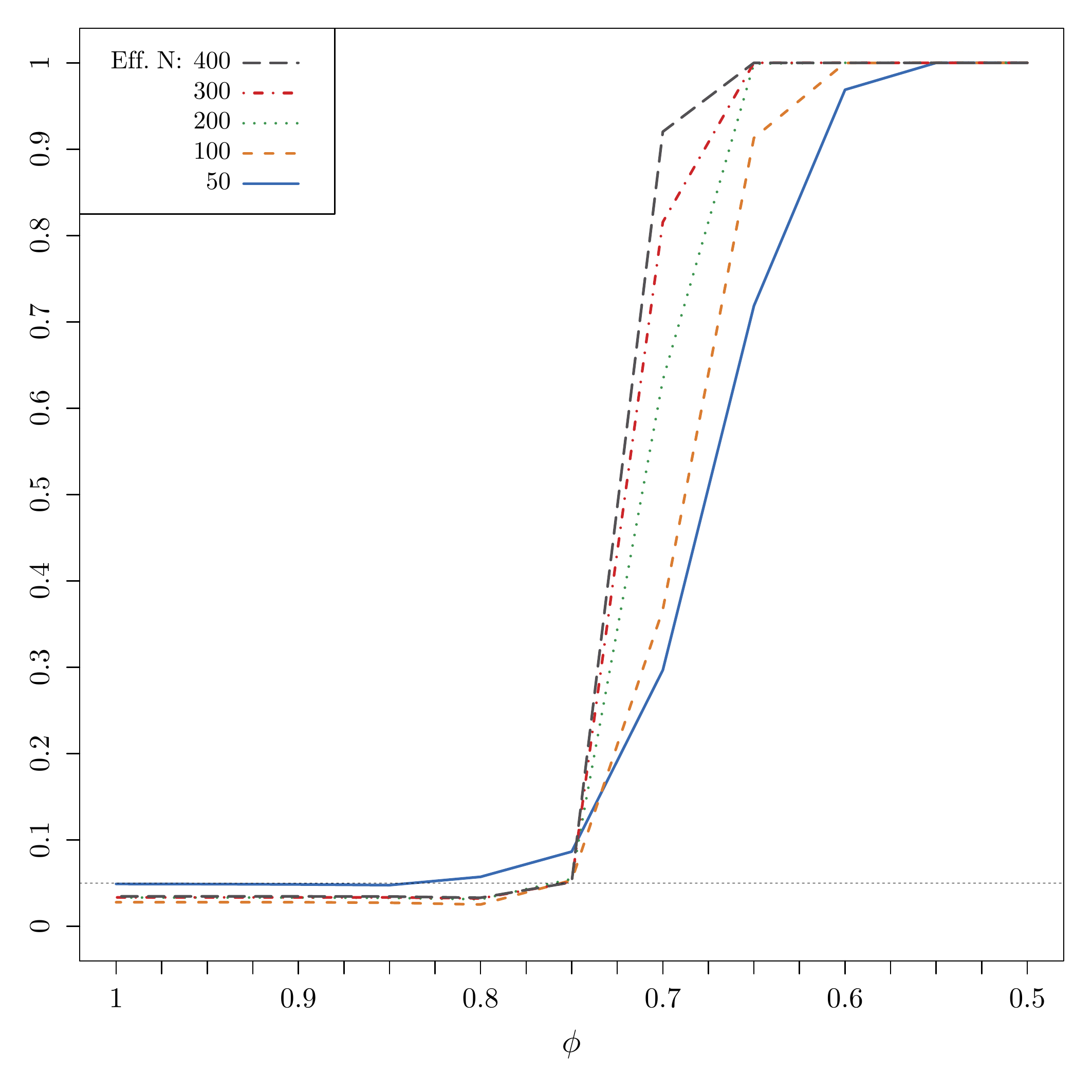}}
}
\subfloat[$\mathsf{H}_{0,4}:a_4\succ a_5\succ a_3\succ a_2\succ a_1$]{
\resizebox{0.4\columnwidth}{!}{\includegraphics{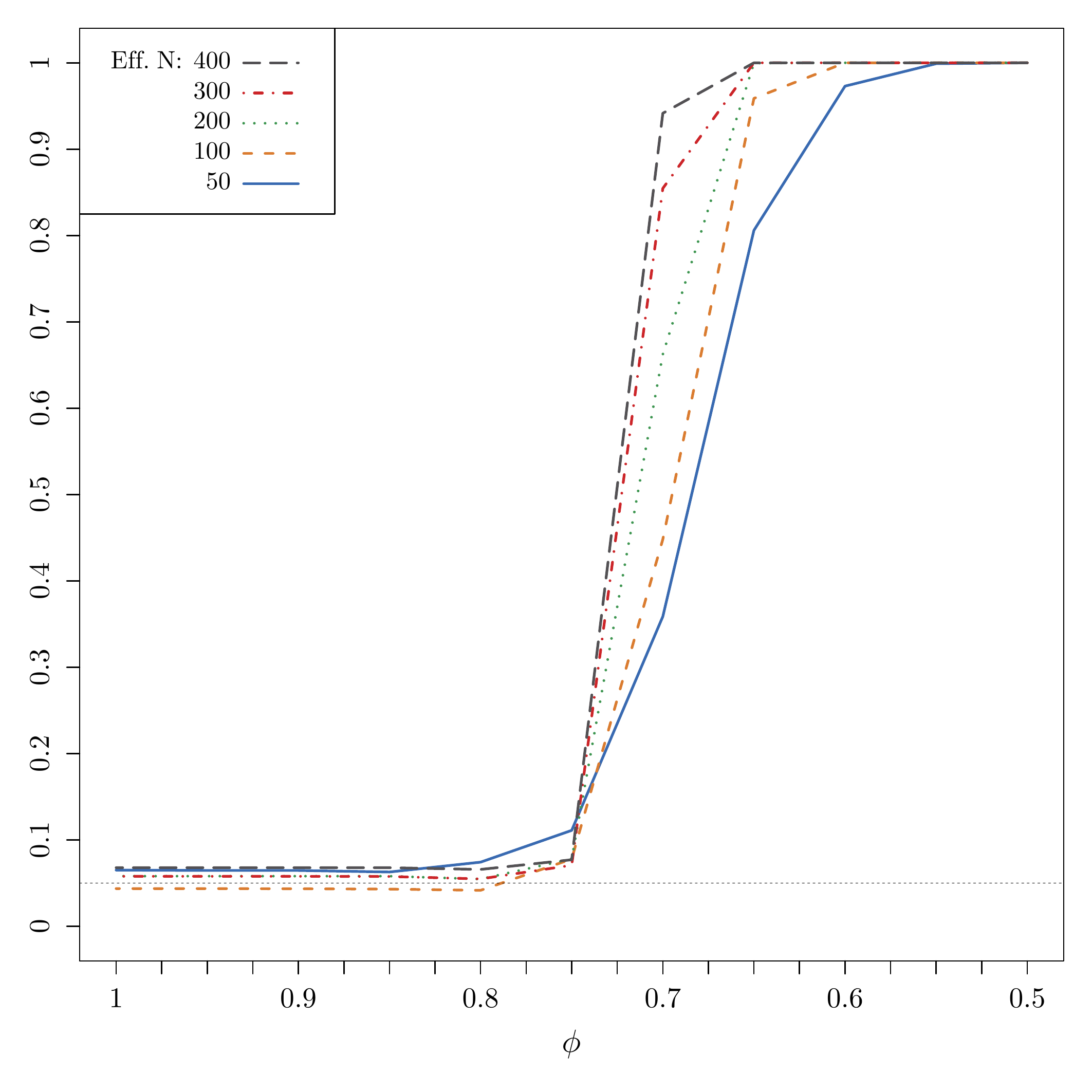}}
}\\
\subfloat[$\mathsf{H}_{0,5}:a_5\succ a_4\succ a_3\succ a_2\succ a_1$]{
\resizebox{0.4\columnwidth}{!}{\includegraphics{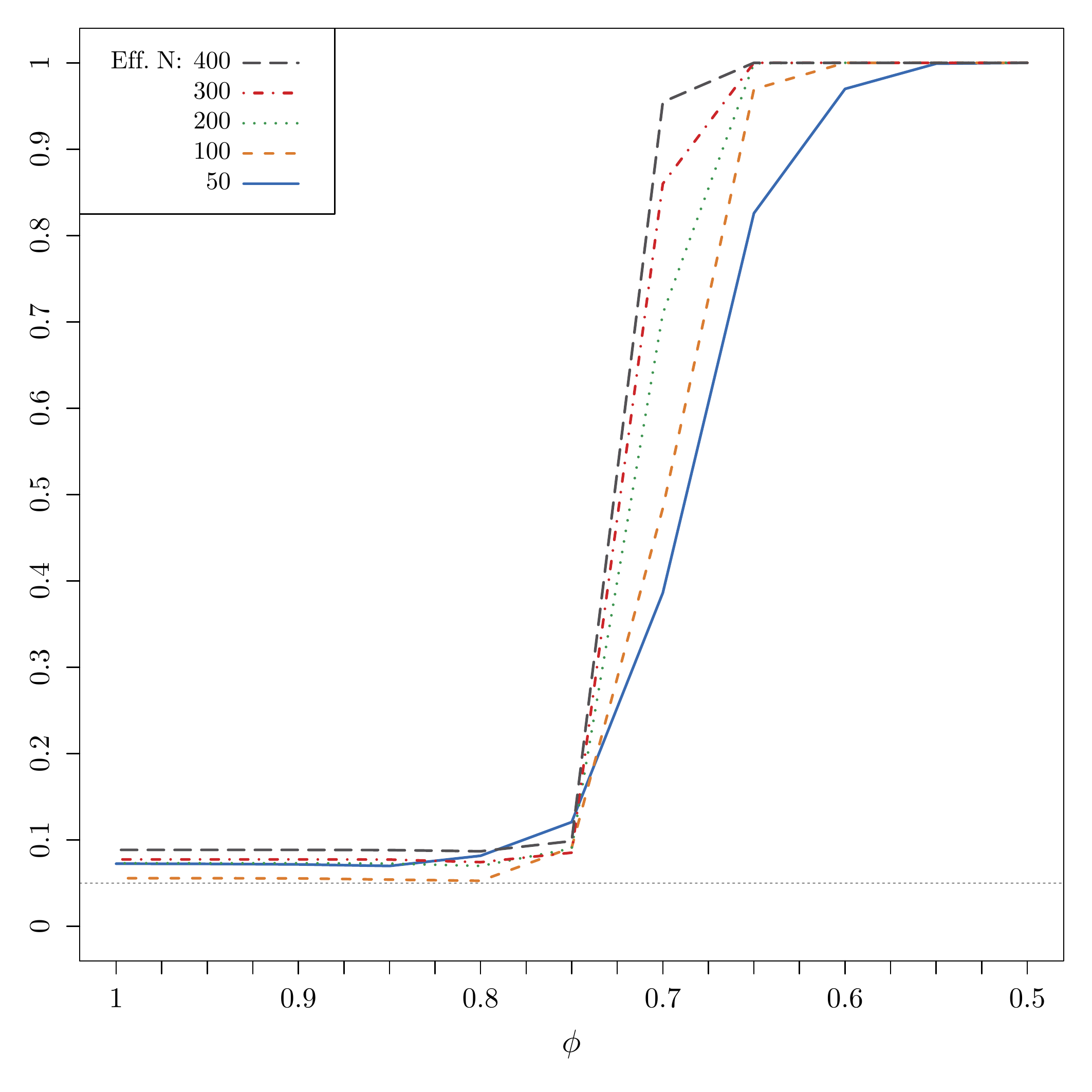}}
}\ \ 
\begin{minipage}{0.4\columnwidth}\footnotesize
\vskip-0.17\textheight Shown in the figure are empirical rejection probabilities testing the five null hypothesis through 5,000 simulations, with nominal size $0.05$. Logit attention rule with $\varsigma=1$ is used, as described in the text. For each simulation repetition, five effective sample sizes are considered $50$, $100$, $200$, $300$ and $400$. 
\end{minipage}
\caption{Empirical Rejection Probabilities ($\varsigma=1$)}\label{fig:figSA-2}
\end{figure}

\singlespacing
\bibliographystyle{jpe}
\bibliography{Cattaneo-Ma-Masatlioglu-Suleymanov_2019_RAM}